\documentclass[acmsmall,10pt]{acmart}
\pdfoutput=1
\usepackage{amsmath}
\usepackage{amssymb}
\usepackage{amsthm}
\usepackage{mathrsfs}
\usepackage{graphicx}
\usepackage{float}
\usepackage{caption}
\usepackage{subcaption,wrapfig}
\usepackage{tikz}
\usetikzlibrary{arrows,automata,positioning,backgrounds,calc}

\usepackage{pifont}
\usepackage[boxed,linesnumbered,noend]{algorithm2e}

\newcommand{\eg}{e.g., }
\newcommand{\ie}{i.e., }

\newcommand{\SO}{\textcolor{red}{\textsf{so}}}

\newcommand{\VIS}{\textcolor{red}{\textsf{vis}}}
\newcommand{\CO}{\textcolor{red}{\textsf{co} }}

\newcommand{\axpre}{\mathsf{Prefix}}
\newcommand{\axconf}{\mathsf{Conflict}}

\newcommand{\axser}{\mathsf{Serializability}}

\newcommand{\set}[1]{{\{{#1}\}}}

\newcommand{\tup}[1]{{\left\langle{#1}\right\rangle}}
\renewcommand{\implies}{\Rightarrow}

\newcommand{\Var}{\mathsf{Var}}
\newcommand{\Val}{\mathsf{Val}}
\newcommand{\OId}{\mathsf{OpId}}
\newcommand{\Op}{\mathsf{Op}}
\newcommand{\xvar}{{x}}
\newcommand{\yvar}{{y}}
\newcommand{\val}{{v}}
\newcommand{\id}{{i}}
\newcommand{\rd}[3][]{\textsf{read}_{#1}({#2},{#3})}
\newcommand{\wrt}[3][]{\textsf{write}_{#1}({#2},{#3})}

\newcommand{\tr}{{t}}

\newcommand{\hist}{{h}}
\newcommand{\po}{\textcolor{red}{\mathsf{po}}}
\newcommand{\so}{\textcolor{red}{\mathsf{so}}}
\newcommand{\co}{\textcolor{red}{\mathsf{co}}}
\newcommand{\wro}[1][]{\textcolor{red}{\mathsf{wr}_{#1}}}
\newcommand{\rwo}{\mathsf{RW}}
\newcommand{\wrosi}{\mathsf{WR}}
\newcommand{\readOp}[1]{\mathsf{reads}({#1})}
\newcommand{\vars}[1]{\mathsf{vars}({#1})}
\newcommand{\writeOp}[1]{\mathsf{writes}({#1})}
\newcommand{\writeVar}[2]{{#1}\ \mathsf{writes}\ {#2}}

\renewcommand{\hist}{{h}}

\newcommand{\vis}{\textcolor{red}{\mathsf{vis}}}

\acmConference[PL'18]{ACM SIGPLAN Conference on Programming Languages}{January 01--03, 2018}{New York, NY, USA}
\acmYear{2018}
\acmISBN{} % \acmISBN{978-x-xxxx-xxxx-x/YY/MM}
\acmDOI{} % \acmDOI{10.1145/nnnnnnn.nnnnnnn}
\startPage{1}

\setcopyright{none}
\usepackage{booktabs}   %% For formal tables:
\usepackage{subcaption} %% For complex figures with subfigures/subcaptions
\renewcommand{\circ}{\mathop{;}}

\begin{document}

%% Title information
\title{On the Complexity of Checking Transactional Consistency}         %% [Short Title] is optional;
                                        %% when present, will be used in
                                        %% header instead of Full Title.
%\titlenote{with title note}             %% \titlenote is optional;
                                        %% can be repeated if necessary;
                                        %% contents suppressed with 'anonymous'
%\subtitle{}                     %% \subtitle is optional
%\subtitlenote{with subtitle note}       %% \subtitlenote is optional;
                                        %% can be repeated if necessary;
                                        %% contents suppressed with 'anonymous'

%% Author information
%% Contents and number of authors suppressed with 'anonymous'.
%% Each author should be introduced by \author, followed by
%% \authornote (optional), \orcid (optional), \affiliation, and
%% \email.
%% An author may have multiple affiliations and/or emails; repeat the
%% appropriate command.
%% Many elements are not rendered, but should be provided for metadata
%% extraction tools.

%% Author with single affiliation.
\author{Ranadeep Biswas}
\affiliation{
  \position{PhD Student}
  \institution{Universite de Paris, IRIF, CNRS}            %% \institution is required
  \city{Paris}
  \postcode{F-75013}
  \country{France}                    %% \country is recommended
}
\email{ranadeep@irif.fr}          %% \email is recommended

%% Author with two affiliations and emails.
\author{Constantin Enea}
\affiliation{
  \position{Associate Professor}
  \institution{Universite de Paris, IRIF, CNRS}            %% \institution is required
  \city{Paris}
  \postcode{F-75013}
  \country{France}                    %% \country is recommended
}
\email{cenea@irif.fr}

%% Abstract
%% Note: \begin{abstract}...\end{abstract} environment must come
%% before \maketitle command
\begin{abstract}
Transactions simplify concurrent programming by enabling computations on shared data that are isolated from other concurrent computations and are resilient to failures. Modern databases provide different consistency models for transactions corresponding to different tradeoffs between consistency and availability. In this work, we investigate the problem of checking whether a given execution of a transactional database adheres to some consistency model. We show that consistency models like read committed, read atomic, and causal consistency are polynomial time checkable while prefix consistency and snapshot isolation are NP-complete in general. These results complement a previous NP-completeness result concerning serializability. Moreover, in the context of NP-complete consistency models, we devise algorithms that are polynomial time assuming that certain parameters in the input executions, e.g., the number of sessions, are fixed. We evaluate the scalability of these algorithms in the context of several production databases.
\end{abstract}

%% 2012 ACM Computing Classification System (CSS) concepts
%% Generate at 'http://dl.acm.org/ccs/ccs.cfm'.
\begin{CCSXML}
<ccs2012>
<concept>
<concept_id>10011007.10011006.10011008</concept_id>
<concept_desc>Software and its engineering~General programming languages</concept_desc>
<concept_significance>500</concept_significance>
</concept>
<concept>
<concept_id>10003456.10003457.10003521.10003525</concept_id>
<concept_desc>Social and professional topics~History of programming languages</concept_desc>
<concept_significance>300</concept_significance>
</concept>
</ccs2012>
\end{CCSXML}

\ccsdesc[500]{Software and its engineering~General programming languages}
\ccsdesc[300]{Social and professional topics~History of programming languages}
%% End of generated code

%% Keywords
%% comma separated list
%\keywords{keyword1, keyword2, keyword3}  %% \keywords are mandatory in final camera-ready submission

%% \maketitle
%% Note: \maketitle command must come after title commands, author
%% commands, abstract environment, Computing Classification System
%% environment and commands, and keywords command.
\maketitle

%!TEX root = draft.tex
\section{Introduction}

Transactions simplify concurrent programming by enabling computations on shared data that are isolated from other concurrent computations and resilient to failures. Modern databases provide transactions in various forms corresponding to different tradeoffs between consistency and availability. The strongest level of consistency is achieved with \emph{serializable} transactions~\cite{DBLP:journals/jacm/Papadimitriou79b} whose outcome in concurrent executions is the same as if the transactions were executed atomically in some order. Unfortunately, serializability carries a significant penalty on the availability of the system assuming, for instance, that the database is accessed over a network that can suffer from partitions or failures. For this reason, modern databases often provide weaker guarantees about transactions, formalized by weak consistency models, e.g., causal consistency~\cite{DBLP:journals/cacm/Lamport78} and snapshot isolation~\cite{DBLP:conf/sigmod/BerensonBGMOO95}.

Implementations of large-scale databases providing transactions are difficult to build and test. For instance, distributed (replicated) databases must account for partial failures, where some components or the network can fail and produce incomplete results. Ensuring fault-tolerance relies on intricate protocols that are difficult to design and reason about. The black-box testing framework Jepsen~\cite{jepsen} found a remarkably large number of subtle problems in many production distributed databases. %~\footnote{https://www.infoq.com/presentations/partitioning-comparison}.

Testing a transactional database raises two issues: (1) deriving a suitable set of testing scenarios, e.g., faults to inject into the system and the set of transactions to be executed, and (2) deriving efficient algorithms for checking whether a given execution satisfies the considered consistency model. The Jepsen framework aims to address the first issue by using randomization, 
%shows that the first issue can be solved using randomization, 
e.g., introducing faults at random and choosing the operations in a transaction randomly. The effectiveness of this approach has been proved formally in recent work~\cite{DBLP:journals/pacmpl/OzkanMNBW18}. The second issue is, however, largely unexplored. Jepsen checks consistency in a rather ad-hoc way, focusing on specific classes of violations to a given consistency model, e.g., dirty reads (reading values from aborted transactions). This problem is challenging because the consistency specifications are non-trivial and they cannot be checked using, for instance, standard local assertions added to the client's code. 

Besides serializability, the complexity of checking correctness of an execution w.r.t. some consistency model is unknown. Checking serializability has been shown to be NP-complete~\cite{DBLP:journals/jacm/Papadimitriou79b}, and checking causal consistency in a \emph{non-transactional} context is known to be polynomial time~\cite{DBLP:conf/popl/BouajjaniEGH17}. In this work, we try to fill this gap by investigating the complexity of this problem w.r.t. several consistency models and, in case of NP-completeness, devising algorithms that are polynomial time assuming fixed bounds for certain parameters of the input executions, e.g., the number of sessions. 

%
%The only result that explores the complexity of this problem 
%
%Except for  serializability, in which case it has been shown that checking  be NP-complete~\cite{DBLP:journals/jacm/Papadimitriou79b}
%% testing, i.e., randomly choosing the faults injected into the system and the transactions to be executed is enough to reveal a
%
%
%The success of Jepsen relies on random transactions as well as randomly introduced partition faults, therefore it is solved. We tackle the second issue for a series of consistency models (Jepsen implements a test of linearizability https://github.com/jepsen-io/knossos and an ad-hoc test for causal consistency restricted to bounded executions, \url{https://github.com/jepsen-io/jepsen/blob/f345226dba1266bc37487d734a02caddf7d1d125/jepsen/src/jepsen/tests/causal.clj})
We consider several consistency models that are the most prevalent in practice. The weakest of them, \emph{Read Committed} (RC)~\cite{DBLP:conf/sigmod/BerensonBGMOO95}, requires that every value read in a transaction is written by a committed transaction. \emph{Read Atomic} (RA)~\cite{DBLP:conf/concur/Cerone0G15} requires that successive reads of the same variable in a transaction return the same value (also known as Repeatable Reads~\cite{DBLP:conf/sigmod/BerensonBGMOO95}), and that a transaction ``sees'' the values written by previous transactions in the same session. In general, we assume that transactions are organized in \emph{sessions}~\cite{DBLP:conf/pdis/TerryDPSTW94}, an abstraction of the sequence of transactions performed during the execution of an application.
\emph{Causal Consistency} (CC)~\cite{DBLP:journals/cacm/Lamport78} requires that if a transaction $\tr_1$ ``affects'' another transaction $\tr_2$, e.g., $\tr_1$ is ordered before $\tr_2$ in the same session or $\tr_2$ reads a value written by $\tr_1$, then these two transactions are observed by any other transaction in this order. \emph{Prefix Consistency} (PC)~\cite{DBLP:conf/ecoop/BurckhardtLPF15} requires that there exists a total commit order between all the transactions such that each transaction observes a prefix of this sequence. \emph{Snapshot Isolation} (SI)~\cite{DBLP:conf/sigmod/BerensonBGMOO95} further requires that two different transactions observe different prefixes if they both write to a common variable.
%Two different transactions $\tr_1$ and $\tr_2$ may observe the same prefix, which is not allowed under \emph{Snapshot Isolation} (SI)~\cite{DBLP:conf/sigmod/BerensonBGMOO95} when these two transactions write on a common variable. 
Finally, we also provide new results concerning the problem of checking serializability (SER) that complement the known result about its NP-completeness. 

The algorithmic issues we explore in this paper have led to a new specification framework for these consistency models that relies on the fact that the \emph{write-read} relation in an execution (also known as \emph{read-from}), relating reads with the transactions that wrote their value, can be defined effectively. The write-read relation can be extracted easily from executions where each value is written at most once (a variable can be written an arbitrary number of times). This can be easily enforced by tagging values with unique identifiers (e.g., a local counter that is incremented with every new write coupled with a client/session identifier)\footnote{This is also used in Jepsen, e.g., checking dirty reads in Galera~\cite{jepsen-galera}.}. Since practical database implementations are data-independent~\cite{DBLP:conf/popl/Wolper86}, i.e., their behavior doesn't depend on the concrete values read or written in the transactions, any potential buggy behavior can be exposed in executions where each value is written at most once. Therefore, this assumption is without loss of generality.

Previous work~\cite{DBLP:conf/popl/BouajjaniEGH17,DBLP:conf/popl/BurckhardtGYZ14,DBLP:conf/concur/Cerone0G15} has formalized such consistency models using two auxiliary relations: a \emph{visibility} relation defining for each transaction the set of transactions it observes, and a \emph{commit order} defining the order in which transactions are committed to the ``global'' memory. An execution satisfying some consistency model is defined as the existence of a visibility relation and a commit order obeying certain axioms. In our case, the write-read relation derived from the execution plays the role of the visibility relation. This simplification allows us to state a series of axioms defining these consistency models, which have a common shape. Intuitively, they define lower bounds on the set of transactions $\tr_1$ that \emph{must} precede in commit order a transaction $\tr_2$ that is read in the execution. Besides shedding a new light on the differences between these consistency models, these axioms are essential for the algorithmic issues we investigate afterwards.

%Based on our formalization of these criteria, 
We establish that checking whether an execution satisfies RC, RA, or CC is polynomial time, while the same problem is NP-complete for PC and SI. Moreover, in the case of the NP-complete consistency models (PC, SI, SER), we show that their verification problem becomes polynomial time provided that, roughly speaking, the number of sessions in the input executions is considered to be fixed (i.e., not counted for in the input size). In more detail, we establish that checking SER reduces to a search problem in a space that has polynomial size when the number of sessions is fixed. (This algorithm applies to arbitrary executions, but its complexity would be exponential in the number of sessions in general.) Then, we show that checking PC or SI can be reduced in polynomial time to checking SER using a transformation of executions that, roughly speaking, splits each transaction in two parts: one part containing all the reads, and one part containing all the writes (SI further requires adding some additional variables in order to deal with transactions writing on a common variable).
We extend these results even further by relying on an abstraction of executions called \emph{communication graphs}~\cite{DBLP:journals/pacmpl/ChalupaCPSV18}. Roughly speaking, the vertices of a communication graph correspond to sessions, and the edges represent the fact that two sessions access (read or write) the same variable. We show that all these criteria are polynomial-time checkable provided that the \emph{biconnected} components of the communication graph are of fixed size.

We provide an experimental evaluation of our algorithms on executions of CockroachDB~\cite{cockroach}, which claims to implement serializability~\cite{cockroach-claim} acknowledging however the possibility of anomalies, Galera~\cite{galera}, whose documentation contains contradicting claims about whether it implements snapshot isolation~\cite{galera-claim,galera-notclaim}, and AntidoteDB~\cite{antidote}, which claims to implement causal consistency~\cite{antidote-claim}.
%Galera~\cite{galera}, and AntidoteDB~\cite{antidote}, which claim to implement serializability~\cite{cockroach-claim}, snapshot isolation~\cite{galera-claim} and causal consistency~\cite{antidote-claim}, respectively (in the default configuration). 
Our implementation reports violations of these criteria in all cases. 
%In the case of CockroachDB, the documentation admits possible anomalies while in the case of Galera we confirm an open issue submitted on Github~\cite{galera-issue}. 
The consistency violations we found for AntidoteDB are novel and have been confirmed by its developers. We show that our algorithms are efficient and scalable.
%and they outperform an encoding to boolean satisfiability of the consistency models. 
In particular, we show that, although the asymptotic complexity of our algorithms is exponential in general (w.r.t. the number of sessions), the worst-case behavior is not exercised in practice.

To summarize, the contributions of this work are fourfold:
\begin{itemize}

  \item We develop a new specification framework for describing common transactional-consistency criteria (§\ref{sec:def});

  \item We show that checking RC, RA, and CC is polynomial time while checking PC and SI is NP-complete (§\ref{sec:general});

  \item We show that PC, SI, and SER are polynomial-time checkable assuming that the communication graph of the input execution has fixed-size biconnected components (§\ref{sec:bounded_width} and §\ref{sec:communication});
  
  \item We perform an empirical evaluation of our algorithms on executions generated by production databases (§\ref{sec:exp});

\end{itemize}
Combined, these contributions form an effective algorithmic framework for the verification of transactional-consistency models. To the best of our knowledge, we are the first to investigate the asymptotic complexity for most of these consistency models, despite their prevalence in practice.

\section{Consistency Criteria}\label{sec:def}

\subsection{Histories}

We consider a transactional database storing a set of variables $\Var=\{\xvar,\yvar,\ldots\}$. Clients interact with the database by issuing transactions formed of $\textsf{read}$ and $\textsf{write}$ operations. Assuming an unspecified set of values $\Val$ and a set of operation identifiers $\OId$, we let 
\begin{align*}
 \Op=\set{\rd[\id]{\xvar}{\val},\wrt[\id]{\xvar}{\val}: \id\in\OId, \xvar\in\Var, \val\in \Val}
\end{align*} 
be the set of operations reading a value $\val$ or writing a value $\val$ to a variable $\xvar$. We omit operation identifiers when they are not important.

\begin{definition}
 A \emph{transaction} $\tup{O, \po}$ is a finite set of operations $O$ along with a strict total order $\po$ on $O$, called \emph{program order}.
\end{definition}

We use $\tr$, $\tr_1$, $\tr_2$, $\ldots$ to range over transactions. The set of read, resp., write, operations in a transaction $\tr$ is denoted by $\readOp{\tr}$, resp., $\writeOp{\tr}$. The extension to sets of transactions is defined as usual. Also, we say that a transaction $\tr$ \emph{writes} a variable $\xvar$, denoted by $\writeVar{\tr}{\xvar}$, when $\wrt[\id]{\xvar}{\val}\in \writeOp{\tr}$ for some $\id$ and $\val$. Similarly, a transaction $\tr$ \emph{reads} a variable $\xvar$ when $\rd[\id]{\xvar}{\val}\in \readOp{\tr}$ for some $\id$ and $\val$.

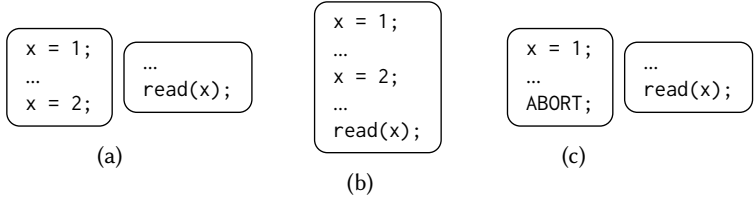
\begin{figure}[t]
 \centering
 \begin{subfigure}{.21\textwidth}
   \centering
  \resizebox{!}{1.3cm}{
   \begin{tikzpicture}[->,>=stealth',shorten >=1pt,auto,node distance=2cm,
     semithick, transform shape]
    \node[draw, rounded corners=2mm] (t1) {\begin{tabular}{l} \texttt{x = 1;} \\ ... \\ \texttt{x = 2;} \end{tabular}};
    \node[draw, rounded corners=2mm, right of = t1] (t2) {\begin{tabular}{l} ... \\ \texttt{read(x);} \end{tabular}};
   \end{tikzpicture}  
  }
%  \caption{Only the lastest write is visible to other transaction}
\caption{}
  \label{rc_eg:1}
 \end{subfigure}
 \hspace{10mm}
 \begin{subfigure}{.1\textwidth}
   \centering
  \resizebox{!}{2cm}{
   \begin{tikzpicture}[->,>=stealth',shorten >=1pt,auto,node distance=2cm,
     semithick, transform shape]
    \node[draw, rounded corners=2mm] (t1) {\begin{tabular}{l} \texttt{x = 1;} \\ ... \\ \texttt{x = 2;} \\ ... \\ \texttt{read(x);}\end{tabular}};
   \end{tikzpicture}  
  }
%  \caption{Always reads the latest write inside a transaction}
\caption{}
  \label{rr_eg:1}
 \end{subfigure}
 \hspace{10mm}
 \begin{subfigure}{.14\textwidth}
   \centering
  \resizebox{!}{1.3cm}{
   \begin{tikzpicture}[->,>=stealth',shorten >=1pt,auto,node distance=2cm,
     semithick, transform shape]
    \node[draw, rounded corners=2mm] (t1) {\begin{tabular}{l} \texttt{x = 1;} \\ ... \\ \texttt{ABORT;} \end{tabular}};
    \node[draw, rounded corners=2mm, right of = t1] (t2) {\begin{tabular}{l} ... \\ \texttt{read(x);} \end{tabular}};
   \end{tikzpicture}  
  }
%  \caption{Aborted transactions are not visible}
\caption{}
  \label{abort:1}
 \end{subfigure}
 \caption{Examples of transactions used to justify our simplifying assumptions (each box represents a different transaction): (a) only the last written value is observable in other transactions, (b) reads following writes to the same variable return the last written value in the same transaction, and (c) values written in aborted transactions are not observable.}
 \label{read_latest}
 \vspace{-3mm}
\end{figure}

To simplify the exposition, we assume that each transaction $\tr$ contains at most one write operation to each variable $\xvar$~\footnote{That is, for every transaction $\tr$, and every $\wrt{\xvar}{\val}, \wrt{\yvar}{\val'}\in \writeOp{\tr}$, we have that $\xvar\neq\yvar$.}, and that a read of a variable $\xvar$ cannot be preceded by a write to $\xvar$ in the same transaction\footnote{That is, for every transaction $\tr=\tup{O, \po}$, if $\wrt{\xvar}{\val}\in \writeOp{\tr}$ and there exists $\rd{\xvar}{\val}\in \readOp{\tr}$, then we have that $\tup{\rd{\xvar}{\val},\wrt{\xvar}{\val}}\in \po$}. If a transaction would contain multiple writes to the same variable, then only the last one should be visible to other transactions (w.r.t. any consistency criterion considered in practice). For instance, the \texttt{read(x)} in Figure~\ref{rc_eg:1} should not return 1 because this is not the last value written to {\tt x} by the other transaction. It can return the initial value or 2.
%In figure (\ref{rc_eg:1}), however the two transactions are executed, the operation \texttt{print(x)} in below transaction should not print \texttt{0}. Because \texttt{x=0} is not the latest write in the above transaction.
Also, if a read would be preceded by a write to the same variable in the same transaction, then it should return a value written in the same transaction (i.e., the value written by the latest write to $\xvar$ in that transaction). 
For instance, the \texttt{read(x)} in Figure~\ref{rr_eg:1} can only return 2 (assuming that there are no other writes on {\tt x} in the same transaction).
%In figure (\ref{rr_eg:1}), the operation \texttt{print(x)} in the transaction should not print \texttt{1}, because $\texttt{print(x)}$ is not the latest write to \texttt{print(x)}.
These two properties can be verified easily (in a syntactic manner) on a given execution. Beyond these two properties, the various consistency criteria used in practice constrain only the last writes to each variable in each transaction and the reads that are not preceded by writes to the same variable in the same transaction.

Consistency criteria are formalized on an abstract view of an execution called~\emph{history}. A history includes only successful or committed transactions. In the context of databases, it is always assumed that the effect of aborted transactions should not be visible to other transactions, and therefore, they can be ignored. For instance, the \texttt{read(x)} in Figure~\ref{abort:1} should not return the value 1 written by the aborted transaction. The transactions are ordered according to a (partial) \emph{session order} $\so$ which represents ordering constraints imposed by the applications using the database. Most often, $\so$ is a union of sequences, each sequence being called a \emph{session}. We assume that the history includes a \emph{write-read} relation that identifies the transaction writing the value returned by each read in the execution. As mentioned before, such a relation can be extracted easily from executions where each value is written at most once. Since in practice, databases are data-independent~\cite{DBLP:conf/popl/Wolper86}, i.e., their behavior does not depend on the concrete values read or written in the transactions, any potential buggy behavior can be exposed in such executions.

\begin{definition}
 A \emph{history} $\tup{T, \so, \wro}$ is a set of transactions $T$ along with a strict partial order $\so$ called \emph{session order}, and a 
 %surjective~\footnote{That is, for all $\rd{\xvar}{\val}\in \readOp{T}$ there exists a transaction $\tr\in T$ such that $\tup{\tr,\rd{\xvar}{\val}}\in \wro$.} 
 relation $\wro\subseteq T\times \readOp{T}$ called \emph{write-read} relation, s.t. 
 \begin{itemize}
  \item the inverse of $\wro$ is a total function, and if $(\tr,\rd{\xvar}{\val})\in\wro$, then $\wrt{\xvar}{\val}\in\tr$, and
  \item $\so\cup\wro$ is acyclic.
 \end{itemize}
\end{definition}

% TODO SAY THAT INITIAL VALUES ARE ASSUMED TO BE WRITTEN BY A TRANSACTION ORDERED IN SO BEFORE ALL THE OTHER TRANSACTIONS.

%The transactions may try to read a variable, even before there is a write to it. In practice, the databases return a default initialized or null value for those reads. This situation can be thought as if the database wrote that value in an \emph{initialization} transaction, in which all variables are written with an initialized value. So we assume a history contains an initialization transaction. This initialization transaction precedes all other transactions by $\so$.

To simplify the technical exposition, we assume that every history includes a distinguished transaction writing the initial values of all variables. This transaction precedes all the other transactions in $\so$. We use $\hist$, $\hist_1$, $\hist_2$, $\ldots$ to range over histories. 

We say that the read operation $\rd{\xvar}{\val}$ reads value $\val$ from variable $\xvar$ written by $\tr$ when $(\tr,\rd{\xvar}{\val})\in\wro$. 
For a given variable $\xvar$, $\wro[\xvar]$ denotes the restriction of $\wro$ to reads of variable $\xvar$, \ie, $\wro[\xvar]=\wro\cap (T\times \{\rd{\xvar}{\val}\mid \val\in \Val\})$. Moreover, we extend the relations $\wro$ and $\wro[\xvar]$ to pairs of transactions as follows: $\tup{\tr_1,\tr_2}\in \wro$, resp., $\tup{\tr_1,\tr_2}\in \wro[\xvar]$, iff there exists a read operation $\rd{\xvar}{\val}\in \readOp{\tr_2}$ such that $\tup{\tr_1,\rd{\xvar}{\val}}\in \wro$, resp., $\tup{\tr_1,\rd{\xvar}{\val}}\in \wro[\xvar]$. We say that the transaction $\tr_1$ is \emph{read} by the transaction $\tr_2$ when $\tup{\tr_1,\tr_2}\in \wro$, and that it is \emph{read} when it is read by some transaction $\tr_2$. 

\subsection{Axiomatic Framework}

% from Table. \ref{weakconsistency:1}. Table. \ref{weakconsistency:2} lists the axioms of consistency models.

%!TEX root = draft.tex
\tikzset{transaction state/.style={draw=black!0}}

 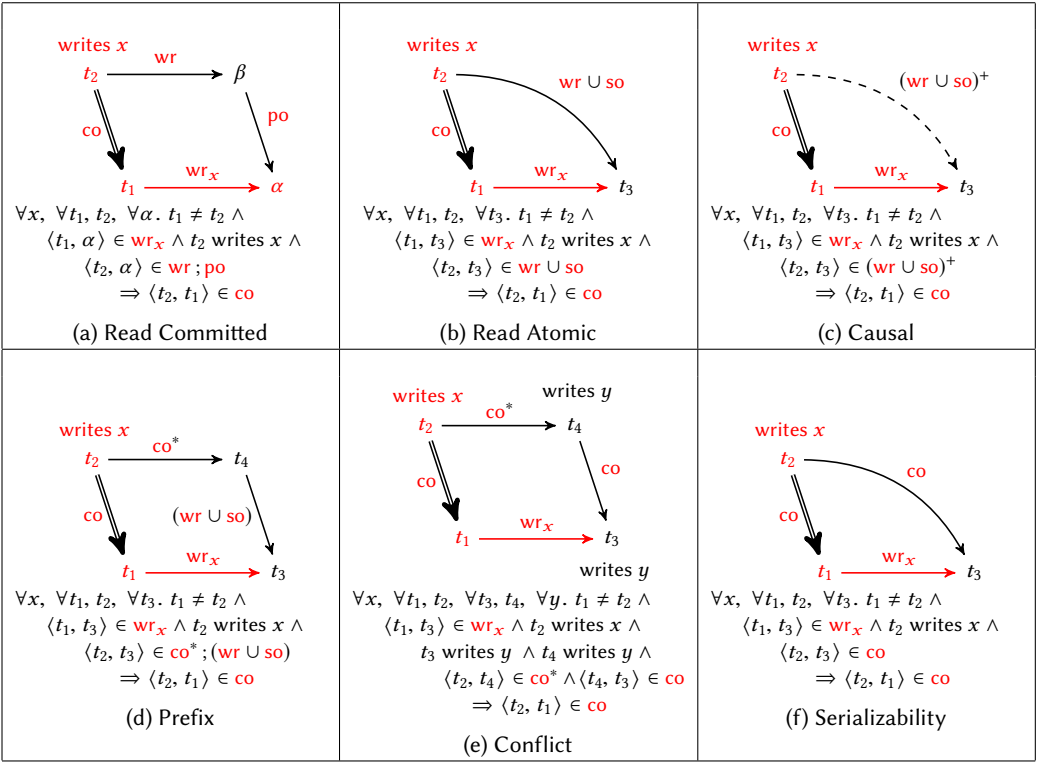
\begin{figure}
   \resizebox{\textwidth}{!}{
   \footnotesize
  \begin{tabular}{|c|c|c|}
   \hline &  & \\
   \begin{subfigure}[t]{.3\textwidth}
    \centering
    \begin{tikzpicture}[->,>=stealth',shorten >=1pt,auto,node distance=1cm,
      semithick, transform shape]
     \node[transaction state, text=red] at (0,0)       (t_1)           {$\tr_1$};
     \node[transaction state, text=red, label={above:\textcolor{red}{$\writeVar{ }{\xvar}$}}] at (-0.5,1.5) (t_2) {$\tr_2$};
     \node[transaction state, text=red] at (2,0)       (o_1)           {$\alpha$};
     \node[transaction state] at (1.5,1.5) (o_2) {$\beta$};
     \path (t_1) edge[red] node {$\wro[\xvar]$} (o_1);
     % \path (t_2) edge[blue] node {$\CO$} (t_1);
     \path (t_2) edge node {$\wro$} (o_2);
     \path (o_2) edge node {$\po$} (o_1);
     \path (t_2) edge[left,double] node {$\co$} (t_1);
    \end{tikzpicture}
    \parbox{\textwidth}{
     $\forall \xvar,\ \forall \tr_1, \tr_2,\ \forall \alpha.\ \tr_1\neq \tr_2\ \land$
     
     \hspace{4mm}$\tup{\tr_1,\alpha}\in \wro[\xvar] \land \writeVar{\tr_2}{\xvar}\ \land$ 
     
     \hspace{9mm}$\tup{\tr_2,\alpha}\in\wro\circ\po$
     
     \hspace{14mm}$\implies \tup{\tr_2,\tr_1}\in\co$
    }
    %\end{align*}
    
    \caption{$\mathsf{Read\ Committed}$}
    \label{lock_rc_def}
   \end{subfigure}
   
          &     
   
   \begin{subfigure}[t]{.3\textwidth}
    \centering
    \begin{tikzpicture}[->,>=stealth',shorten >=1pt,auto,node distance=4cm,
      semithick, transform shape]
     \node[transaction state, text=red] at (0,0)       (t_1)           {$\tr_1$};
     \node[transaction state] at (2,0)       (t_3)           {$\tr_3$};
     \node[transaction state, text=red,label={above:\textcolor{red}{$\writeVar{ }{\xvar}$}}] at (-.5,1.5) (t_2) {$\tr_2$};
     \path (t_1) edge[red] node {$\wro[\xvar]$} (t_3);
     % \path (t_2) edge[blue] node {$\CO$} (t_1);
     \path (t_2) edge[bend left] node {$\wro \cup \so$} (t_3);
     \path (t_2) edge[left,double] node {$\co$} (t_1);
    \end{tikzpicture}
    \parbox{\textwidth}{
     $\forall \xvar,\ \forall \tr_1, \tr_2,\ \forall \tr_3.\ \tr_1\neq \tr_2\ \land$
     
     \hspace{4mm}$\tup{\tr_1,\tr_3}\in \wro[\xvar] \land \writeVar{\tr_2}{\xvar}\ \land$ 
     
     \hspace{9mm}$\tup{\tr_2,\tr_3}\in\wro\cup\so$
     
     \hspace{14mm}$\implies \tup{\tr_2,\tr_1}\in\co$
    }
    
    \caption{$\mathsf{Read\ Atomic}$}
    \label{ra_def}
   \end{subfigure}
   
   &
   
   \begin{subfigure}[t]{.3\textwidth}
    \centering
    \begin{tikzpicture}[->,>=stealth',shorten >=1pt,auto,node distance=4cm,
      semithick, transform shape]
     \node[transaction state, text=red] at (0,0)       (t_1)           {$\tr_1$};
     \node[transaction state] at (2,0)       (t_3)           {$\tr_3$};
     \node[transaction state, text=red,label={above:\textcolor{red}{$\writeVar{ }{\xvar}$}}] at (-.5,1.5) (t_2) {$\tr_2$};
     \path (t_1) edge[red] node {$\wro[\xvar]$} (t_3);
     % \path (t_2) edge[blue] node {$\CO$} (t_1);
     \path (t_2) edge[dashed, bend left] node {$(\wro \cup \so)^+$} (t_3);
     %   \path [->, decoration={snake}] (t_2) edge[decorate] node[auto] {F} (t_3);
     \path (t_2) edge[left,double] node {$\co$} (t_1);
    \end{tikzpicture}
    \parbox{\textwidth}{
     $\forall \xvar,\ \forall \tr_1, \tr_2,\ \forall \tr_3.\ \tr_1\neq \tr_2\ \land$
     
     \hspace{4mm}$\tup{\tr_1,\tr_3}\in \wro[\xvar] \land \writeVar{\tr_2}{\xvar}\ \land$ 
     
     \hspace{9mm}$\tup{\tr_2,\tr_3}\in(\wro\cup\so)^+$
     
     \hspace{14mm}$\implies \tup{\tr_2,\tr_1}\in\co$
    }
    
    \caption{$\mathsf{Causal}$}
    \label{cc_def}
   \end{subfigure}

   \\ \hline & & \\
    
   \begin{subfigure}[t]{.3\textwidth}
    \centering
    \begin{tikzpicture}[->,>=stealth',shorten >=1pt,auto,node distance=4cm,
      semithick, transform shape]
     \node[transaction state, text=red] at (0,0)       (t_1)           {$\tr_1$};
     \node[transaction state] at (2,0)       (t_3)           {$\tr_3$};
     \node[transaction state, text=red,label={above:\textcolor{red}{$\writeVar{ }{\xvar}$}}] at (-0.5,1.5) (t_2) {$\tr_2$};
     \node[transaction state] at (1.5,1.5) (t_4) {$\tr_4$};
     \path (t_1) edge[red] node {$\wro[\xvar]$} (t_3);
     % \path (t_2) edge[blue] node {$\CO$} (t_1);
     \path (t_2) edge node {$\co^*$} (t_4);
     \path (t_4) edge[left] node {$(\wro \cup \so)$} (t_3);
     \path (t_2) edge[left,double] node {$\co$} (t_1);
    \end{tikzpicture}
    \parbox{\textwidth}{
     $\forall \xvar,\ \forall \tr_1, \tr_2,\ \forall \tr_3.\ \tr_1\neq \tr_2\ \land$
     
     \hspace{4mm}$\tup{\tr_1,\tr_3}\in \wro[\xvar] \land \writeVar{\tr_2}{\xvar}\ \land$ 
     
     \hspace{9mm}$\tup{\tr_2,\tr_3}\in\co^*\circ\,(\wro\cup\so)$
     
     \hspace{14mm}$\implies \tup{\tr_2,\tr_1}\in\co$
    }
    
    \caption{$\mathsf{Prefix}$}
    \label{pre_def}
   \end{subfigure}

   &
   \begin{subfigure}[t]{.32\textwidth}
    \centering
    \begin{tikzpicture}[->,>=stealth',shorten >=1pt,auto,node distance=4cm,
      semithick, transform shape]
     \node[transaction state, text=red] at (0,0)       (t_1)           {$\tr_1$};
     \node[transaction state, label={below:$\writeVar{ }{\yvar}$}] at (2,0)       (t_3)           {$\tr_3$};
     \node[transaction state, text=red,label={above:\textcolor{red}{$\writeVar{ }{\xvar}$}}] at (-.5,1.5) (t_2) {$\tr_2$};
     \node[transaction state, label={above:{$\writeVar{}{\yvar}$}}] at (1.5,1.5) (t_4) {$\tr_4$};
     \path (t_1) edge[red] node {$\wro[\xvar]$} (t_3);
     % \path (t_2) edge[blue] node {$\CO$} (t_1);
     \path (t_2) edge node {$\co^*$} (t_4);
     \path (t_4) edge node {$\co$} (t_3);
     \path (t_2) edge[left,double] node {$\co$} (t_1);
    \end{tikzpicture}
    \parbox{\textwidth}{
     $\forall \xvar,\ \forall \tr_1, \tr_2,\ \forall \tr_3, \tr_4,\ \forall \yvar.\ \tr_1\neq \tr_2\ \land$
     
     \hspace{4mm}$\tup{\tr_1,\tr_3}\in \wro[\xvar] \land \writeVar{\tr_2}{\xvar}\ \land$ 
     
     \hspace{9mm}$\writeVar{\tr_3}{\yvar}\ \land \writeVar{\tr_4}{\yvar}\ \land$ 
     
     \hspace{12mm}$\tup{\tr_2,\tr_4}\in\co^*\ \land \tup{\tr_4,\tr_3}\in\co$
     
     \hspace{16mm}$\implies \tup{\tr_2,\tr_1}\in\co$
    }
    
    \caption{$\mathsf{Conflict}$}
    \label{confl_def}
   \end{subfigure}
          &     
   \begin{subfigure}[t]{.3\textwidth}
    \centering
    \begin{tikzpicture}[->,>=stealth',shorten >=1pt,auto,node distance=4cm,
      semithick, transform shape]
     \node[transaction state, text=red] at (0,0)       (t_1)           {$\tr_1$};
     \node[transaction state] at (2,0)       (t_3)           {$\tr_3$};
     \node[transaction state, text=red, label={above:\textcolor{red}{$\writeVar{ }{\xvar}$}}] at (-.5,1.5) (t_2) {$\tr_2$};
     \path (t_1) edge[red] node {$\wro[\xvar]$} (t_3);
     % \path (t_2) edge[blue] node {$\CO$} (t_1);
     \path (t_2) edge[bend left] node {$\CO$} (t_3);
     \path (t_2) edge[left,double] node {$\co$} (t_1);
    \end{tikzpicture}
    \parbox{\textwidth}{
     $\forall \xvar,\ \forall \tr_1, \tr_2,\ \forall \tr_3.\ \tr_1\neq \tr_2\ \land$
     
     \hspace{4mm}$\tup{\tr_1,\tr_3}\in \wro[\xvar] \land \writeVar{\tr_2}{\xvar}\ \land$ 
     
     \hspace{9mm}$\tup{\tr_2,\tr_3}\in\co$
     
     \hspace{14mm}$\implies \tup{\tr_2,\tr_1}\in\co$
    }
    
    \caption{$\mathsf{Serializability}$}
    \label{ser_def}
   \end{subfigure}
   \\ \hline
  \end{tabular}
  }
  \caption{Definitions of consistency axioms. The reflexive and transitive, resp., transitive, closure of a relation $rel$ is denoted by $rel^*$, resp., $rel^+$. Also, $\circ$ denotes the composition of two relations, i.e., $rel_1 \circ rel_2 = \{\tup{a, b} | \exists c. \tup{a, c} \in rel_1 \land \tup{c, b} \in rel_2\}$.}
  \label{consistency_defs}
 \end{figure}

% Practically, \textsc{Int} ensures each transaction takes one global snapshot of variables at the beginning. Then no other global changes affect that snapshot for any read or write local to that transaction. \textsc{Ext} ensures each transaction always observes the latest global snapshot that is visible to it.

We describe an axiomatic framework to characterize the set of histories satisfying a certain consistency criterion. The overarching principle is to say that a history satisfies a certain criterion if there exists a strict total order on its transactions, called \emph{commit order} and denoted by $\co$, which extends the write-read relation and the session order, and which satisfies certain properties. These properties are expressed by a set of axioms that relate the commit order with the session-order and the write-read relation in the history.

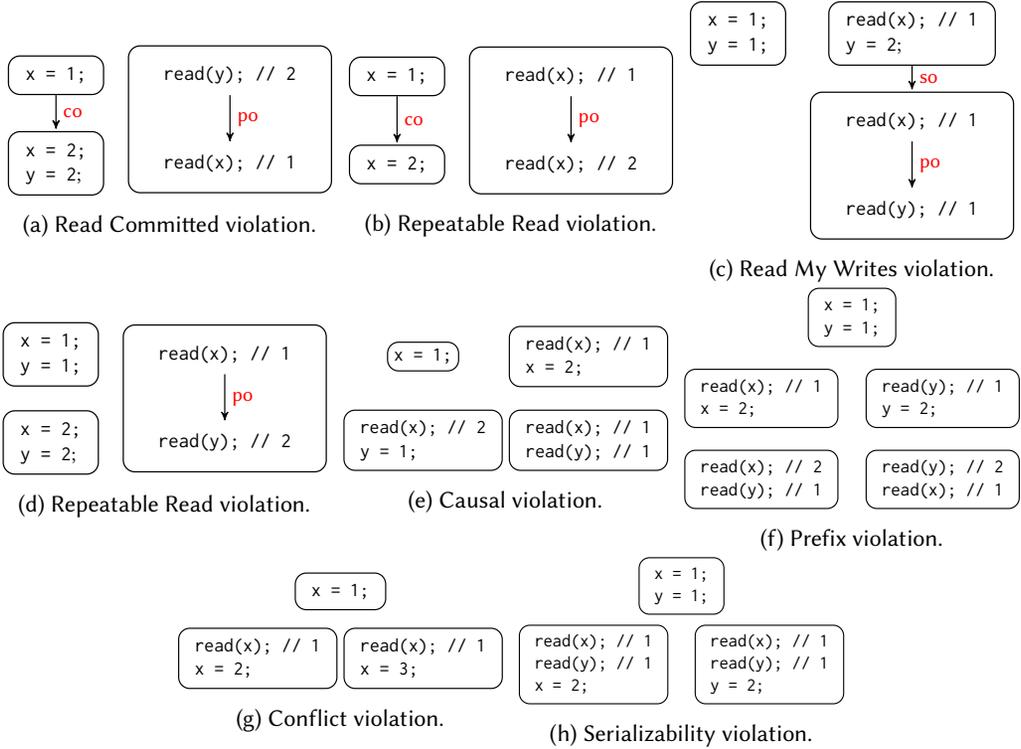
\begin{figure}
  
   \centering
   \begin{subfigure}{.32\textwidth}
  \resizebox{\textwidth}{!}{
\begin{tikzpicture}[->,>=stealth',shorten >=1pt,auto,node distance=3cm,
    semithick, transform shape]
    \node[draw, rounded corners=2mm] (t1) at (0, 0) {\begin{tabular}{l} \texttt{x = 1;} \end{tabular}};
   \node[draw, rounded corners=2mm] (t2) at (0, -1.5) {\begin{tabular}{l} \texttt{x = 2;} \\ \texttt{y = 2};\end{tabular}};
   \node[draw, rounded corners=2mm, minimum width=3.5cm, minimum height=2.5cm] (t3) at (3, -0.75) {};
   \node (t3_1) at (3, 0) {\begin{tabular}{l} \texttt{read(y); // 2} \end{tabular}};
   \node (t3_2) at (3, -1.5) {\begin{tabular}{l} \texttt{read(x); // 1} \end{tabular}};
   % \path (t1) edge node {} (t3_2);
   % \path (t2) edge node {} (t3_1);
   \path (t1) edge node {$\co$} (t2);
   \path (t3_1) edge node {$\po$} (t3_2);
  \end{tikzpicture}  
    }
    \caption{$\mathsf{Read\ Committed}$ violation.}
    \label{rc_example:1}
    
\end{subfigure}
\begin{subfigure}{.32\textwidth}
\resizebox{\textwidth}{!}{
\begin{tikzpicture}[->,>=stealth',shorten >=1pt,auto,node distance=3cm,
 semithick, transform shape]
 \node[draw, rounded corners=2mm] (t1) at (0, 0) {\begin{tabular}{l} \texttt{x = 1;} \end{tabular}};
\node[draw, rounded corners=2mm] (t2) at (0, -1.5) {\begin{tabular}{l} \texttt{x = 2;} \end{tabular}};
\node[draw, rounded corners=2mm, minimum width=3.5cm, minimum height=2.5cm] (t3) at (3, -0.75) {};
\node (t3_1) at (3, -0) {\begin{tabular}{l} \texttt{read(x); // 1} \end{tabular}};
\node (t3_2) at (3, -1.5) {\begin{tabular}{l} \texttt{read(x); // 2} \end{tabular}};
% \path (t1) edge node {} (t3_1);
% \path (t2) edge node {} (t3_2);
\path (t1) edge node {$\co$} (t2);
\path (t3_1) edge node {$\po$} (t3_2);
\end{tikzpicture}  
}
 \caption{Repeatable Read violation.}
 \label{rr_example:1}
\end{subfigure}
\begin{subfigure}{.32\textwidth}
\resizebox{\textwidth}{!}{
\begin{tikzpicture}[->,>=stealth',shorten >=1pt,auto,node distance=3cm,
 semithick, transform shape]
 \node[draw, rounded corners=2mm] (t1) at (0, 1.5) {\begin{tabular}{l} \texttt{x = 1;} \\ \texttt{y = 1;}\end{tabular}};
\node[draw, rounded corners=2mm] (t2) at (3, 1.5) {\begin{tabular}{l} \texttt{read(x); // 1} \\ \texttt{y = 2};\end{tabular}};
\node[draw, rounded corners=2mm, minimum width=3.5cm, minimum height=2.5cm] (t3) at (3, -0.75) {};
\node (t3_1) at (3, 0) {\begin{tabular}{l} \texttt{read(x); // 1} \end{tabular}};
\node (t3_2) at (3, -1.5) {\begin{tabular}{l} \texttt{read(y); // 1} \end{tabular}};
% \path (t1) edge node {} (t3);
\path (t2) edge node {$\so$} (t3);
% \path (t1) edge node {} (t2);
\path (t3_1) edge node {$\po$} (t3_2);
\end{tikzpicture}  
}
 \caption{Read My Writes violation.}
 \label{rmw_example:1}
\end{subfigure}
\begin{subfigure}{.32\textwidth}
\resizebox{\textwidth}{!}{
\begin{tikzpicture}[->,>=stealth',shorten >=1pt,auto,node distance=3cm,
 semithick, transform shape]
 \node[draw, rounded corners=2mm] (t1) at (0, 0) {\begin{tabular}{l} \texttt{x = 1;} \\ \texttt{y = 1;} \end{tabular}};
\node[draw, rounded corners=2mm] (t2) at (0, -1.5) {\begin{tabular}{l} \texttt{x = 2;} \\ \texttt{y = 2};\end{tabular}};
\node[draw, rounded corners=2mm, minimum width=3.5cm, minimum height=2.5cm] (t3) at (3, -0.75) {};
\node (t3_2) at (3, -1.5) {\begin{tabular}{l} \texttt{read(y); // 2} \end{tabular}};
\node (t3_1) at (3, -0) {\begin{tabular}{l} \texttt{read(x); // 1} \end{tabular}};
% \path (t1) edge node {} (t3_1);
% \path (t2) edge node {} (t3_2);
% \path (t1) edge node {$\co$} (t2);
\path (t3_1) edge node {$\po$} (t3_2);
\end{tikzpicture}  
}
 \caption{Repeatable Read violation.}
 \label{ra_example:1}
\end{subfigure}
\begin{subfigure}{.32\textwidth}
\resizebox{\textwidth}{!}{
\begin{tikzpicture}[->,>=stealth',shorten >=1pt,auto,node distance=3cm,
 semithick, transform shape]
 \node[draw, rounded corners=2mm] (t1) at (0, 1.5) {\texttt{x = 1;}};
\node[draw, rounded corners=2mm] (t2) at (3, 1.5) {\begin{tabular}{l} \texttt{read(x); // 1} \\ \texttt{x = 2;} \end{tabular}};
\node[draw, rounded corners=2mm] (t3) at (3, 0) {\begin{tabular}{l} \texttt{read(x); // 1} \\ \texttt{read(y); // 1} \end{tabular}};
\node[draw, rounded corners=2mm] (t4) at (0, 0) {\begin{tabular}{l} \texttt{read(x); // 2} \\ \texttt{y = 1;} \end{tabular}};
% \path (t1) edge node {} (t3);
% \path (t2) edge node {$\so$} (t3);
% \path (t1) edge node {} (t2);
% \path (t3_1) edge node {$\po$} (t3_2);
\end{tikzpicture}  
}
 \caption{$\mathsf{Causal}$ violation.}
 \label{cc_example:1}
\end{subfigure}
\begin{subfigure}{.33\textwidth}
\resizebox{\textwidth}{!}{
\begin{tikzpicture}[->,>=stealth',shorten >=1pt,auto,node distance=3cm,
 semithick, transform shape]
 \node[draw, rounded corners=2mm] (t1) at (0, 0) {\begin{tabular}{l} \texttt{x = 1;} \\ \texttt{y = 1;}\end{tabular}};
 \node[draw, rounded corners=2mm] (t2) at (-1.7, -1.5) {\begin{tabular}{l} \texttt{read(x); // 1} \\ \texttt{x = 2;} \end{tabular}};
\node[draw, rounded corners=2mm] (t3) at (1.7, -1.5) {\begin{tabular}{l} \texttt{read(y); // 1} \\ \texttt{y = 2;} \end{tabular}};
\node[draw, rounded corners=2mm] (t4) at (-1.7, -3) {\begin{tabular}{l} \texttt{read(x); // 2} \\ \texttt{read(y); // 1} \end{tabular}};
\node[draw, rounded corners=2mm] (t5) at (1.7, -3) {\begin{tabular}{l} \texttt{read(y); // 2} \\ \texttt{read(x); // 1} \end{tabular}};
% \node[draw, rounded corners=2mm] (t3) at (1.5, 0) {\begin{tabular}{l} \texttt{read(x); // 2} \\ \texttt{read(y); // 1} \end{tabular}};
% \path (t1) edge node {} (t3);
% \path (t2) edge node {$\so$} (t3);
% \path (t1) edge node {} (t2);
% \path (t3_1) edge node {$\po$} (t3_2);
\end{tikzpicture}  
}
 \caption{$\mathsf{Prefix}$ violation.}
 \label{pre_example:1}
\end{subfigure}
\begin{subfigure}{.32\textwidth}
\resizebox{\textwidth}{!}{
\begin{tikzpicture}[->,>=stealth',shorten >=1pt,auto,node distance=3cm,
 semithick, transform shape]
 \node[draw, rounded corners=2mm] (t1) at (0, 0) {\begin{tabular}{l} \texttt{x = 1;} \end{tabular}};
 \node[draw, rounded corners=2mm] (t2) at (-1.5, -1.2) {\begin{tabular}{l} \texttt{read(x); // 1} \\ \texttt{x = 2;} \end{tabular}};
 \node[draw, rounded corners=2mm] (t3) at (1.5, -1.2) {\begin{tabular}{l} \texttt{read(x); // 1} \\ \texttt{x = 3;} \end{tabular}};
 % \node[draw, rounded corners=2mm] (t3) at (0, -2.4) {\begin{tabular}{l} \texttt{read(x); // 2} \end{tabular}};
% \node[draw, rounded corners=2mm] (t3) at (1.7, -1.5) {\begin{tabular}{l} \texttt{read(y); // 1} \\ \texttt{y = 2;} \end{tabular}};
% \node[draw, rounded corners=2mm] (t3) at (1.5, 0) {\begin{tabular}{l} \texttt{read(x); // 2} \\ \texttt{read(y); // 1} \end{tabular}};
% \path (t1) edge node {} (t3);
% \path (t2) edge node {$\co$} (t3);
% \path (t1) edge node {} (t2);
% \path (t3_1) edge node {$\po$} (t3_2);
\end{tikzpicture} 
}
 \caption{$\mathsf{Conflict}$ violation.}
 \label{conf_example:1}
\end{subfigure}
\begin{subfigure}{.32\textwidth}
\resizebox{\textwidth}{!}{
\begin{tikzpicture}[->,>=stealth',shorten >=1pt,auto,node distance=3cm,
 semithick, transform shape]
 \node[draw, rounded corners=2mm] (t1) at (0, 0) {\begin{tabular}{l} \texttt{x = 1;} \\ \texttt{y = 1;}\end{tabular}};
 \node[draw, rounded corners=2mm] (t2) at (-1.7, -1.5) {\begin{tabular}{l} \texttt{read(x); // 1} \\ \texttt{read(y); // 1} \\ \texttt{x = 2;} \end{tabular}};
 \node[draw, rounded corners=2mm] (t3) at (1.7, -1.5) {\begin{tabular}{l} \texttt{read(x); // 1} \\ \texttt{read(y); // 1} \\ \texttt{y = 2;} \end{tabular}};
% \node[draw, rounded corners=2mm] (t3) at (1.7, -1.5) {\begin{tabular}{l} \texttt{read(y); // 1} \\ \texttt{y = 2;} \end{tabular}};
% \node[draw, rounded corners=2mm] (t3) at (1.5, 0) {\begin{tabular}{l} \texttt{read(x); // 2} \\ \texttt{read(y); // 1} \end{tabular}};
% \path (t1) edge node {} (t3);
% \path (t2) edge node {$\so$} (t3);
% \path (t1) edge node {} (t2);
% \path (t3_1) edge node {$\po$} (t3_2);
\end{tikzpicture}  
}
 \caption{$\mathsf{Serializability}$ violation.}
 \label{ser_example:1}
\end{subfigure}

  \caption{Examples of histories used to explain the axioms in Figure~\ref{consistency_defs}. For readability, the $\wro$ relation is defined by the values written in comments with each {\tt read}.}
  \label{counter_example:1}
\vspace{-3mm}
\end{figure}

The axioms we use have a uniform shape: they define mandatory $\co$ predecessors $\tr_2$ of a transaction $\tr_1$ that is read in the history. For instance, the criterion called \textsc{Read Committed} (RC)~\cite{DBLP:conf/sigmod/BerensonBGMOO95} requires that every value read in the history was written by a committed transaction, and also, that the reads in the same transaction are ``monotonic'' in the sense that they do not return values that are older, w.r.t. the commit order, than other values read in the past\footnote{This monotonicity property corresponds to the fact that in the original formulation of \textsc{Read Committed}~\cite{DBLP:conf/sigmod/BerensonBGMOO95}, every write is guarded by the acquisition of a lock on the written variable, that is held until the end of the transaction.}. While the first condition holds for any history (because of the surjectivity of $\wro$), the second condition is expressed by the axiom $\mathsf{Read\ Committed}$ in Figure~\ref{lock_rc_def}. This axiom states that for any transaction $\tr_1$ writing a variable $\xvar$ that is read in a transaction $\tr$, the set of transactions $\tr_2$ writing $\xvar$ and read previously in the same transaction must precede $\tr_1$ in commit order. For instance, Figure~\ref{rc_example:1} shows a history and a (partial) commit order that does not satisfy this axiom because ${\tt read(x)}$ returns the value written in a transaction ``older'' than the transaction read in the previous ${\tt read(y)}$. An example of a history and commit order satisfying this axiom is given in Figure~\ref{rr_example:1}.

%TODO GIVE A POSITIVE AND A NEGATIVE EXAMPLE W.R.T. READ COMMITTED, AND EXPLAIN THE APPLICATION OF THIS AXIOM. THE EXAMPLES SHOULD BE "SYNTHETIC", WITH VARIABLES $\xvar$ and $\yvar$.
More precisely, the axioms are first-order formulas\footnote{These formulas are interpreted on tuples $\tup{\hist,\co}$ of a history $\hist$ and a commit order $\co$ on the transactions in $\hist$ as usual.} of the following form:
\begin{align*}
  & \forall \xvar,\ \forall \tr_1,\tr_2,\ \forall \alpha.\ \tr_1\neq \tr_2\land \tup{\tr_1,\alpha}\in \wro[\xvar] \land \writeVar{\tr_2}{\xvar} \land \phi(\tr_2,\alpha) \implies \tup{\tr_2,\tr_1}\in\co 
\end{align*}
where $\phi$ is a property relating $\tr_2$ and $\alpha$ (i.e., the read or the transaction reading from $\tr_1$) that varies from one axiom to another. Intuitively, this axiom schema states the following: in order for $\alpha$ to read specifically $t_1$'s write on $x$, it must be the case that every $t_2$ that also writes $x$ and satisfies $\phi(t_2,\alpha)$ was committed before $t_1$. Note that in all cases we consider, $\phi(t_2,\alpha)$ already ensures that $t_2$ is committed before the read $\alpha$, so this axiom schema ensures that $t_2$ is furthermore committed before $t_1$'s write.

The axioms used throughout the paper are given in Figure~\ref{consistency_defs}. The property $\phi$ relates $\tr_2$ and $\alpha$ using the write-read relation and the session order in the history, and the commit order. 
%The axioms are first-order formulas whose satisfaction on tuples $\tup{\hist,\co}$ of a history $\hist$ and a commit order $\co$ on the transactions in $\hist$ is defined as usual.

In the following, we explain the rest of the consistency criteria we consider and the axioms defining them. \textsc{Read Atomic} (RA)~\cite{DBLP:conf/concur/Cerone0G15} is a strengthening of \textsc{Read Committed} defined by the axiom $\mathsf{Read\ Atomic}$, which states that for any transaction $\tr_1$ writing a variable $\xvar$ that is read in a transaction $\tr_3$, the set of $\wro$ or $\so$ predecessors of $\tr_3$ writing $\xvar$ must precede $\tr_1$ in commit order. The case of $\wro$ predecessors corresponds to the Repeatable Read criterion in~\cite{DBLP:conf/sigmod/BerensonBGMOO95} which requires that successive reads of the same variable in the same transaction return the same value, Figure~\ref{rr_example:1} showing a violation, and also that every read of a variable $\xvar$ in a transaction $\tr$ returns the value written by the maximal transaction $\tr'$ (w.r.t. the commit order) that is read by $\tr$, Figure~\ref{ra_example:1} showing a violation (for any commit order between the transactions on the left, either ${\tt read(x)}$ or ${\tt read(y)}$ will return a value not written by the maximal transaction). The case of $\so$ predecessors corresponds to the  ``read-my-writes'' guarantee~\cite{DBLP:conf/pdis/TerryDPSTW94} concerning sessions, which states that a transaction $\tr$ must observe previous writes in the same session. For instance, {\tt read(y)} returning 1 in Figure~\ref{rmw_example:1} shows that the last transaction on the right does not satisfy this guarantee: the transaction writing 1 to {\tt y} was already visible to that session before it wrote 2 to {\tt y}, and therefore the value 2 should have been read. $\mathsf{Read\ Atomic}$ requires that the $\so$ predecessor of the transaction reading {\tt y} be ordered in $\co$ before the transaction writing 1 to {\tt y}, which makes the union $\co\cup\wro$ cyclic.

The following lemma shows that for histories satisfying $\mathsf{Read\ Atomic}$, the inverse of $\wro[\xvar]$ extended to transactions is a total function (see Appendix~\ref{app:definitions} for the proof).

\begin{lemma}
 Let $\hist=\tup{T, \so, \wro}$ be a history. 
 If $\tup{\hist,\co}$ satisfies $\mathsf{Read\ Atomic}$, then %the extension of $\wro[\xvar]$ to transactions 
 for every transaction $\tr$ and two reads $\rd[\id_1]{\xvar}{\val_1},\rd[\id_2]{\xvar}{\val_2}\in \readOp{\tr}$, $\wro^{-1}(\rd[\id_1]{\xvar}{\val_1})=\wro^{-1}(\rd[\id_2]{\xvar}{\val_2})$ and $\val_1 = \val_2$.
\end{lemma}

\begin{table}[t]
 \centering
 % \resizebox{\textwidth}{!}{
 \begin{tabular}{|l|l|}
  \hline
  \shortstack{Consistency model}   & Axioms                                   \\
  \hline
  \textsc{Read Committed} (RC)     & $\mathsf{Read\ Committed}$               \\
  \hline
  \textsc{Read Atomic} (RA)        & $\mathsf{Read\ Atomic}$                  \\
  \hline
  \textsc{Causal consistency} (CC) & $\mathsf{Causal}$                        \\
  \hline
  \textsc{Prefix consistency} (PC) & $\mathsf{Prefix}$                        \\
  \hline
  \textsc{Snapshot isolation} (SI) & $\mathsf{Prefix}\land \mathsf{Conflict}$ \\
  \hline
  \textsc{Serializability} (SER)   & $\mathsf{Serializability}$               \\
  \hline
  % \multicolumn{3}{|c|}{
  %  $\textbf{RA} \supset \textbf{CC} \supset \textbf{PC} \supset \textbf{SI} \supset \textbf{SER}$
  % }                                                                   \\
  % \hline
 \end{tabular}
 % }
 \caption{Consistency model definitions}
 \label{weakconsistency:2}
 \vspace{-3mm}
\end{table}

\textsc{Causal Consistency} (CC)~\cite{DBLP:journals/cacm/Lamport78} is defined by the axiom $\mathsf{Causal}$, which states that for any transaction $\tr_1$ writing a variable $\xvar$ that is read in a transaction $\tr_3$, the set of $(\wro\cup \so)^+$ predecessors of $\tr_3$ writing $\xvar$ must precede $\tr_1$ in commit order ($(\wro\cup \so)^+$ is usually called the \emph{causal} order). A violation of this axiom can be found in Figure~\ref{cc_example:1}: the transaction $\tr_2$ writing 2 to {\tt x} is a $(\wro\cup \so)^+$ predecessor of the transaction $\tr_3$ reading 1 from {\tt x} because the transaction $\tr_4$, writing 1 to {\tt y}, reads {\tt x} from $\tr_2$ and $\tr_3$ reads {\tt y} from $\tr_4$. This implies that $\tr_2$ should precede in commit order the transaction $\tr_1$ writing 1 to {\tt x}, which again, is inconsistent with the write-read relation ($\tr_2$ reads from $\tr_1$).

\textsc{Prefix consistency} (PC)~\cite{DBLP:conf/ecoop/BurckhardtLPF15} is a strengthening of CC, which requires that every transaction observes a prefix of a commit order between all the transactions. With the intuition that the observed transactions are $\wro\cup\so$ predecessors, the axiom $\mathsf{Prefix}$ defining PC, states that for any transaction $\tr_1$ writing a variable $\xvar$ that is read in a transaction $\tr_3$, the set of $\co^*$ predecessors of transactions observed by $\tr_3$ writing $\xvar$ must precede $\tr_1$ in commit order (we use $\co^*$ to say that even the transactions observed by $\tr_3$ must precede $\tr_1$). This ensures the prefix property stated above. An example of a PC violation can be found in Figure~\ref{pre_example:1}: the two transactions on the bottom read from the three transactions on the top, but any serialization of those three transactions will imply that one of the combinations {\tt x=1}, {\tt y=2} or {\tt x=2}, {\tt y=1} cannot be produced at the end of a prefix in this serialization.

\textsc{Snapshot Isolation} (SI)~\cite{DBLP:conf/sigmod/BerensonBGMOO95} is a strengthening of PC that disallows two transactions to observe the same prefix of a commit order if they \emph{conflict}, i.e., write to a common variable. It is defined by the conjunction of $\mathsf{Prefix}$ and another axiom called $\mathsf{Conflict}$, which requires that for any transaction $\tr_1$ writing a variable $\xvar$ that is read in a transaction $\tr_3$, the set of $\co^*$ predecessors writing $\xvar$ of transactions conflicting with $\tr_3$ and before $\tr_3$ in commit order, must precede $\tr_1$ in commit order. Figure~\ref{conf_example:1} shows a $\mathsf{Conflict}$ violation.

Finally, \textsc{Serializability} (SER)~\cite{DBLP:journals/jacm/Papadimitriou79b} is defined by the axiom with the same name, which requires that for any transaction $\tr_1$ writing to a variable $\xvar$ that is read in a transaction $\tr_3$, the set of $\co$ predecessors of $\tr_3$ writing $\xvar$ must precede $\tr_1$ in commit order. This ensures that each transaction observes the effects of all the $\co$ predecessors. Figure~\ref{ser_example:1} shows a $\mathsf{Serializability}$ violation.

\begin{lemma}
 The following entailments hold:
 \begin{align*}
   & \mathsf{Causal} \implies \mathsf{Read\ Atomic}\implies \mathsf{Read\ Committed} \\
   & \mathsf{Prefix} \implies \mathsf{Causal}                                        \\
   & \mathsf{Serializability} \implies \mathsf{Prefix}\land \mathsf{Conflict}        
 \end{align*} 
 \label{axioms-rel}
\end{lemma}

\begin{definition}
 Given a set of axioms $X$ defining a criterion $C$ like in Table~\ref{weakconsistency:2}, a history $\hist=\tup{T, \so, \wro}$ \emph{satisfies} $C$ iff there exists a strict total order $\co$ such that $\wro\cup\so\subseteq \co$ and $\tup{h,\co}$ satisfies $X$.
 % Given a $\CO$(\textit{commit order}), a total order on $T$ which extends $\wro \cup \so$, we can define consistency axioms from table \ref{consistency_defs}. For each axiom, the situation in the table implies, $\Path{\tr_2}{\CO}{\tr_1}$.
 \label{axiom-criterion}
\end{definition}

Definition~\ref{axiom-criterion} and Lemma~\ref{axioms-rel} imply that each consistency criterion in Table~\ref{weakconsistency:2} is stronger than its predecessors (reading them from top to bottom), e.g., CC is stronger than RA and RC. This relation is strict, e.g., RA is not stronger than CC. 
%stronger from top to bottom order. Infact each criteria is strictly stronger than its previously weaker criteria, \ie \textsf{Snapshot Isolation} does not imply \textsf{Serializability}. 
These definitions are equivalent with previous formalizations by~\citet{DBLP:conf/concur/Cerone0G15} (see Appendix~\ref{app:gotsman}).

\section{Checking Consistency Criteria}\label{sec:general}

% So we have reduced the problem of verifying a history to the problem of finding appropriate dependencies between transactions of dependency edges in a dependency graph which satisfies the semantic definition of a consistency.

This section establishes the complexity of checking the different consistency criteria in Table~\ref{weakconsistency:2} for a given history. More precisely, we show that \textsc{Read Committed}, \textsc{Read Atomic}, and \textsc{Causal Consistency} can be checked in polynomial time while the problem of checking the rest of the criteria is NP-complete. 

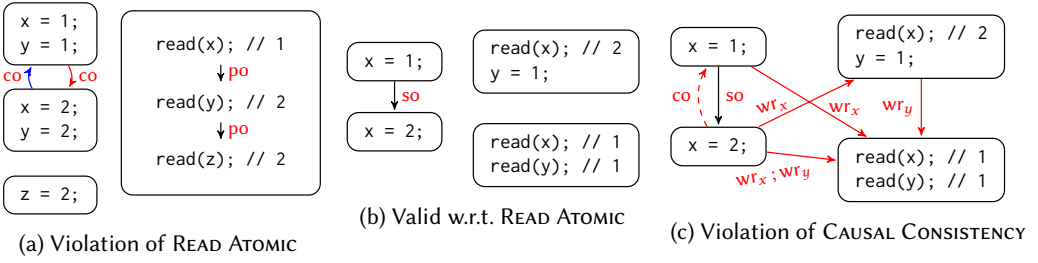
\begin{figure}
 \begin{subfigure}{.32\textwidth}
  \resizebox{\textwidth}{!}{
   \begin{tikzpicture}[->,>=stealth',shorten >=1pt,auto,node distance=3cm,
     semithick, transform shape]
    \node[draw, rounded corners=2mm] (t1) at (0, .2) {\begin{tabular}{l} \texttt{x = 1;} \\ \texttt{y = 1;} \end{tabular}};
    \node[draw, rounded corners=2mm] (t2) at (0, -1.3) {\begin{tabular}{l} \texttt{x = 2;} \\ \texttt{y = 2;} \end{tabular}};
    \node[draw, rounded corners=2mm] (t3) at (0, -2.6) {\begin{tabular}{l} \texttt{z = 2;} \end{tabular}};
    \node[draw, rounded corners=2mm, minimum width=3.5cm, minimum height=3.2cm] (t4) at (3, -1) {};
    \node (t4_1) at (3, .0) {\begin{tabular}{l} \texttt{read(x); // 1} \end{tabular}};
    \node (t4_2) at (3, -1) {\begin{tabular}{l} \texttt{read(y); // 2} \end{tabular}};
    \node (t4_3) at (3, -2) {\begin{tabular}{l} \texttt{read(z); // 2} \end{tabular}};
    \path (t4_1) edge node {$\po$} (t4_2);
    \path (t1) edge[red, bend left] node {$\co$} (t2);
    \path (t4_2) edge node {$\po$} (t4_3);
    \path (t2) edge[blue, bend left] node {$\co$} (t1);
   \end{tikzpicture}  
  }
  \caption{Violation of \textsc{Read Atomic}}
  \label{ra_algo_counter_example:1}
 \end{subfigure}
  \hspace{.05cm}
 \begin{subfigure}{.29\textwidth}
  \resizebox{\textwidth}{!}{
   \begin{tikzpicture}[->,>=stealth',shorten >=1pt,auto,node distance=3cm,
     semithick, transform shape]
    \node[draw, rounded corners=2mm] (s11) at (0, 0) {\begin{tabular}{l} \texttt{x = 1;} \end{tabular}};
    \node[draw, rounded corners=2mm] (s12) at (0, -1.2) {\begin{tabular}{l} \texttt{x = 2;} \end{tabular}};
    \node[draw, rounded corners=2mm] (t1) at (2.8, 0) {\begin{tabular}{l} \texttt{read(x); // 2} \\ \texttt{y = 1;} \end{tabular}};
    \node[draw, rounded corners=2mm] (t2) at (2.8, -1.6) {\begin{tabular}{l} \texttt{read(x); // 1} \\ \texttt{read(y); // 1} \end{tabular}};
    \path (s11) edge node {$\so$} (s12);
    % \path (s11) edge[red, below] node {$\wro[\xvar]$} (t2);
    % \path (s12) edge[dashed, red, bend left] node {$\co$} (s11);
    % \path (s12) edge[red, bend right, left] node {$\wro^+$} (t2);
   \end{tikzpicture}  
  }
  \caption{Valid w.r.t. \textsc{Read Atomic}}
  \label{ra_algo_example:1}
 \end{subfigure}
 \hspace{.05cm}
 \begin{subfigure}{.36\textwidth}
  \resizebox{.93\textwidth}{!}{
   \begin{tikzpicture}[->,>=stealth',shorten >=1pt,auto,node distance=4cm,
     semithick, transform shape]
    \node[draw, rounded corners=2mm] (s11) at (0, 0) {\begin{tabular}{l} \texttt{x = 1;} \end{tabular}};
    \node[draw, rounded corners=2mm] (s12) at (0, -1.7) {\begin{tabular}{l} \texttt{x = 2;} \end{tabular}};
    \node[draw, rounded corners=2mm] (t1) at (3.5, 0) {\begin{tabular}{l} \texttt{read(x); // 2} \\ \texttt{y = 1;} \end{tabular}};
    \node[draw, rounded corners=2mm] (t2) at (3.5, -2.1) {\begin{tabular}{l} \texttt{read(x); // 1} \\ \texttt{read(y); // 1} \end{tabular}};
    \path (s11) edge node {$\so$} (s12);
    \path (s11) edge[red, right] node[pos=0.6] {$\wro[\xvar]$} (t2);
    \path (s12) edge[dashed, red, bend left] node {$\co$} (s11);
    \path (s12) edge[red, left,rotate=10] node[pos=.7,yshift=-6] {$\wro[\xvar] \circ \wro[\yvar]$} (t2);
    \path (s12) edge[red, left,rotate=17] node[pos=0.4,yshift=4] {$\wro[\xvar]$} (t1);
    \path (t1) edge[red, left] node {$\wro[\yvar]$} (t2);
   \end{tikzpicture}  
  }
  \caption{Violation of \textsc{Causal Consistency}}
  \label{cc_algo_counter_example:1}
 \end{subfigure}
 \caption{Applying the RA and CC checking algorithms.}
 \vspace{-3mm}
 \label{ptime_algo_examples}
\end{figure}

Intuitively, the polynomial time results are based on the fact that the axioms defining those consistency criteria do not contain the commit order ($\co$) on the left-hand side of the entailment. Therefore, proving the existence of a commit order satisfying those axioms can be done using a saturation procedure that builds a ``partial'' commit order based on instantiating the axioms on the write-read relation and the session order in the given history. Since the commit order must be an extension of the write-read relation and the session order, it contains those two relations from the beginning. 
This saturation procedure stops when the order constraints derived this way become cyclic. For instance, let us consider applying such a procedure corresponding to RA on the histories in Figure~\ref{ra_algo_counter_example:1} and Figure~\ref{ra_algo_example:1}. Applying the axiom in Figure~\ref{ra_def} on the first history, since the transaction on the right reads 2 from $\yvar$, we get that its $\wro[\xvar]$ predecessor (i.e., the first transaction on the left) must precede the transaction writing 2 to $\yvar$ in commit order (the red edge). This holds because the $\wro[\xvar]$ predecessor writes on $\yvar$. Similarly, since the same transaction reads 1 from $\xvar$, we get that its $\wro[\yvar]$ predecessor must precede the transaction writing 1 to $\xvar$ in commit order (the blue edge). This already implies a cyclic commit order, and therefore, this history does not satisfy RA. On the other hand, for the history in Figure~\ref{ra_algo_example:1}, all the axiom instantiations are vacuous, i.e., the left part of the entailment is false, and therefore, it satisfies RA. Checking CC on the history in Figure~\ref{cc_algo_counter_example:1} requires a single saturation step: since the transaction on the bottom right reads 1 from $\xvar$, its $\wro[\xvar]\circ\wro[\yvar]$ predecessor that writes on $\xvar$ (the transaction on the bottom left) must precede in commit order the transaction writing 1 to $\xvar$. Since this is already inconsistent with the session order, we get that this history violates CC.

\begin{algorithm}[t]
{\footnotesize
 \SetKwInOut{KwInput}{Input}
 \SetKwInOut{KwOutput}{Output}
 \KwIn{A history $\hist = \tup{T, \so, \wro}$}
 \KwOut{$\mathit{true}$ iff $\hist$ satisfies \textsc{Causal consistency}}
 \BlankLine
 \If{$\so\cup\wro$ is cyclic} {
  \Return{false}\;
 }
 $\co \leftarrow \so\cup\wro$\;
 \ForEach{$\xvar \in \vars{\hist}$}{
  \ForEach{$\tr_1 \neq \tr_2 \in T$ s.t. $\tr_1$ and $\tr_2$ write $\xvar$}{
   \If{$\exists \tr_3.\ \tup{\tr_1,\tr_3}\in \wro[\xvar]\land \tup{\tr_2,\tr_3}\in (\so\cup\wro)^+$} { %\Path{\tr_2}{E_1^+}{\tr_3}, \Path{\tr_1}{\wro[\xvar]}{\tr_3}
    $\co \leftarrow \co \cup \{\tup{\tr_2, \tr_1}\}$\;
   }
  }
 }
 \eIf{$\co$ is cyclic}{
  \Return{false}\;
 }{
  \Return{true}\;
 }}
 \caption{Checking \textsc{Causal consistency}}
 \label{ccalgo:1}
\end{algorithm}

Algorithm~\ref{ccalgo:1} lists our procedure for checking CC. As explained above, $\CO$ is initially set to $\so\cup \wro$, and then, it is saturated with other ordering constraints implied by non-vacuous instantiations of the axiom $\mathsf{Causal}$ (where the left-hand side of the implication evaluates to true). The algorithms concerning RC and RA are defined in a similar way by essentially changing the test at line 6 so that it corresponds to the left-hand side of the implication in the corresponding axiom. Algorithm~\ref{ccalgo:1} can be rewritten as a Datalog program containing straightforward Datalog rules for computing transitive closures and relation composition, and a rule of the form\footnote{We write Datalog rules using a standard notation $\mathit{head}\text{ :- }\mathit{body}$ where $\mathit{head}$ is a relational atom (written as $\tup{a,b}\in R$ where $a$, $b$ are elements and $R$ a binary relation) and $\mathit{body}$ is a list of relational atoms.}
\begin{align*}
\tup{\tr_2, \tr_1} \in \CO \text{ :- } \tr_1\neq\tr_2, \tup{\tr_1,\tr_3}\in \wro[\xvar], \tup{\tr_2,\tr_3}\in (\so\cup\wro)^+
\end{align*}
to represent the $\mathsf{Causal}$ axiom.
%specification using the predicate from line 6-7; $\forall \tr_1, \tr_2, \tr_3.~uniq(\tr_1, \tr_2, \tr_3)$~\footnote{$uniq(\tr_1, \tr_2, \tr_3) = \tr_1 \neq \tr_2 \land \tr_2 \neq \tr_3 \land \tr_3 \neq \tr_1$}~$\land \tup{\tr_1,\tr_3}\in \wro[\xvar]\land \tup{\tr_2,\tr_3}\in (\so\cup\wro)^+ \implies \tup{\tr_2, \tr_1} \in \CO$. 
The following is a consequence of the fact that these algorithms run in polynomial time (or equivalently, the Datalog programs can be evaluated in polynomial time over a database that contains the $\wro$ and $\so$ relations in a given history).%(see Appendix~\ref{app:sec:general}).
%The algorithms for checking \textsc{Read Committed}, \textsc{Read Atomic}, and \textsc{Causal Consistency} are given in Algorithm~\ref{} (TODO). 

\begin{theorem}
For any criterion $C \in \{\emph{\textsc{Read Committed}}, \emph{\textsc{Read Atomic}}, \emph{\textsc{Causal consistency}} \}$, 
the problem of checking whether a given history satisfies $C$ % \emph{\textsc{Read Committed}}, \emph{\textsc{Read Atomic}}, or \emph{\textsc{Causal consistency}} 
is polynomial time.
\end{theorem}

On the other hand, checking PC, SI, and SER is NP-complete in general. We show this using a reduction from boolean satisfiability (SAT) that covers uniformly all the three cases. In the case of SER, it provides a new proof of the NP-completeness result by \citet{DBLP:journals/jacm/Papadimitriou79b} which uses a reduction from the so-called \emph{non-circular} SAT and which cannot be extended to PC and SI.

\begin{theorem}
 \label{npcproof:0}
\hspace{-2mm}
For any criterion $C \hspace{-.7mm}\in\hspace{-.7mm} \{\emph{\textsc{Prefix Consistency}},\hspace{-.5mm}\emph{\textsc{Snapshot Isolation}},\hspace{-.5mm}\emph{\textsc{Serializability}} \}$
the problem of checking whether a given history satisfies $C$ 
is NP-complete.
%The problem of checking whether a history satisfies any criteria between \emph{\textsc{Prefix Consistency}} and \emph{\textsc{Serializability}}, is NP-complete.
\end{theorem}
%!TEX root = draft.tex
\tikzset{transaction state /.style={draw=black!0}}
\begin{proof}
 Given a history, any of these three criteria %between \textsf{Prefix Consistency} and \textsf{Serializability} 
 can be checked by guessing a total commit order on its transactions and verifying whether it satisfies the corresponding axioms. This shows that the problem is in NP.
 
 To show NP-hardness, we define a reduction from boolean satisfiability. Therefore, let $\varphi=D_1\land\ldots\land D_m$ be a \textsf{CNF} formula over the boolean variables $x_1, \ldots, x_n$ where each $D_i$ is a disjunctive clause with $m_i$ literals.  %we reduce it to a history in polynominal time.
 Let $\lambda_{ij}$ denote the $j$-th literal of $D_i$. 
 
We construct a history $h_\varphi$ such that $\varphi$ is satisfiable if and only if $h_\varphi$ satisfies PC, SI, or SER. Since $SER\implies SI\implies PC$, we show that (1) if $h_\varphi$ satisfies PC, then $\varphi$ is satisfiable, and (2) if $\varphi$ is satisfiable, then $h_\varphi$ satisfies SER.

The main idea of the construction is to represent truth values of each of the variables and literals in $\varphi$ with the polarity of the commit order between corresponding transaction pairs.
For each variable $x_k$, $h_\varphi$ contains a pair of transactions $a_k$ and $b_k$, and for each literal $\lambda_{ij}$, $h_\varphi$ contains a set of transactions $w_{ij}$, $y_{ij}$ and $z_{ij}$~\footnote{We assume that the transactions $a_k$ and $b_k$ associated to a variable $x_k$ are distinct and different from the transactions associated to another variable $x_{k'}\neq x_k$ or to a literal $\lambda_{ij}$. Similarly, for the transactions $w_{ij}$, $y_{ij}$ and $z_{ij}$ associated to a literal $\lambda_{ij}$.}. We want to have that $x_k$ is false if and only if $\tup{a_k, b_k} \in \CO$, and $\lambda_{ij}$ is false if and only if $\tup{y_{ij}, z_{ij}} \in \CO$ (the transaction $w_{ij}$ is used to "synchronize" the truth value of the literals with that of the variables, which is explained later).

The history $h_\varphi$ should ensure that the $\co$ ordering constraints corresponding to an assignment that falsifies the formula (\ie one of its clauses) form a cycle. 
%If for an assignment does not satisfy a clause, we want the corresponding relations to create a cycle in $\CO$. 
To achieve that, we add all pairs $\tup{z_{ij}, y_{i,(j+1)\% m_i}}$ in the session order $\so$. An unsatisfied clause $D_i$, \ie every $\lambda_{ij}$ is false, leads to a cycle of the form $y_{i1} \xrightarrow{\CO} z_{i1} \xrightarrow{\so} y_{i2} \xrightarrow{\co} z_{i2} \cdots z_{i m_i} \xrightarrow{\so} y_{i1}$.

\begin{figure}[t]
 \resizebox{\textwidth}{!}{
  \begin{subfigure}{.55\textwidth}
   \centering
   \begin{tikzpicture}[->,>=stealth',shorten >=1pt,auto,node distance=4cm,
     semithick, transform shape]
     % \node[transaction state,label={$\writeVar{}{v_j}$}] at (5, 2) (a_j)                                    {$a_j$};
    \node[transaction state] at (5, 2) (a_j)                                    {$a_k$};
    % \node[transaction state,label=below:{$\writeVar{}{v_j}$}] at (5,0)       (b_j)           {$b_j$};
    \node[transaction state] at (5,0)       (b_j)           {$b_k$};
    % \node[transaction state] at (-3,3)        (c_j)                     {$c_j$};
    \node[transaction state] at (4,0)        (w_ik)                     {$w_{ij}$};
    \node[transaction state,label={$\writeVar{}{v_{ij}}$}] at (2,2) (y_ik) {$y_{ij}$};
    \node[transaction state] at (2,0) (z_ik) {$z_{ij}$};
    \node[transaction state] at (.5,2.5) (z_ik1) {$z_{i,j-1}$};
    \node[transaction state] at (.5, -.5) (y_ik1) {$y_{i,j+1}$};
    
    \path (a_j) edge[dotted, red, bend left=25] node {\CO} (b_j)
    (b_j) edge[dashed, red, bend left=25] node {\CO} (a_j);
    % (b_j) edge[bend left] node {\WR} (c_j)
    % (c_j) edge[dashed, red, bend left] node {\RW} (a_j);
    
    \path (y_ik) edge[dotted, blue, bend left=25] node {\CO} (z_ik)
    (z_ik) edge[dashed, blue, bend left=25] node {\CO} (y_ik)
    (z_ik) edge node {$\wro[v_{ij}]$} (w_ik);
    % (w_ik) edge[dashed, blue, bend left] node {\RW} (y_ik);
    
    \path (y_ik) edge node {$\so$} (a_j)
    (b_j) edge node[above] {$\so$} (w_ik)
    (z_ik1) edge node[below] {$\so$} (y_ik)
    (z_ik) edge node {$\so$} (y_ik1);
   \end{tikzpicture}
   \caption{$\lambda_{ij} = x_k$}
   \label{fig:lambda_i_k_x_j_}
  \end{subfigure}
  \begin{subfigure}{.55\textwidth}
   \centering
   \begin{tikzpicture}[->,>=stealth',shorten >=1pt,auto,node distance=4cm,
     semithick, transform shape]
     % \node[transaction state,label={$\writeVar{}{v_j}$}] at (5, 2) (b_j)                                    {$b_j$};
    \node[transaction state] at (5, 2) (b_j)                                    {$b_k$};
    % \node[transaction state,label=below:{$\writeVar{}{v_j}$}] at (5,0)       (a_j)           {$a_j$};
    \node[transaction state] at (5,0)       (a_j)           {$a_k$};
    % \node[transaction state] at (-3,3)        (c_j)                     {$c_j$};
    \node[transaction state] at (4,0)        (w_ik)                     {$w_{ij}$};
    \node[transaction state,label={$\writeVar{}{v_{ij}}$}] at (2,2) (y_ik) {$y_{ij}$};
    \node[transaction state] at (2,0) (z_ik) {$z_{ij}$};
    \node[transaction state] at (.5,2.5) (z_ik1) {$z_{i,j-1}$};
    \node[transaction state] at (.5, -.5) (y_ik1) {$y_{i,j+1}$};
    
    \path (a_j) edge[dashed, red, bend left=25] node {\CO} (b_j)
    (b_j) edge[dotted, red, bend left=25] node {\CO} (a_j);
    % (b_j) edge[bend right] node {\WR} (c_j)
    % (c_j) edge[dashed, red, bend right] node {\RW} (a_j);
    
    \path (y_ik) edge[dotted, blue, bend left=25] node {\CO} (z_ik)
    (z_ik) edge[dashed, blue, bend left=25] node {\CO} (y_ik)
    (z_ik) edge node {$\wro[v_{ij}]$} (w_ik);
    % (w_ik) edge[dashed, blue, bend left] node {\RW} (y_ik);
    
    \path (y_ik) edge node {$\so$} (b_j)
    (a_j) edge node[above] {$\so$} (w_ik)
    (z_ik1) edge node[below] {$\so$} (y_ik)
    (z_ik) edge node {$\so$} (y_ik1);
   \end{tikzpicture}
   \caption{$\lambda_{ij} = \neg x_k$}
   \label{fig:lambda_i_k_n_x_j_}
  \end{subfigure}
 }
  \vspace{-2mm}
 \caption{Sub-histories included in $h_\varphi$ for each literal $\lambda_{ij}$ and variable $x_k$.}
 \vspace{-3mm}
\end{figure}
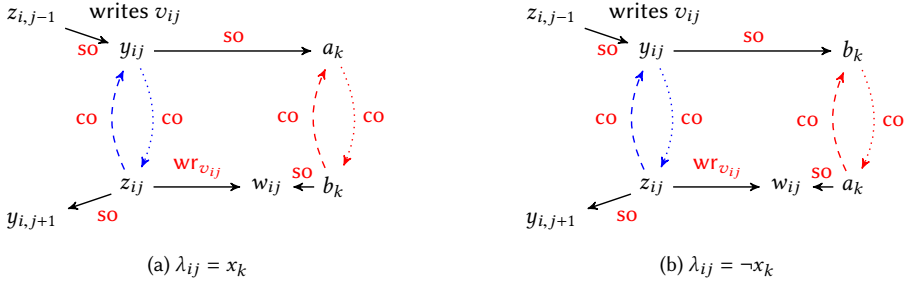

The most complicated part of the construction is to ensure the consistency between the truth value of the literals and the truth value of the variables, e.g., $\lambda_{ij} = x_k$ is false iff $x_k$ is false. We use special sub-histories to enforce that if history $h_\varphi$ satisfies PC (i.e., the axiom $\mathsf{Prefix}$), 
%any criteria between \textsf{Prefix Consistency} and \textsf{Serializability}, 
then there exists a commit order $\CO$ such that $\tup{h_\varphi,\co}$ satisfies $\mathsf{Prefix}$ (Figure~\ref{pre_def}) and:
\begin{align}
\mbox{$\tup{a_k, b_k} \in \CO$ iff $\tup{y_{ij}, z_{ij}} \in \CO$ when $\lambda_{ij} = x_k$, and }\label{eq:iffs}\\
\mbox{$\tup{a_k, b_k} \in \CO$ iff $\tup{z_{ij}, y_{ij}} \in \CO$ when $\lambda_{ij} = \neg x_k$}.\nonumber
\end{align}
%provided that the history $h_\varphi$ satisfies any criteria between \textsf{Prefix Consistency} and \textsf{Serializability}.
%
%We will particularly show if $\varphi_\hist$ satisfies \textsf{Prefix Consistency}(weakest criteria), then the $\varphi$ is satisfiable and if $\varphi$ is satisfiable, then $\varphi_\hist$ satisfies \textsf{Serializability}(strongest criteria).
%
%Now the truth value of each $\lambda_{ik} = x_j$ or $\neg x_j$ has a consistent truth value according to the truth value of $x_j$. We use special subhistories to enforce the equivalent consistency between $\tup{y_{ik}, z_{ik}}$ and $\tup{a_j, b_j}$.
%
%These sub-histories are shown in Figure~\ref{fig:lambda_i_k_x_j_} and Figure~\ref{fig:lambda_i_k_n_x_j_}. 
%show the special subhistories for each $\lambda_{ik} = x_j$ and $\lambda_{ik} = \neg x_j$ respectively. 
%
Figure~\ref{fig:lambda_i_k_x_j_} shows the sub-history associated to a positive literal $\lambda_{ij} = x_k$ while Figure~\ref{fig:lambda_i_k_n_x_j_} shows the case of a negative literal $\lambda_{ij} = \neg x_k$.

For a positive literal $\lambda_{ij} = x_k$ (Figure~\ref{fig:lambda_i_k_x_j_}), (1) we enrich session order with the pairs $\tup{y_{ij}, a_k}$ and $\tup{b_k, w_{ij}}$, (2) we include writes to a variable $v_{ij}$ in the transactions $y_{ij}$ and $z_{ij}$, and (3) we make $w_{ij}$ read from $z_{ij}$, $\ie$ $\tup{z_{ij}, w_{ij}}\in \wro_{v_{ij}}$. The case of a negative literal is similar, switching the roles of $a_k$ and $b_k$.

This construction ensures that if the $\co$ goes downwards on the right-hand side ($\tup{a_k, b_k} \in \co$ in the case of a positive literal, and $\tup{b_k, a_k} \in \co$ in the case of a negative literal), then it must also go downwards on the left-hand side ($\tup{y_{ij}, z_{ij}} \in \co$) to satisfy $\mathsf{Prefix}$. For instance, in the case of a positive literal, note that if $\tup{a_k, b_k} \in \co$, then $\tup{y_{ij}, w_{ij}} \in \so \circ \co \circ \so$. Therefore, for every commit order $\CO$ such that $\tup{h_\varphi,\co}$ satisfies $\mathsf{Prefix}$, $\tup{a_k, b_k} \in \CO$ implies $\tup{y_{ij}, z_{ij}} \in \co$. Indeed, if $\tup{a_k, b_k} \in \CO$, instantiating the $\mathsf{Prefix}$ axiom  where $y_{ij}$ plays the role of $t_2$, $z_{ij}$ plays the role of $t_1$, and $w_{ij}$ plays the role of $t_3$, we obtain that $\tup{y_{ij}, z_{ij}} \in \co$. 

In contrast, when the $\co$ goes upwards on the right-hand side ($\tup{b_k, a_k} \in \co$ in the case of a positive literal, and $\tup{a_k, b_k} \in \co$ in the case of a negative literal) then it imposes no constraint on the direction of $\co$ on the left-hand side. Therefore, any commit order $\co$ satisfying $\mathsf{Prefix}$ that goes upwards on the right-hand side (e.g., $\tup{b_k, a_k} \in \co$ in the case of a positive literal) and downwards on the left-hand side ($\tup{y_{ij}, z_{ij}} \in \co$) in some sub-history (associated to some literal), thereby contradicting Property (\ref{eq:iffs}), can be modified into another commit order satisfying $\mathsf{Prefix}$ that goes upwards on the left-hand side as well. Formally, let $\CO$ be a commit order such that $\tup{h_\varphi,\co}$ satisfies $\mathsf{Prefix}$ and 
\begin{align*}
\tup{b_k, a_k} \in \CO \land \tup{y_{ij}, z_{ij}} \in \CO
\end{align*}
for some literal $\lambda_{ij} = x_k$ (the case of negative literals can be handled in a similar manner). Let $\CO_1$ be the restriction of $\CO$ on the set of tuples 
\begin{align*}
\{\tup{a_{k'}, b_{k'}}, \tup{b_{k'}, a_{k'}} | 1\leq k'\leq n\} \cup \{\tup{y_{i'j'}, z_{i'j'}}, \tup{z_{i'j'}, y_{i'j'}} | \text{for each }i', j'\} \cup \so \cup \wro. 
\end{align*}
Since $\CO_1 \subseteq \CO$, we have that $\CO_1$ is acyclic. 
%Let $\lambda_{ij} = x_k$ be a literal such that $\tup{y_{ij}, z_{ij}} \in \CO_1$ and $\tup{b_k, a_k} \in \CO_1$. 
Let $\CO_2$ be a relation obtained from $\CO_1$ by flipping the order between $y_{ij}$ and $z_{ij}$ (\ie $\CO_2 = \CO_1 \setminus \{ \tup{y_{ij}, z_{ij}} \} \cup \{ \tup{z_{ij}, y_{ij}} \}$). This flipping does not introduce any cycle because $\CO_2$ contains no path ending in $z_{ij}$ (see Fig~\ref{fig:lambda_i_k_x_j_}). Also, $\CO_2$ still satisfies the $\mathsf{Prefix}$ axiom (since $\tup{b_k, a_k} \in \CO_2$ there is no path from $y_{ij}$ to $w_{ij}$ satisfying the constraints in the $\mathsf{Prefix}$ axiom). Since $\CO_2$ is acyclic, it can be extended to a total commit order $\CO_3$ that satisfies $\mathsf{Prefix}$. This is a consequence of the following lemma whose proof follows easily from definitions (the part of this lemma concerning $\mathsf{Serializability}$ will be used later).%Moreover, $\CO_3$ has strictly less literals $\lambda_{ij} = x_k$ satisfying $\tup{y_{ij}, z_{ij}} \in \CO \land \tup{b_k, a_k} \in \CO$ than $\CO$, which contradicts the hypothesis.

\begin{lemma}\label{lem:extensions}
Let $\co$ be an acyclic relation that includes $\so\cup\wro$, $\tup{a_{k}, b_{k}}$ or $\tup{b_{k}, a_{k}}$, for each $k$, and $\tup{y_{ij},z_{ij}}$ or $\tup{z_{ij},y_{ij}}$, for each $i$, $j$. For each axiom $A\in \{\mathsf{Prefix}, \mathsf{Serializability}\}$, if $\tup{h_\varphi,\co}$ satisfies $A$, then there exists a total commit order $\co'$ such that $\co\subseteq \co'$ and $\tup{h_\varphi,\co'}$ satisfies $A$.
\end{lemma}

Therefore, $\tup{h_\varphi,\co_3}$ satisfies $\mathsf{Prefix}$, and $\tup{b_k, a_k} \in \CO_3 \land \tup{z_{ij}, y_{ij}} \in \CO_3$ ($\co_3$ goes upwards on both sides of a sub-history like in Figure~\ref{fig:lambda_i_k_x_j_}). This transformation can be applied iteratively until obtaining a commit order that satisfies both $\mathsf{Prefix}$ and Property (\ref{eq:iffs}).

Next, we complete the correctness proof of this reduction. % \ie $\varphi$ is satisfiable if and only if $h_\varphi$ satisfies any criteria between PC and SER. 
For the ``if'' direction, if $h_\varphi$ satisfies PC, then there exists a total commit order $\co$ between the transactions described above, which together with $h_\varphi$ satisfies $\mathsf{Prefix}$. The assignment of the variables $x_k$ explained above (defined by the $\co$ order between $a_k$ and $b_k$, for each $k$) satisfies the formula $\varphi$ since there exists no cycle between the transactions $y_{ij}$ and $z_{ij}$, which implies that for each clause $D_i$, there exists a $j$ such that $\tup{y_{ij}, z_{ij}} \not \in \CO$ which means that $\lambda_{ij}$ is satisfied.
For the ``only-if'' direction, let $\gamma$ be a satisfying assignment for $\varphi$. Also, let $\CO'$ be a binary relation that includes $\so$ and $\wro$ such that if $\gamma(x_k)=\mathit{false}$, then $\tup{a_k, b_k} \in \CO'$, $\tup{y_{ij}, z_{ij}} \in \CO'$ for each $\lambda_{ij} = x_k$, and $\tup{z_{ij}, y_{ij}} \in \CO'$ for each $\lambda_{ij} = \neg x_k$, and if $\gamma(x_k)=\mathit{true}$, then $\tup{b_k, a_k} \in \CO'$, $\tup{z_{ij}, y_{ij}} \in \CO'$ for each $\lambda_{ij} = x_k$, and $\tup{y_{ij}, z_{ij}} \in \CO'$ for each $\lambda_{ij} = \neg x_k$. Note that $\CO'$ is acyclic: no cycle can contain $w_{ij}$ because $w_{ij}$ has no ``outgoing'' dependency (\ie $\CO'$ contains no pair with $w_{ij}$ as a first component), there is no cycle including some pair of transactions $a_k$, $b_k$ and some pair $y_{ij}$, $z_{ij}$ because there is no way to reach $y_{ij}$ or $z_{ij}$ from $a_k$ or  $b_k$, there is no cycle including only transactions $a_k$ and $b_k$ because $a_{k_1}$ and $b_{k_1}$ are not related to $a_{k_2}$ and $b_{k_2}$, for $k_1\neq k_2$, there is no cycle including transactions $y_{i_1,j_1}$, $z_{i_1,j_1}$ and $y_{i_2,j_2}$, $z_{i_2,j_2}$ for $i_1\neq i_2$ since these are disconnected as well, and finally, there is no cycle including only transactions $y_{ij}$ and $z_{ij}$, for a fixed $i$, because $\varphi$ is satisfiable. By Lemma~\ref{lem:extensions}, the acyclic relation $\co'$ can be extended to a total commit order $\co$ which together with $h_\varphi$ satisfies the $\mathsf{Serializability}$ axiom. Therefore, $h_\varphi$ satisfies SER.
%
%TODO I STOPPED HERE
%
%Also, there is no cycle in $\co$. So there is no cycle of the form $y_{i1} \xrightarrow{\CO} z_{i1} \xrightarrow{\so} y_{i2} \cdots z_{ik} \xrightarrow{\so} y_{i1}$ for any $i$. So each clause $i$ has a $k$ such that $\tup{y_{ik}, z_{ik}} \not\in \co$ which implies there exists an assignment(given by $\tup{a_j, b_j}$) for which each clause is satisfied. Thus $\varphi$ is satisfiable.
%
%For the other direction, we show, there no other kind of cycle is possible in $h_\varphi$, when $\tup{y_{ik}, z_{ik}}$ and $\tup{a_j, b_j}$ pairs are fixed.
%\begin{itemize}
% \item First note, no cycle can contain $w_{ik}$ because it does not have any outgoing relation(TODO better word). 
% \item Also it is easy to see, there is no cycle involving both $a_j, b_j, y_{ik}, z_{ik}$ because there is no way to reach $y_{ik}, z_{ik}$ from any $a_j, b_j$. 
% \item $a_j, b_j$ can not have cycles because each of $a_{j1}, b_{j1}$ and $a_{j2}, b_{j2}$ are disconnected. 
% \item Each of $y_{i1k1}, z_{i1k1}$ and $y_{i2k2}, z_{i2k2}$ are also disconnected.
%\end{itemize}
%
%So only possible cycle is in $y_{ik}, z_{ik}$ for each clause $i$. But, given an satisfying assignment of $\phi$, we can set $\tup{a_j, b_j}$ and $\tup{y_{ik}, z_{ik}}$ accordingly. Since every clause $i$ is satisfied, there is at least one $\tup{y_{ik}, z_{ik}} \not\in \co$ \ie there is no cycle in $y_{ik}, z_{ik}$. Hence, we can extend that strict partial order to a total order $\CO$. Thus, $h_\varphi$ is PC, SI or SER.
\end{proof}

\section{Checking Consistency of Bounded-Width Histories}\label{sec:bounded_width}

% Given a history it takes polynomial time to check for \textsc{Int} and \textsf{Ext}. Transactions, in a read atomic consistent history, can be assumed to read and write each variable at max once - just take the first read and last write of each variable. So for the rest of the paper, we will assume this.

In this section, we show that checking prefix consistency, snapshot isolation, and serializability becomes polynomial time under the assumption that the \emph{width} of the given history, i.e., the maximum number of mutually-unordered transactions w.r.t. the session order, is bounded by a fixed constant. If we consider the standard case where the session order is a union of transaction sequences (modulo the fictitious transaction writing the initial values), i.e., a set of sessions, then the width of the history is the number of sessions. We start by presenting an algorithm for checking serializability which is polynomial time when the width is bounded by a fixed constant. In general, the asymptotic complexity of this algorithm is exponential in the width of the history, but this worst-case behavior is not exercised in practice as shown in Section~\ref{sec:exp}. Then, we prove that checking prefix consistency and snapshot isolation can be reduced in polynomial time to the problem of checking serializability. 

% \subsection{Serialization and Linearization}
\subsection{Checking Serializability}\label{ssec:ser_checking}

We present an algorithm for checking serializability of a given history which constructs a valid commit order (satisfying $\mathsf{Serialization}$), if any, by 
``linearizing'' transactions one by one in an order consistent with the session order. At any time, the set of already linearized transactions is uniquely determined by an antichain of the session order (i.e., a set of mutually-unordered transactions w.r.t. $\so$), and the next transaction to linearize is chosen among the immediate $\so$ successors of the transactions in this antichain. The crux of the algorithm is that the next transaction to linearize can be chosen such that it does not produce violations of $\mathsf{Serialization}$ in a way that  does not depend on the order between the already linearized transactions. Therefore, the algorithm can be seen as a search in the space of $\so$ antichains. If the width of the history is bounded (by a fixed constant), then the number of possible $\so$ antichains is polynomial in the size of the history, which implies that the search can be done in polynomial time.

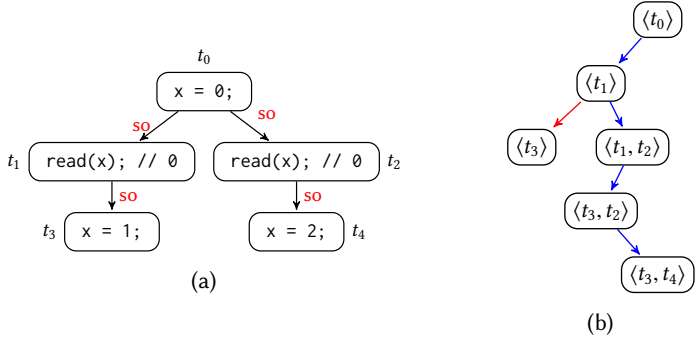
\begin{figure}
  
  \begin{subfigure}{.4\textwidth}
  \resizebox{\textwidth}{!}{
  \begin{tikzpicture}[->,>=stealth',shorten >=1pt,auto,node distance=3cm,
   semithick, transform shape]
   \node[draw, rounded corners=2mm,label={$\tr_0$}] (t1) at (0, 1.2) {\begin{tabular}{l} \texttt{x = 0;} \end{tabular}};
   \node[draw, rounded corners=2mm,label=left:{$\tr_1$}] (t2r) at (-1.6, 0) {\begin{tabular}{l} \texttt{read(x); // 0} \end{tabular}};
   \node[draw, rounded corners=2mm,label=left:{$\tr_3$}] (t2w) at (-1.6, -1.2) {\begin{tabular}{l}  \texttt{x = 1;} \end{tabular}};
   \node[draw, rounded corners=2mm,label=right:{$\tr_2$}] (t3r) at (1.6, 0) {\begin{tabular}{l} \texttt{read(x); // 0} \end{tabular}};
   \node[draw, rounded corners=2mm,label=right:{$\tr_4$}] (t3w) at (1.6, -1.2) {\begin{tabular}{l} \texttt{x = 2;} \end{tabular}};
  \path (t2r) edge node {$\so$} (t2w);
  \path (t3r) edge node {$\so$} (t3w);
  \path (t1) edge node[left] {$\so$} (t2r);
  \path (t1) edge node {$\so$} (t3r);
  \end{tikzpicture}  
  }
   \caption{ }
   \label{ser_algo_example:1}
  \end{subfigure}
  \hspace{1cm}
   \begin{subfigure}{.19\textwidth}
    \resizebox{\textwidth}{!}{
     \begin{tikzpicture}[->,>=stealth',shorten >=1pt,auto,node distance=3cm,
       semithick, transform shape]
      \node[draw, rounded corners=2mm] (v00) at (0, 0) {$\left\langle \tr_0 \right\rangle$};
      \node[draw, rounded corners=2mm] (v10) at (-.9, -1) {$\left\langle \tr_1  \right\rangle$};
      \node[draw, rounded corners=2mm] (v20) at (-2, -2) {$\left\langle \tr_3 \right\rangle$};
      \node[draw, rounded corners=2mm] (v11) at (-.4, -2) {$\left\langle \tr_1, \tr_2  \right\rangle$};
      \node[draw, rounded corners=2mm] (v21) at (-.9, -3) {$\left\langle \tr_3, \tr_2  \right\rangle$};
      \node[draw, rounded corners=2mm] (v22) at (0, -4) {$\left\langle \tr_3, \tr_4  \right\rangle$};
      % \pic (v10) at (-3.5, -4)    {ser_history=v10};
      % \pic (v20) at (-7, -8)    {ser_history=v20};
      % \pic (v11) at (0, -8)    {ser_history=v11};
      % \pic (v21) at (-3.5, -12)    {ser_history=v21};
      % \pic (v22) at (0, -16)    {ser_history=v22};
      \path (v00) edge[blue] node {} (v10); 
      \path (v10) edge[red] node {} (v20); 
      \path (v10) edge[blue] node {} (v11); 
      \path (v11) edge[blue] node {} (v21); 
      \path (v21) edge[blue] node {} (v22);
      
      % \draw[-, blue, dashed] plot[smooth, tension=2] coordinates { (-1.5,1.5) (0, .5) (1.5, 1.5)};
      % \draw[-, blue, dashed] plot coordinates { (-1.2,1.7) (-1.2, .7) (1.2, .7) (1.2, 1.7)};
      % \draw[-, blue, dashed] plot coordinates { (-7,-2.3) (-7,-4.5) (-3.5,-4.5) (-3.5,-3.5) (-1.5,-3.5) (-1.5,-2.3) };
      % \draw[-, blue, dashed] plot coordinates { (-3.5,-6.3) (-3.5,-8.5) (3.4,-8.5) (3.4,-6.3) };
      % % \draw[-, blue, dashed] plot coordinates { (-10.5,-6.5) (-10.5,-10) (-7, -10) (-7, -7.5) (-3.6,-7.5) (-3.6,-6.5) };
      % \draw[-, blue, dashed] plot coordinates { (-10.5,-6.3) (-10.5,-8.5) };
      % \draw[-, red, dashed] plot coordinates { (-10.5,-8.5) (-10.5,-10) (-7,-10) (-7, -8.5) };
      % \draw[-, blue, dashed] plot coordinates { (-7, -8.5) (-7, -7.5) (-3.6,-7.5) (-3.6,-6.3) };
      % \draw[-, blue, dashed] plot coordinates { (-7,-10.3) (-7,-14) (-3.5, -14) (-3.5, -12.5) (-.1,-12.5) (-.1,-10.3) };
      % \draw[-, blue, dashed] plot coordinates { (-3.5,-14.3) (-3.5,-17.9) (3.4, -17.9) (3.4,-14.3) };
      % \draw[-, blue, dashed] plot coordinates { (-10.5,-6.5) (-10.5,-8.5) (-3.5,-8.5) (-3.5,-6.5) };
      
     \end{tikzpicture}  
    }
    \caption{ }
    \label{ser_algo_example:3}
   \end{subfigure}
   \vspace{-4mm}
  \caption{Applying \textsf{checkSER} on the serializable history on the left. The right part pictures a search for valid extensions of serializable prefixes, represented by their boundaries. The red arrow means that the search is blocked (the prefix at the target is not a valid extension), while blue arrows mean that the search continues.}
  \label{ser_algo_example}
  \vspace{-4mm}
\end{figure}

A \emph{prefix} of a history $\hist = \tup{T, \so, \wro}$ is a set of transactions $T'\subseteq T$ such that all the $\so$ predecessors of transactions in $T'$ are also in $T'$, i.e., $\forall \tr\in T.\ \so^{-1}(\tr)\in T$. A prefix $T'$ is uniquely determined by the set of transactions in $T'$ which are maximal w.r.t. $\so$. This set of transactions forms an \emph{antichain} of $\so$, i.e., any two elements in this set are incomparable w.r.t. $\so$. Given an antichain $\{\tr_1,\ldots,\tr_n\}$ of $\so$, we say that $\{\tr_1,\ldots,\tr_n\}$ is the \emph{boundary} of the prefix $T'=\{t:\exists i.\ \tup{t,t_i}\in \so\lor t = t_i\}$. For instance, given the history in Figure~\ref{ser_algo_example:1}, the set of transactions $\{\tr_0,\tr_1,\tr_2\}$ is a prefix with boundary $\{\tr_1,\tr_2\}$ (the latter is an antichain of the session order).

%TODO SHOW THE ALGORITHM RUNNING ON A HISTORY: GIVE A SERIALIZABLE HISTORY, A GRAPH WHERE NODES ARE ANTICHAINS AND EDGES ARE "VALID EXTENSION" (WE SHOULD HAVE SOME PATH THAT BLOCKS, AN ANTICHAIN WHICH CANNOT BE EXTENDED). FOR HERE, CIRCLE A PREFIX ON THAT HISTORY AND REFER TO IT.

A prefix $T'$ of a history $\hist$ is called \emph{serializable} iff there exists a \emph{partial} commit order $\co$ on the transactions in $\hist$ such that the following hold:
\begin{itemize}
  \item $\co$ does not contradict the session order and the write-read relation in $\hist$, i.e., $ \wro \cup \so\cup \co$ is acyclic, 
 \item $\co$ is a total order on transactions in $T'$, 
 \item $\co$ orders transactions in $T'$ before transactions in $T\setminus T'$, i.e., $\tup{\tr_1,\tr_2}\in \co$ for every $\tr_1\in T'$ and $\tr_2\in T\setminus T'$,  
 \item $\co$ does not order any two transactions $\tr_1, \tr_2 \not\in T'$
 \item the history $\hist$ along with the commit order $\co$ satisfies the axiom defining serializability, i.e., $\tup{\hist, \co} \models \mathsf{Serialization}$.
\end{itemize}

For the history in Figure~\ref{ser_algo_example:1}, the prefix $\{\tr_0,\tr_1,\tr_2\}$ is serializable since there exists a partial commit order $\co$ which orders 
$\tr_0$, $\tr_1$, $\tr_2$ in this order, and both $\tr_1$ and $\tr_2$ before $\tr_3$ and $\tr_4$. The axiom $\mathsf{Serialization}$ is satisfied trivially, since the prefix contains a single transaction writing $\xvar$ and all the transactions outside of the prefix do not read $\xvar$.

A prefix $T'\uplus\{\tr\}$ of $\hist$ is called a \emph{valid extension} of a serializable prefix $T'$ of $\hist$~\footnote{We assume that $\tr\not\in T'$ which is implied by the use of the disjoint union $\uplus$.}, denoted by $T'\vartriangleright T'\uplus\{\tr\}$%~\footnote{Since $T'\cup\{\tr\}$ is required to be a prefix, $\tr$ is necessarily an $\so$ successor of some transaction in the boundary of $T'$ or an $\so$ successor of the transaction writing the initial values.}, 
if:
\begin{itemize}
 \item $\tr$ does not read from a transaction outside of $T'$, i.e., for every $\tr'\in T\setminus T'$, $\tup{\tr',\tr}\not\in\wro$, and
 \item for every variable $\xvar$ written by $\tr$, there exists no transaction $\tr_2\neq \tr$ outside of $T'$ which reads a value of $\xvar$ written by a transaction $\tr_1$ in $T'$, i.e., for every $\xvar$ written by $\tr$ and every $\tr_1\in T'$ and $\tr_2\in T\setminus (T'\uplus\{\tr\})$, $\tup{\tr_1,\tr_2}\not\in\wro$.
\end{itemize}

For the history in Figure~\ref{ser_algo_example:1}, we have $\{\tr_0,\tr_1\}\vartriangleright \{\tr_0,\tr_1\}\uplus\{\tr_2\}$ because $\tr_2$ reads from $\tr_0$ and it does not write any variable. On the other hand $\{\tr_0,\tr_1\}\not\vartriangleright \{\tr_0,\tr_1\}\uplus\{\tr_3\}$ because $\tr_3$ writes $\xvar$ and the transaction $\tr_2$, outside of this prefix, reads from the transaction $\tr_0$ included in the prefix.

Let $\vartriangleright^*$ denote the reflexive and transitive closure of $\vartriangleright$.

The following lemma is essential in proving that iterative valid extensions of the initial empty prefix can be used to show that a given history is serializable.

\begin{lemma}
 \label{serpref:1}
 % $\mathcal{P}$ is a serializable prefix and $\tr \not\in \mathcal{P}$. $\mathcal{P} \cup \{\tr\}$ is also a serializable prefix if $\tup{\mathcal{P}, \tr}$ satisfies \textsc{SerializableStep}.
 For a serializable prefix $T'$ of a history $h$, a prefix $T'\uplus\{\tr\}$ is serializable if it is a valid extension of $T'$.
\end{lemma}
%$\textsc{SerializableStep}(\mathcal{P}, \tr) = 
% \forall \tr' \not\in \mathcal{P}, \tup{\tr', \tr} \not\in \wro \cup \so \land
% \forall \xvar \in \vars{\hist}.  \forall \tr_1 \in \mathcal{P}. \forall \tr_2 \not\in \mathcal{P}. \tup{\tr_1, \tr_2} \in \wro[\xvar] \land \writeVar{\tr}{x} \Rightarrow \tr_2 = \tr$
  %!TEX root = draft.tex
\begin{proof}
\renewcommand{\qedsymbol}{}
Let $\co'$ be the partial commit order for $T'$ which satisfies the serializable prefix conditions. We extend $\co'$ to a partial order $\co = \co' \cup \{ \tup{\tr,\tr'} | \tr' \not\in T'\uplus\{\tr'\} \}$. We show that $\tup{\hist, \co} \models \mathsf{Serialization}$. The other conditions for $T'\uplus\{t\}$ being a serializable prefix are satisfied trivially by $\co$. 

Assume by contradiction that $\tup{\hist, \co}$ does not satisfy the axiom $\mathsf{Serialization}$. Then, there exists $\tr_1, \tr_2, \tr_3$, $\xvar \in \vars{\hist}$ s.t. $\tup{\tr_1, \tr_3} \in \wro[\xvar]$ and $\tr_2$ writes on $\xvar$ and $\tup{\tr_1, \tr_2}, \tup{\tr_2, \tr_3} \in \co$. Since $\tup{\hist,\co'}$ satisfies this axiom, at least one of these two $\co$ ordering constraints are of the form $\tup{\tr, \tr'}$ where $\tr' \not\in T' \uplus \{\tr\}$:
       \begin{itemize}
        \item the case $\tr_1 = \tr$ and $\tr_2 \not\in T' \uplus \{\tr\}$ is not possible because $\co'$ contains no pair of the form $\tup{\tr', \_} \in \co'$ with $\tr' \not\in T'$ (recall that $\tup{\tr_2, \tr_3}$ should be also included in $\co$). 
        \item If $\tr_2 = \tr$ then, $\tup{\tr_1, \tr_2} \in \co'$ and $\tup{\tr_2, \tr_3}$ for some $\tr_3 \not\in T' \uplus \{\tr\}$. But, by the definition of valid extension, for all variables $\xvar$ written by $\tr$, there exists no transaction $\tr_3 \not\in T' \uplus \{\tr\}$ such that it reads $\xvar$ from $\tr_1 \in T'$. Therefore, this is also a contradiction.\hfill $\Box$
       \end{itemize}
%\end{itemize}
 \vspace{-3mm}
\end{proof}

\begin{algorithm}[t]
{\small
 \SetKwInOut{KwInput}{Input}
 \SetKwInOut{KwOutput}{Output}
 \KwIn{A history $\hist = (T, \so, \wro)$, a serializable prefix $T'$ of $\hist$, \\ A set, in \emph{global scope}, $\mathit{seen}$ of prefixes of $\hist$ which are not serializable}
 \KwOut{$\mathit{true}$ iff $T'\vartriangleright^* h$}
 \BlankLine
 \If{$T'$ = $T$}{
  \Return{true}\;
 }
  \ForEach{$\tr \not\in T'$ s.t. $\forall \tr' \not\in T'.\ \tup{\tr', \tr} \not\in \wro \cup \so$}{
   \If{$T'\not\vartriangleright T'\uplus\{\tr\}$}{
    continue\;
   }
   %$\mathit{pref}$   $\leftarrow$ source $\cup \{\tr_\text{next}\}$\;
   \If{$T'\uplus\{\tr\} \not\in\mathit{seen} \land \mathsf{checkSER}(\hist,T'\uplus\{\tr\}, \mathit{seen})$}{
    \Return{true}\;
   }
   $\mathit{seen}\leftarrow\mathit{seen}\cup\{(T'\uplus\{\tr\})\}$\;
  }
  \Return{false}\;
 }
 \caption{The algorithm $\mathsf{checkSER}$ for checking serializabilty}
 \label{seralgo:2}
\end{algorithm}

Algorithm~\ref{seralgo:2} lists our algorithm for checking serializability. It is defined as a recursive procedure that searches for a sequence of valid extensions of a given prefix (initially, this prefix is empty) until covering the whole history. Figure~\ref{ser_algo_example:3} pictures this search on the history in Figure~\ref{ser_algo_example:1}. The right branch (containing blue edges) contains only valid extensions and it reaches a prefix that includes all the transactions in the history.

%\begin{algorithm}
% \SetKwInOut{KwInput}{Input}
% \SetKwInOut{KwOutput}{Output}
% \KwIn{source: serializable source prefix, target: target prefix, seen: seen prefixes, $\hist$: history}
% \KwOut{$\mathit{true}$ iff target is a serializable prefix}
% \BlankLine
% \eIf{source = target}{
%  \Return{true}\;
% }{
%  \ForEach{$\tr_\text{next} \not\in$ source, s.t. $\forall \tr \not\in$ source, $\tup{\tr, \tr_\text{next}} \not\in \wro \cup \so$}{
%   \If{$\tup{\text{source}, \tr_\text{next}} \not\models$ \textsc{SerializableStep}}{
%    continue\;
%   }
%   next $\leftarrow$ source $\cup \{\tr_\text{next}\}$\;
%   \If{next $\not\in$ seen $\land$ dfs(next, target, seen, $\hist$)}{
%    \Return{true}\;
%   }
%  }
%  seen $\leftarrow$ seen $\cup$ \{source\}\;
%  \Return{false}\;
% }
% \caption{DFS algorithm to extend \textsc{Serializable prefix}}
% \label{seralgo:2}
%\end{algorithm}

\begin{theorem}
 A history $\hist$ is serializable iff $\mathsf{checkSER}(\hist,\emptyset,\emptyset)$ returns true.
\end{theorem}
\begin{proof}
The ``if'' direction is a direct consequence of Lemma~\ref{serpref:1}. 
 %
%\textsf{checkSER} uses depth first search to find next possible serializable prefix. If it reaches the whole history $\hist$ from starting from empty prefix, that means, $\hist$ is itself a serializable prefix, which directly implies, $\hist$ has a total order $\co$ which satisfies serializability for the whole history.
%
For the reverse, assume that $\hist=\tup{T,\so,\wro}$ is serializable with a (total) commit order $\co$. Let $\co_i$ be the set of transactions in the prefix of $\co$ of length $i$. 
% and that orders all the transactions in this prefix before all the other transactions. Abusing the terminology, we refer to $\co_i$ as a prefix that contains the transactions 
%where we have to if the history is serializable, the depth first search will reach $\hist$ as serializable prefix. We will prove this by induction. There exists a serializable order $s$ for $\hist$.
Since $\co$ is consistent with $\so$, we have that $\co_i$ is a prefix of $\hist$, for any $i$.
We show by induction that $\co_{i+1}$ is a valid extension of $\co_i$. The base case is trivial. For the induction step, let $\tr$ be the last transaction in the prefix of $\co$ of length $i+1$. Then,
%
%Our hypothesis is, \textsf{checkSER} can reach a prefix $S_i = {s_j | j < i}$ of serializable order $s$, then it can also reach $S_{i+1} = {s_j | j < i + 1} = S_i \cup \{s_i\}$.
%
\begin{itemize}
%  \item Base case. Empty prefix of the serializable order $s$, is a empty serializable prefix. \textsf{checkSER} begins with empty serializable prefix itself. 
%  \item Induction step. \textsf{checkSER} has reached to $S_i$. Obviously, $s_i \not\in S_i$. Now we have to show $S_i \vartriangleright S_i \cup \{s_i\}$.
%
%  $S_i = S_i \cup \{s_i\}$ is a prefix. If it is not, there exists a $s_1 \in S_i$ and $s_2 \not\in S_i$ such that $\tup{s_2, s_1} \in \so$. But in serialization order, $s_1$ comes before $s_2$ which is a contradiction of the serialization order.
%
\item $\tr$ cannot read from a transaction outside of $\co_i$ because $\co$ is consistent with the write-read relation $\wro$, 
%Also $s_i$ does not read from outside $S_i$. It it is not the case, thesre exists $s' \not\in S_i$ s.t. $\tup{s', s_i} \in \wro$. But by serialization order $s$, $s_i$ comes before $s'$. Therefore, for all $s' \in T \in S_i$, $\tup{s', s_i} \not\in \wro$.
%
\item  also, for every variable $\xvar$ written by $\tr$, there exists no transaction $\tr_2 \neq \tr$ outside of $\co_i$ which reads a value of $\xvar$ written by a transaction $\tr_1 \in \co_i$. Otherwise, $\tup{\tr_1,\tr_2}\in\wro[\xvar]$, $\tup{\tr,\tr_2}\in \co$, and $\tup{\tr_1,\tr}\in\co$ which implies that $\tup{\hist,\co}$ does not satisfy $\mathsf{Serializability}$.
%there exists $s_1, s_2, s_i$ such that $\tup{s_2, s_i}, \tup{s_i, s_1} \in s$, and $s_i$ writes on $\xvar$ and $\tup{s_2, s_1} \in \wro[\xvar]$. This is a violation of serializability axiom which contradicts $s$ is a serialization order of $\hist$.  
\end{itemize} 
 % It is interesting to note, reachablity problem is in fact is in \textsf{NL} complexisty class and we can use logarithmic space to represent the prefix and do a nondeterministic reachablity test on the search space. So our problem is in fact in in \textsf{NL}.
% Therefore, $S_i \vartriangleright S_i \cup \{s_i\}$. So \textsf{checkSER} must have reached $S_{i+1}$ unless it reached $\hist$ already in the depth first search. This proves, the depth first search will always reach a serializable prefix of a serialization order if it exists. 
This implies that $\mathsf{checkSER}(\hist,\emptyset,\emptyset)$ returns true.
 %  As for the linearization, it is the same algorithm, except, we have to consider real-time order $\textsf{RO}$ when checking for valid next serializable prefix.
\end{proof}

By definition, the size of each antichain of a history $\hist$ is smaller than the width of $\hist$. Therefore, the number of possible antichains of a history $\hist$ is $O(\mathsf{size}(h)^{\mathsf{width}(h)})$ where $\mathsf{size}(h)$, resp., $\mathsf{width}(h)$, is the number of transactions, resp., the width, of $\hist$. Since the valid extension property can be checked in quadratic time, the asymptotic time complexity of the algorithm defined by $\mathsf{checkSER}$ is upper bounded by $O(\mathsf{size}(h)^{\mathsf{width}(h)}\cdot \mathsf{size}(h)^3)$.
The following corollary is a direct consequence of this observation.

\begin{corollary}\label{cor:ser}

For an arbitrary but fixed constant $k\in\mathbb{N}$, the problem of checking serializability for histories of width at most $k$ is polynomial time.
\end{corollary}

%\begin{proof}
% A direct consequence of the fact that the number of antichains of a bounded-width history is polynomial in the size of the history.
% %This is essentially a reachablity algorithm on a directed graph where each node is a frontier. So the size of the search space is the product of sizes of all sessions. For a history with $k$ of sessions and each session with size $\mathcal{O}(n)$, the search space is of size $\mathcal{O}(n^k)$. For a bounded $k$, that is polysize. So the algorithm runs in polytime if the history has bounded number of sessions.
%\end{proof}

%\vspace{1em}

%!TEX root = draft.tex
\subsection{Reducing Prefix Consistency to Serializability}\label{ssec:pc}

We describe a polynomial time reduction of checking prefix consistency of bounded-width histories to the analogous problem for serializability. Intuitively, as opposed to serializability, prefix consistency allows that two transactions read the same snapshot of the database and commit together even if they write on the same variable. Based on this observation, given a history $\hist$ for which we want to check prefix consistency, we define a new history $\hist_{R|W}$ where each transaction $\tr$ is split into a transaction performing all the reads in $\tr$ and another transaction performing all the writes in $\tr$ (the history $\hist_{R|W}$ retains all the session order and write-read dependencies of $\hist$). We show that if the set of read and write transactions obtained this way can be shown to be serializable, then the original history satisfies prefix consistency, and vice-versa. 
For instance, Figure~\ref{pre_red_example} shows this transformation on the two histories in Figure~\ref{pre_red_example:1} and Figure~\ref{pre_red_example:3}, which represent typical anomalies known as ``long fork'' and ``lost update'', respectively. The former is not admitted by PC while the latter is admitted. It can be easily seen that the transformed history corresponding to the ``long fork'' anomaly is not serializable while the one corresponding to ``lost update'' is serializable.
We show that this transformation leads to a history of the same width, which 
%We show that $\hist$ satisfies prefix consistency iff $\hist_{R|W}$ satisfies serializabilty and that the two histories have the same width. 
by Corollary~\ref{cor:ser}, implies that checking prefix consistency of bounded-width histories is polynomial time.

%given a history $\hist$ we define a transformation 
% to serialization verification problem. Given a history $\hist$, we will construct a new history $\hist'$, s.t. $\hist$ is prefix consistent (resp. snapshot isolation) if and only if $\hist'$ is serializable and $\hist$ and $\hist'$ have equal bounded-width and the number of transactions in $\hist'$ will be twice of the number of transactions in $\hist$. So prefix consistency (resp. snapshot isolation) verification problem of bounded-width histories is also in \textsf{PTIME}.

Thus, given a history $\hist = \tup{T, \wro, \so}$, we define the history $\hist_{R|W} = \tup{T', \wro', \so'}$ as follows:
\begin{itemize}
 \item $T'$ contains a transaction $R_\tr$, called a \emph{read} transaction, and a transaction $W_\tr$, called a \emph{write} transaction, for each transaction $\tr$ in the original history, i.e., $T' = \{R_\tr | \tr \in T\} \cup \{W_\tr | \tr \in T\}$
 \item the write transaction $W_{\tr}$ writes exactly the same set of variables as $\tr$, i.e., for each variable $\xvar$, $W_{\tr}$ writes to $\xvar$ iff $\tr$ writes to $\xvar$.
 \item the read transaction $R_{\tr}$ reads exactly the same values and the same variables as $\tr$, i.e., for each variable $\xvar$,
 $\wro[\xvar]' = \{\tup{W_{\tr_1}, R_{\tr_2}} | \tup{\tr_1, \tr_2} \in \wro[\xvar]\}$
 \item the session order between the read and the write transactions corresponds to that of the original transactions and read transactions precede their write counterparts, i.e.,
 \begin{align*}
 \so' = \{\tup{R_\tr, W_\tr} | \tr \in T\} \cup \{\tup{R_{\tr_1}, R_{\tr_2}}, \tup{R_{\tr_1}, W_{\tr_2}}, \tup{W_{\tr_1}, R_{\tr_2}}, \tup{W_{\tr_1}, W_{\tr_2}} | \tup{\tr_1,\tr_2} \in \so \}
 \end{align*}
\end{itemize}

\begin{figure}
  \centering
  \begin{subfigure}{.49\textwidth}
  \resizebox{\textwidth}{!}{
  \begin{tikzpicture}[->,>=stealth',shorten >=1pt,auto,node distance=3cm,
   semithick, transform shape]
   % \node[draw, rounded corners=2mm] (t1) at (0, 0) {\begin{tabular}{l} \texttt{x = 1;} \\ \texttt{y = 1;}\end{tabular}};
   \node[draw, rounded corners=2mm] (t2) at (-1.7, -1.5) {\begin{tabular}{l} \texttt{read(x); // 0} \\ \texttt{x = 1;} \end{tabular}};
  \node[draw, rounded corners=2mm] (t3) at (1.4, -1.5) {\begin{tabular}{l} \texttt{read(y); // 0} \\ \texttt{y = 1;} \end{tabular}};
  \node[draw, rounded corners=2mm] (t4) at (4.5, -1.5) {\begin{tabular}{l} \texttt{read(x); // 1} \\ \texttt{read(y); // 0} \end{tabular}};
  \node[draw, rounded corners=2mm] (t5) at (7.6, -1.5) {\begin{tabular}{l} \texttt{read(x); // 0} \\ \texttt{read(y); // 1} \end{tabular}};
  % \node[draw, rounded corners=2mm] (t3) at (1.5, 0) {\begin{tabular}{l} \texttt{read(x); // 2} \\ \texttt{read(y); // 1} \end{tabular}};
  % \path (t1) edge node {} (t3);
  % \path (t2) edge node {$\so$} (t3);
  % \path (t1) edge node {} (t2);
  % \path (t3_1) edge node {$\po$} (t3_2);
  \end{tikzpicture}  
  }
   \caption{Long fork}
   \label{pre_red_example:1}
  \end{subfigure}
%  \hspace{.5cm}
  \begin{subfigure}{.49\textwidth}
  \resizebox{\textwidth}{!}{
  \begin{tikzpicture}[->,>=stealth',shorten >=1pt,auto,node distance=3cm,
   semithick, transform shape]
   % \node[draw, rounded corners=2mm] (t1) at (0, 0) {\begin{tabular}{l} \texttt{x = 1;} \\ \texttt{y = 1;}\end{tabular}};
   \node[draw, rounded corners=2mm] (t2r) at (-1.7, -1.5) {\begin{tabular}{l} \texttt{read(x); // 0} \end{tabular}};
   \node[draw, rounded corners=2mm] (t2w) at (-1.7, -3.2) {\begin{tabular}{l} \texttt{x = 1;} \end{tabular}};
   \node[draw, rounded corners=2mm] (t3r) at (1.4, -1.5) {\begin{tabular}{l} \texttt{read(y); // 0} \end{tabular}};
  \node[draw, rounded corners=2mm] (t3w) at (1.4, -3.2) {\begin{tabular}{l} \texttt{y = 1;} \end{tabular}};
  \node[draw, rounded corners=2mm] (t4r) at (4.5, -1.5) {\begin{tabular}{l} \texttt{read(x); // 1} \\ \texttt{read(y); // 0} \end{tabular}};
  \node[draw, rounded corners=2mm] (t5r) at (7.6, -1.5) {\begin{tabular}{l} \texttt{read(y); // 1} \\ \texttt{read(x); // 0} \end{tabular}}; 
  \node[draw, rounded corners=2mm] (t4w) at (4.5, -3.2) {\begin{tabular}{l} \texttt{// empty} \end{tabular}};
   \node[draw, rounded corners=2mm] (t5w) at (7.6, -3.2) {\begin{tabular}{l} \texttt{// empty} \end{tabular}};
  % \node[draw, rounded corners=2mm] (t3) at (1.5, 0) {\begin{tabular}{l} \texttt{read(x); // 2} \\ \texttt{read(y); // 1} \end{tabular}};
  \path (t2r) edge node {$\so$} (t2w);
  \path (t3r) edge node {$\so$} (t3w);
  \path (t4r) edge node {$\so$} (t4w);
  \path (t5r) edge node {$\so$} (t5w);
  % \path (t2) edge node {$\so$} (t3);
  % \path (t1) edge node {} (t2);
  % \path (t3_1) edge node {$\po$} (t3_2);
  \end{tikzpicture}  
  }
   \caption{Long fork (transformed)}
   \label{pre_red_example:2}
  \end{subfigure}

\vspace{3mm}
  \begin{subfigure}{.25\textwidth}
  \resizebox{\textwidth}{!}{
  \begin{tikzpicture}[->,>=stealth',shorten >=1pt,auto,node distance=3cm,
   semithick, transform shape]
   % \node[draw, rounded corners=2mm] (t1) at (0, 1.2) {\begin{tabular}{l} \texttt{x = 1;} \end{tabular}};
   \node[draw, rounded corners=2mm] (t2) at (-1.6, 0) {\begin{tabular}{l} \texttt{read(x); // 0} \\ \texttt{x = 1;} \end{tabular}};
   \node[draw, rounded corners=2mm] (t3) at (1.6, 0) {\begin{tabular}{l} \texttt{read(x); // 0} \\ \texttt{x = 2;} \end{tabular}};
   % \node[draw, rounded corners=2mm] (t3) at (0, -2.4) {\begin{tabular}{l} \texttt{read(x); // 2} \end{tabular}};
  % \node[draw, rounded corners=2mm] (t3) at (1.7, -1.5) {\begin{tabular}{l} \texttt{read(y); // 1} \\ \texttt{y = 2;} \end{tabular}};
  % \node[draw, rounded corners=2mm] (t3) at (1.5, 0) {\begin{tabular}{l} \texttt{read(x); // 2} \\ \texttt{read(y); // 1} \end{tabular}};
  % \path (t1) edge node {} (t3);
  % \path (t2) edge node {$\co$} (t3);
  % \path (t1) edge node {} (t2); 
  % \path (t3_1) edge node {$\po$} (t3_2);
  \end{tikzpicture}  
  }
   \caption{Lost update}
   \label{pre_red_example:3}
  \end{subfigure}
  \hspace{2cm}
  \begin{subfigure}{.32\textwidth}
  \resizebox{.75\textwidth}{!}{
  \begin{tikzpicture}[->,>=stealth',shorten >=1pt,auto,node distance=3cm,
   semithick, transform shape]
   % \node[draw, rounded corners=2mm] (t1) at (0, 1.2) {\begin{tabular}{l} \texttt{x = 1;} \end{tabular}};
   \node[draw, rounded corners=2mm] (t2r) at (0, 0) {\begin{tabular}{l} \texttt{read(x); // 0} \end{tabular}};
   \node[draw, rounded corners=2mm] (t2w) at (0, -1.5) {\begin{tabular}{l}  \texttt{x = 1;} \end{tabular}};
   \node[draw, rounded corners=2mm] (t3r) at (3.2, 0) {\begin{tabular}{l} \texttt{read(x); // 0} \end{tabular}};
   \node[draw, rounded corners=2mm] (t3w) at (3.2, -1.5) {\begin{tabular}{l} \texttt{x = 2;} \end{tabular}};
   % \node[draw, rounded corners=2mm] (t3) at (0, -2.4) {\begin{tabular}{l} \texttt{read(x); // 2} \end{tabular}};
  % \node[draw, rounded corners=2mm] (t3) at (1.7, -1.5) {\begin{tabular}{l} \texttt{read(y); // 1} \\ \texttt{y = 2;} \end{tabular}};
  % \node[draw, rounded corners=2mm] (t3) at (1.5, 0) {\begin{tabular}{l} \texttt{read(x); // 2} \\ \texttt{read(y); // 1} \end{tabular}};
  % \path (t1) edge node {} (t3);
  % \path (t2) edge node {$\co$} (t3);
  % \path (t1) edge node {} (t2); 
  % \path (t3_1) edge node {$\po$} (t3_2);
  \path (t2r) edge node {$\so$} (t2w);
  \path (t3r) edge node {$\so$} (t3w);
  \end{tikzpicture}  
  }
   \caption{Lost update (transformed)}
   \label{pre_red_example:4}
  \end{subfigure}
  \vspace{-3mm}
  \caption{Reducing PC to SER. Initially, the value of every variable is 0.}
  \label{pre_red_example}
  \vspace{-3mm}
\end{figure}
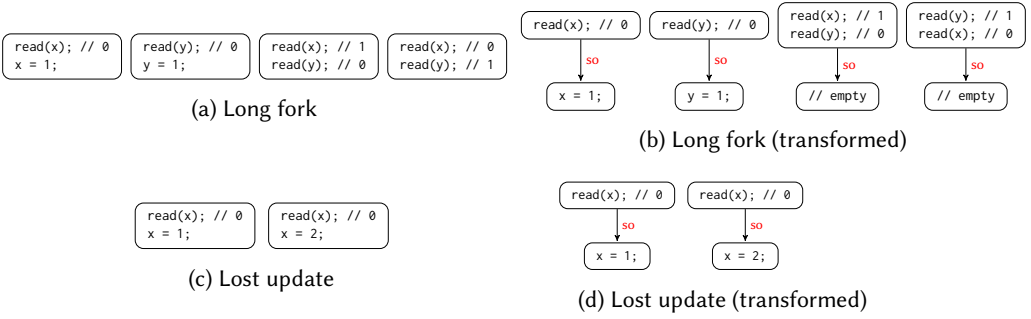

The following lemma is a straightforward consequence of the definitions (see Appendix~\ref{app:pc_red}).

\begin{lemma}\label{lem:pc_width}
The histories $\hist$ and $\hist_{R|W}$ have the same width.
\end{lemma}

Next, we show that $\hist_{R|W}$ is serializable if $\hist$ is prefix consistent. Formally, we show that
 \begin{align*}
  \forall \co.\ \exists \co'.\ \tup{\hist, \co} \models \axpre \Rightarrow \tup{\hist_{R|W}, \co'} \models \axser 
 \end{align*}
%\begin{theorem}
% There is a polynomial time reduction from prefix consistency verification problem to serialization verification problem - without increasing the width of the history.
%\end{theorem}
%
%
%Now we claim, $\hist$ is prefix consistent if and only if $\hist_{R|W}$ is serializable.
%
%\begin{proof}
% First we will prove, if $\hist$ is prefix consistent, then $\hist_{R|W}$ is serializable. Formally we are trying to prove,
% 
% \begin{align}
%  \forall \co, \exists \co' \tup{\hist, \co} \models \axpre \Rightarrow \tup{\hist_{R|W}, \co'} \models \axser \label{pre_leftright}
% \end{align}
%
% Consider a total order $\co$ for $\hist$ which satisfies prefix consistency. We will show there exists a $\co'$ for $\hist_{R|W}$ which satisfies serialization. To find such $\co'$, first we construct a partial order $\co'_1$ on $T'$ and then try to extend $\co'_1$ to $\co'$.
Thus, let $\co$ be a commit (total) order on transactions of $\hist$ which together with $\hist$ satisfies the prefix consistency axiom. We define two \emph{partial} commit orders $\co'_1$ and $\co'_2$, $\co'_2$ a strengthening of $\co'_1$, which we prove that they are acyclic and that any linearization $\co'$ of $\co'_2$ is a valid witness for $\hist_{R|W}$ satisfying serializability.

Thus, let $\co'_1$ be a \emph{partial} commit order on transactions of $\hist_{R|W}$ defined as follows:
 \begin{align*}
  \co'_1 = \{\tup{R_{\tr}, W_{\tr}} | \tr \in T\} \cup \{\tup{W_{\tr_1}, W_{\tr_2}} | \tup{\tr_1, \tr_2} \in \co\}\ \cup \{\tup{W_{\tr_1},R_{\tr_2}} | \tup{\tr_1, \tr_2} \in \wro \cup \so\} 
 \end{align*}
 
We show that if $\co'_1$ were to be cyclic, then it contains a minimal cycle with one read transaction, and at least one but at most two write transactions. Then, we show that such cycles cannot exist. 

 \begin{lemma}\label{lem:co1}
The relation $\co'_1$ is acyclic.
\end{lemma}
% \begin{proof}
% Now we will show $\co'_1$ is acyclic. To do that, first we show few properties of a minimal cycle in $\co'_1$ to ease our proofs. 
 \textsc{Proof.} We first show that if $\co'_1$ were to be cyclic, then it contains a minimal cycle with one read transaction, and at least one but at most two write transactions. Then, we show that such cycles cannot exist. 
Therefore, let us assume that $\co'_1$ is cyclic. Then,
%\vspace{-4mm}
  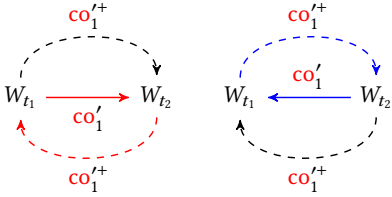
\begin{wrapfigure}{l}{0.4\textwidth} 
%  \centering
  \begin{subfigure}{.19\textwidth}
   \resizebox{\textwidth}{!}{
    \begin{tikzpicture}[->,>=stealth',shorten >=1pt,auto,node distance=4cm,
      semithick, transform shape]
     \node[transaction state] at (0,0)       (t_1)           {$W_{\tr_1}$};
     \node[transaction state] at (2,0)       (t_2)           {$W_{\tr_2}$};
     \path (t_1) edge[dashed, bend left=90] node {$\co'^+_1$} (t_2);
     \path (t_2) edge[dashed, color=red, bend left=90] node {$\co'^+_1$} (t_1);
     \path (t_1) edge[color=red] node[below] {$\co'_1$} (t_2);
    \end{tikzpicture}  
   }
   \caption{$\tup{W_{\tr_1}, W_{\tr_2}} \in \co'_1$}
   \label{ww_consecutive:a}
  \end{subfigure}
  \hspace{1mm}
  \begin{subfigure}{.19\textwidth}
   \resizebox{\textwidth}{!}{
    \begin{tikzpicture}[->,>=stealth',shorten >=1pt,auto,node distance=4cm,
      semithick, transform shape]
     \node[transaction state] at (0,0)       (t_1)           {$W_{\tr_1}$};
     \node[transaction state] at (2,0)       (t_2)           {$W_{\tr_2}$};
     \path (t_1) edge[dashed, color=blue, bend left=90] node {$\co'^+_1$} (t_2);
     \path (t_2) edge[dashed, bend left=90] node {$\co'^+_1$} (t_1);
     \path (t_2) edge[color=blue] node[above] {$\co'_1$} (t_1);
    \end{tikzpicture}  
   }
   \caption{$\tup{W_{\tr_2}, W_{\tr_1}} \in \co'_1$}
   \label{ww_consecutive:b}
  \end{subfigure}
  \vspace{-3mm}
  \caption{Cycles with non-consecutive write transactions.}
  \label{ww_consecutive}
  \vspace{-5mm}
 \end{wrapfigure} 
 \begin{itemize}
  \item Since $\tup{W_{\tr_1}, W_{\tr_2}} \in \co'_1$ implies $\tup{\tr_1, \tr_2} \in \co$, for every $\tr_1$ and $\tr_2$, a cycle in $\co'_1$ cannot contain only write transactions. Otherwise, it will imply a cycle in the original commit order $\co$. Therefore, a cycle in $\co'_1$ must contain at least one read transaction. 
  \item Assume that a cycle in $\co'_1$ contains two write transactions $W_{\tr_1}$ and $W_{\tr_2}$ which are not consecutive, like in Figure~\ref{ww_consecutive}.
%  In figure \ref{ww_consecutive}, we have a minimal cycle in $\co'_1$ in which there are two transactions $W_{\tr_1}$ and $W_{\tr_2}$ which are not consecutive. 
Since either $\tup{W_{\tr_1}, W_{\tr_2}}\in \co'_1$ or $\tup{W_{\tr_1}, W_{\tr_2}}\in \co'_1$, there exists a smaller cycle in $\co'_1$ where these two write transactions are consecutive. If $\tup{W_{\tr_1}, W_{\tr_2}}\in \co'_1$, then $\co'_1$ contains the smaller cycle on the lower part of the original cycle (Figure~\ref{ww_consecutive:a}), and if $\tup{W_{\tr_2}, W_{\tr_1}}\in \co'_1$, then $\co'_1$ contains the cycle on the upper part of the original cycle (Figure~\ref{ww_consecutive:b}). Thus, all the write transactions in a minimal cycle of $\co'_1$ must be consecutive. 
% So all $W_{\_}$ transactions in a minimal cycle in $\co'_1$ must be consecutive.
\end{itemize}

\begin{itemize}
  \item If a minimal cycle were to contain three write transactions, then all of them cannot be consecutive unless they all three form a cycle, which is not possible. So a minimal cycle contains at most two write transactions.
  \item Since $\co'_1$ contains no direct relation between read transactions, it cannot contain a cycle with two consecutive read transactions, or only read transactions.
  %there is no relation of the form $\tup{R_{\_}, R_{\_}}$ in $\co'_1$. So there is no cycle with consecutive $R_{\_}$ transactions.
%  \item All these above properties togerther imply a minimal cycle in $\co'_1$ contains atleast one but atmost two $W_{\_}$ transactions and one $R_{\_}$ transaction.
 \end{itemize}
This shows that a minimal cycle of $\co'_1$ would include a read transaction and a write transaction, and at most one more write transaction. We prove that such cycles are however impossible:
 \begin{itemize}
  \item if the cycle is of size 2, then it contains two transactions $W_{\tr_1}$ and $R_{\tr_2}$ such that $\tup{W_{\tr_1}, R_{\tr_2}}\in\co'_1$ and $\tup{R_{\tr_2}, W_{\tr_1}}\in \co'_1$. Since all the $\tup{R_{\_}, W_{\_}}$ dependencies in $\co'_1$ are of the form $\tup{R_\tr, W_\tr}$, it follows that $\tr_1 = \tr_2$. Then, we have $\tup{W_{\tr_1}, R_{\tr_1}} \in \co'_1$ which implies $\tup{\tr_1, \tr_1} \in \wro \cup \so$, a contradiction.
  \item if the cycle is of size 3, then it contains three transactions $W_{\tr_1}$, $W_{\tr_2}$, and $R_{\tr_3}$ such that $\tup{W_{\tr_1}, W_{\tr_2}}\in \co'_1$,  $\tup{W_{\tr_2}, R_{\tr_3}}\in \co'_1$, and $\tup{R_{\tr_3}, W_{\tr_1}} \in \co'_1$. Using a similar argument as in the previous case, $\tup{R_{\tr_3}, W_{\tr_1}} \in \co'_1$ implies $\tr_3 = \tr_1$. Therefore, $\tup{\tr_1, \tr_2} \in \co$ and $\tup{\tr_2, \tr_1} \in \wro \cup \so$, which contradicts the fact that $\wro \cup \so\subseteq \co$. \hfill $\Box$
  %  satisfies prefix consistency for $\hist$.
 \end{itemize}
% \end{proof}

 %Therefore, $\co'_1$ is acyclic. 
%Now we want to remove the choices for which any acyclic extension of $\co'_1$ will violate $\axser$. The extension will be a total order, so if we do not want a relation $\tup{\tr_1, \tr_2}$ in $\co'$, then $\tup{\tr_2, \tr_1}$ must be in it. So we collect all such relations implied by $\co'_1$ in 
 We define a strengthening of $\co'_1$ where intuitively, we add all the dependencies from read transactions $\tr_3$ to write transactions $\tr_2$ that ``overwrite'' values read by $\tr_3$. Formally, $\co'_2= \co'_1\cup\rwo(\co'_1)$ where 
 \begin{align*}
  \rwo(\co'_1) = \{\tup{\tr_3, \tr_2}| \exists \xvar \in \vars{h}.\ \exists \tr_1\in T'.\ \tup{\tr_1,\tr_3} \in \wro[\xvar]', \tup{\tr_1, \tr_2} \in \co'_1, \writeVar{\tr_2}{\xvar} \} 
 \end{align*}
 
 It can be shown that any cycle in $\co'_2$ would correspond to a $\mathsf{Prefix}$ violation in the original history. Therefore,
 
 \vspace{-1mm}
 \begin{lemma}\label{lem:co2}
 The relation $\co'_2$ is acyclic.
 \end{lemma}
 \vspace{-3mm}
  \begin{wrapfigure}{l}{0.5\textwidth} 
  \centering
  \begin{subfigure}{.22\textwidth}
   \resizebox{\textwidth}{!}{
    \begin{tikzpicture}[->,>=stealth',shorten >=1pt,auto,node distance=4cm,
      semithick, transform shape]
     \node[transaction state] at (0,0)       (t_1)           {$W_{\tr_1}$};
     \node[transaction state] at (2,0)       (t_3)           {$R_{\tr_3}$};
     \node[transaction state, label={above:$\writeVar{ }{\xvar}$}] at (-0.5,1.5) (t_2) {$W_{\tr_2}$};
     \node[transaction state] at (1.5,1.5) (t_4) {$W_{\tr_4}$};
     \path (t_1) edge node {$\wro[\xvar]$} (t_3);
     % \path (t_2) edge[blue] node {$\CO$} (t_1);
     \path (t_2) edge[red] node {$\co'^*_1$} (t_4);
     \path (t_4) edge[red] node {$\co'_1$} (t_3);
     \path (t_1) edge[left] node {$\co'_1$} (t_2);
     \path (t_3) edge[red, right] node[pos=.8,rotate=-30,yshift=2.5mm] {$\rwo(\co'_1)$} (t_2);
    \end{tikzpicture}
   }
   \caption{Minimal cycle in $\co'_2$.}
   \label{pc_p_proof:2a}
  \end{subfigure}
  \hspace{1mm}
  \begin{subfigure}{0.26\textwidth}
   \resizebox{\textwidth}{!}{
    \begin{tikzpicture}[->,>=stealth',shorten >=1pt,auto,node distance=4cm,
      semithick, transform shape]
     \node[transaction state, text=red] at (0,0)       (t_1)           {$\tr_1$};
     \node[transaction state] at (2,0)       (t_3)           {$\tr_3$};
     \node[transaction state, text=red,label={above:\textcolor{red}{$\writeVar{ }{\xvar}$}}] at (-0.5,1.5) (t_2) {$\tr_2$};
     \node[transaction state] at (1.5,1.5) (t_4) {$\tr_4$};
     \path (t_1) edge[red] node {$\wro[\xvar]$} (t_3);
     % \path (t_2) edge[blue] node {$\CO$} (t_1);
     \path (t_2) edge node {$\co^*$} (t_4);
     \path (t_4) edge node {$\wro \cup \so$} (t_3);
     \path (t_1) edge[left] node {$\co$} (t_2);
    \end{tikzpicture}
   }
   \caption{$\axpre$ violation in $\tup{\hist, \co}$.}
   \label{pc_p_proof:2b}
  \end{subfigure}
   \vspace{-3mm}
  \caption{Cycles in $\co'_2$ correspond to $\axpre$ violations.}
  \label{pc_p_proof:2}
   \vspace{-2.5mm}
 \end{wrapfigure}
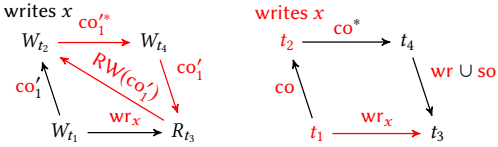
 \textsc{Proof.}
 Assume that $\co'_2$ is cyclic. Any minimal cycle in $\co'_2$ still satisfies the properties of minimal cycles of $\co'_1$ proved in Lemma~\ref{lem:co1} (because all write transactions are still totally ordered and $\co'_2$ doesn't relate directly read transactions). 
 %- $\co'_2$ has relations of the form $\tup{R_{\_}, W_{\_}}$. 
 So, a minimal cycle in $\co'_2$ contains a read transaction and a write transaction, and at most one more write transaction.
 
 Since $\co'_1$ is acyclic, a cycle in $\co'_2$, and in particular a minimal one, must  necessarily contain a dependency from $\rwo(\co'_1)$. Note that a minimal cycle cannot contain two such dependencies since this would imply that it contains two non-consecutive write transactions. 
%  $\co'_1$ was acyclic. All the relations in $\co'_2$ are of the form $\tup{R_{\_}, W_{\_}}$. If $(\co'_1 \cup \co'_2)$ has a cycle, then the cycle must contain an relation from $\co'_2$. But two $\tup{R_{\_}, W_{\_}}$ in a cycle implies, two non-consecutive $W_{\_}$ in a cycle. So a simple cycle in $(\co'_1 \cup \co'_2)$ would contain only one relation of the form $\tup{R_{\_}, W_{\_}}$.
% 
The red edges in Figure~\ref{pc_p_proof:2a} show a minimal cycle of $\co'_2$ satisfying all the properties mentioned above. This cycle contains a dependency $\tup{R_{\tr_3}, W_{\tr_2}}\in \rwo(\co'_1)$ which implies the existence of a write transaction $W_{\tr_1}$ in $\hist_{R|W}$ s.t. $\tup{W_{\tr_1}, R_{\tr_3}} \in \wro[\xvar]'$ and $\tup{W_{\tr_1}, W_{\tr_2}} \in \co'_1$ and $W_{\tr_1}, W_{\tr_2}$ write on $\xvar$ (these dependencies are represented by the black edges in Figure~\ref{pc_p_proof:2a}). The relations between these transactions of $\hist_{R|W}$ imply that the corresponding transactions of $\hist$ are related as shown in Figure~\ref{pc_p_proof:2b}: $\tup{W_{\tr_1}, W_{\tr_2}} \in \co'_1$ and $\tup{W_{\tr_2}, W_{\tr_4}} \in \co'^*_1$ imply $\tup{\tr_1, \tr_2} \in \co$ and $\tup{\tr_2, \tr_4} \in \co^*$, respectively, $\tup{W_{\tr_1}, W_{\tr_3}} \in \wro[\xvar]'$ implies $\tup{\tr_1, \tr_3} \in \wro[\xvar]$, and $\tup{W_{\tr_4}, R_{\tr_3}} \in \co'_1$ implies $\tup{\tr_4, \tr_3} \in \wro \cup \so$. This implies that $\tup{\hist,\co}$ doesn't satisfy the $\axpre$ axiom, a contradiction. \hfill $\Box$
%But to satisfy $\axpre$, $\tup{\tr_2, \tr_1}$ should have been in $\co$, not $\tup{\tr_1, \tr_2}$ - which is a contradiction.
%  \end{proof}

 \begin{lemma}\label{lem:pc1:app}
If a history $\hist$ satisfies prefix consistency, then $\hist_{R|W}$ is serializable.
\end{lemma}
 \begin{proof}
% So $\co'_1 \cup \co'_2$ is acyclic. We take $\co'$ to be any topological order of $\co'_1 \cup \co'_2$. 
 Let $\co'$ be any total order consistent with $\co'_2$. Assume by contradiction that $\tup{\hist_{R|W},\co'}$ doesn't satisfy $\axser$. Then, there exist $\tr'_1, \tr'_2, \tr'_3 \in T'$ such that $\tup{\tr'_1, \tr'_2}, \tup{\tr'_2, \tr'_3} \in \co'$ and $\tr'_1, \tr'_2$ write on some variable $\xvar$ and $\tup{\tr'_1, \tr'_3} \in \wro[\xvar]'$. But then $\tr'_1, \tr'_2$ are write transactions and $\co'_1$ must contain $\tup{\tr'_1, \tr'_2}$. Therefore, $\rwo(\co'_1)$ and $\co'_2$ should contain $\tup{\tr'_3, \tr'_2}$, a contradiction with $\co'$ being consistent with $\co'_2$.
%  we must have added $\tup{\tr'_3, \tr'_2} \in \co'_2$. So $\tup{\tr'_2, \tr'_3}$ can not be in $\co'$. Therefore, $\co'$ must satisfy $\axser$ and proves our claim in predicate (\ref{pre_leftright}).
 \end{proof}

 Finally, it can be proved that any linearization $\co'$ of $\co'_2$ satisfies $\mathsf{Serializability}$ (together with $\hist_{R|W}$). Moreover, it can also be shown that the serializability of $\hist_{R|W}$ implies that $\hist$ satisfies PC. Therefore,

\begin{theorem}\label{th:pc}
A history $\hist$ satisfies prefix consistency iff $\hist_{R|W}$ is serializable.
\end{theorem}
 \begin{wrapfigure}{l}{0.53\textwidth} 
  \centering
  \begin{subfigure}{.26\textwidth}
   \resizebox{\textwidth}{!}{
    \begin{tikzpicture}[->,>=stealth',shorten >=1pt,auto,node distance=4cm,
      semithick, transform shape]
     \node[transaction state, text=red] at (0,0)       (t_1)           {$\tr_1$};
     \node[transaction state] at (2,0)       (t_3)           {$\tr_3$};
     \node[transaction state, text=red,label={above:\textcolor{red}{$\writeVar{ }{\xvar}$}}] at (-0.5,1.5) (t_2) {$\tr_2$};
     \node[transaction state] at (1.5,1.5) (t_4) {$\tr_4$};
     \path (t_1) edge[red] node {$\wro[\xvar]$} (t_3);
     % \path (t_2) edge[blue] node {$\CO$} (t_1);
     \path (t_2) edge node {$\co^*$} (t_4);
     \path (t_4) edge node {$\wro \cup \so$} (t_3);
     \path (t_1) edge[left] node {$\co$} (t_2);
    \end{tikzpicture}
   }
   \caption{$\axpre$ violation in $\tup{\hist, \co}$}
   \label{pc_p_proof:3a}
  \end{subfigure}
  \begin{subfigure}{0.26\textwidth}
   \resizebox{\textwidth}{!}{
    \begin{tikzpicture}[->,>=stealth',shorten >=1pt,auto,node distance=4cm,
      semithick, transform shape]
     \node[transaction state] at (0,0)       (t_1)           {$W_{\tr_1}$};
     \node[transaction state] at (2,0)       (t_3)           {$R_{\tr_3}$};
     \node[transaction state, label={above:\textcolor{red}{$\writeVar{ }{\xvar}$}}] at (-0.5,1.5) (t_2) {$W_{\tr_2}$};
     \node[transaction state] at (1.5,1.5) (t_4) {$W{\tr_4}$};
     \path (t_1) edge node[near start] {$\wro[\xvar]'$} (t_3);
     % \path (t_2) edge[blue] node {$\CO$} (t_1);
     \path (t_2) edge[red] node {$\co'^*$} (t_4);
     \path (t_4) edge[red] node {$\wro' \cup \so'$} (t_3);
     \path (t_1) edge[left] node {$\co'$} (t_2);
     \path (t_3) edge[red,above right] node {$\co'$} (t_2);
    \end{tikzpicture}
   }
   \caption{Cycle in $\co'$.}
   \label{pc_p_proof:3b}
  \end{subfigure}
  
  \caption{$\axpre$ violations correspond to cycles in $\co'$.}
  \label{pc_p_proof:3}
 \end{wrapfigure}
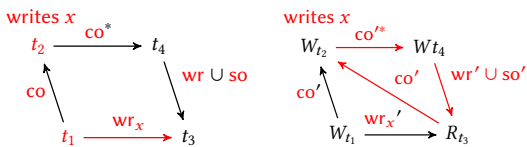
 \textsc{Proof.}
The ``only-if'' direction is proven by Lemma~\ref{lem:pc1:app}. For the reverse, we show that 
 % Now we will prove, if $\hist_{R|W}$ is serializable, then $\hist$ is prefix consistent. Formally we are trying to prove,

\noindent
$\forall \co'.\ \exists \co.\ \tup{\hist_{R|W}, \co'} \models \axser $

\hspace{3cm}
$\Rightarrow \tup{\hist, \co} \models \axpre$ %\label{pre_rightleft}
 
 Thus, let $\co'$ be a commit (total) order on transactions of $\hist_{R|W}$ which together with $\hist_{R|W}$ satisfies the serializability axiom.
 % Consider a total order $\co'$ for $\hist_{R|W}$ which satisfies serialization. We will show there exists a $\co$ for $h$ which satisfies prefix consistency. 
 Let $\co$ be a commit order on transactions of $\hist$ defined by 
 $\co = \{\tup{\tr_1, \tr_2} | \tup{W_{\tr_1}, W_{\tr_2}} \in \co'\}$ ($\co$ is clearly a total order). If $\co$ were not to be consistent with $\wro \cup \so$, then there would exist transactions $\tr_1$ and $\tr_2$ such that  $\tup{\tr_1, \tr_2} \in \wro \cup \so$ and $\tup{\tr_2, \tr_1} \in \co$, which would imply that $\tup{W_{\tr_1}, R_{\tr_2}} , \tup{R_{\tr_2}, W_{\tr_2}} \in \wro \cup \so$ and $\tup{W_{\tr_2}, W_{\tr_1}} \in \co'$, which violates the acylicity of $\co'$. We show that $\tup{\hist, \co}$ satisfies $\axpre$. Assume by contradiction that there exists a $\axpre$ violation between $\tr_1$, $\tr_2$, $\tr_3$, $\tr_4$ (shown in Figure \ref{pc_p_proof:3a}), i.e., for some $\xvar \in \vars{\hist}$, $\tup{\tr_1, \tr_3} \in \wro[\xvar]$ and $\writeVar{\tr_2}{\xvar}$, $\tup{\tr_1, \tr_2} \in \co$, $\tup{\tr_2, \tr_4} \in \co^*$ and $\tup{\tr_4, \tr_3} \in \wro \cup \so$. Then, the corresponding transactions $W_{\tr_1}, W_{\tr_2}, W_{\tr_4}, R_{\tr_3}$ in $\hist_{R|W}$ would be related as follows: 
$\tup{W_{\tr_1}, W_{\tr_2}} \in \co'$ and $\tup{W_{\tr_1}, R_{\tr_3}} \in \wro[\xvar]'$ because $\tup{\tr_1, \tr_3} \in \wro[\xvar]$ and $\tup{\tr_1, \tr_2} \in \co$.
        Since $\co'$ satisfies $\axser$, then $\tup{R_{\tr_3}, W_{\tr_2}} \in \co'$.
        But $\tup{\tr_2, \tr_4} \in \co^*$ and $\tup{\tr_4, \tr_3} \in \wro \cup \so$ imply that $\tup{W_{\tr_2}, W_{\tr_4}} \in \co'^*$ and $\tup{W_{\tr_4}, R_{\tr_3}} \in \wro' \cup \so'$, which show that $\co'$ is cyclic (the red cycle in Figure \ref{pc_p_proof:3b}), a contradiction. \hfill $\Box$

Since the history $\hist_{R|W}$ can be constructed in linear time, Lemma~\ref{lem:pc_width}, Theorem~\ref{th:pc}, and Corollary~\ref{cor:ser} imply the following result.
 
 \begin{corollary}\label{cor:pc}
 
For an arbitrary but fixed constant $k\in\mathbb{N}$, the problem of checking prefix consistency for histories of width at most $k$ is polynomial time.
 \end{corollary}
 
%\end{proof}

%!TEX root = draft.tex

\subsection{Reducing Snapshot Isolation to Serializability}\label{ssec:si}

We extend the reduction of prefix consistency to serializability to the case of snapshot isolation. Compared to prefix consistency, snapshot isolation disallows transactions that read the same snapshot of the database to commit together if they write on a common variable (stated by the $\mathsf{Conflict}$ axiom). More precisely, for any pair of transactions $\tr_1$ and $\tr_2$ writing to a common variable, $\tr_1$ must observe the effects of $\tr_2$ or vice-versa. 
%when one of these transactions reads a variable written by the other one. 
We refine the definition of $\hist_{R|W}$ such that any ``serialization'' (i.e.., commit order satisfying $\mathsf{Serializability}$) disallows that the read transactions corresponding to two such transactions are ordered both before their write counterparts. We do this by introducing auxiliary variables that are read or written by these transactions. For instance, 
Figure~\ref{si_red_example} shows this transformation on the two histories in Figure~\ref{si_red_example:1} and Figure~\ref{si_red_example:3}, which represent the anomalies known as ``lost update'' and ``write skew'', respectively. The former is not admitted by SI while the latter is admitted. Concerning ``lost update'', the read counterpart of the transaction on the left writes to a variable {\tt x12} which is read by its write counterpart, but also written by the write counterpart of the other transaction. This forbids that the latter is serialized in between the read and write counterparts of the transaction on the left. A similar scenario is imposed on the transaction on the right, which makes that the transformed history is not serializable. Concerning the ``write skew'' anomaly, the transformed history is exactly as for the PC reduction since the two transactions don't write on a common variable. It is clearly serializable.

%\begin{theorem}
% There is a polynomial time reduction from snapshot isolation verification problem to serialization verification problem - without increasing the width of the history.
%\end{theorem}

For a history $\hist = \tup{T, \wro, \so}$, the history $\hist_{R|W}^c = \tup{T', \wro', \so'}$ is defined as $\hist_{R|W}$ with the following additional construction: for every two transactions $\tr_1$ and $\tr_2 \in T$ that write on a common variable,
% and $\hist_{R|W} = \tup{T', \wro', \so'}$, we define the history $\hist_{R|W}^c = \tup{T', \wro'', \so'}$ where 
\begin{itemize}
%\item for every variable $\xvar$ in the original history $\hist$, i.e., $\wro''[\xvar]=\wro'[\xvar]$,
\item $R_{\tr_1}$ and $W_{\tr_2}$ (resp., $R_{\tr_2}$ and $W_{\tr_1}$) write on a variable $\xvar_{1,2}$ (resp., $\xvar_{2,1}$),
\item the write transaction of $\tr_i$ reads $\xvar_{i,j}$ from the read transaction of $\tr_i$, for all $i\neq j\in\{1,2\}$, i.e., $\wro[\xvar_{1,2}]= \{\tup{R_{\tr_1}, W_{\tr_1}}\}$ and $\wro[\xvar_{2,1}]= \{\tup{R_{\tr_2}, W_{\tr_2}}\}$.
\end{itemize}

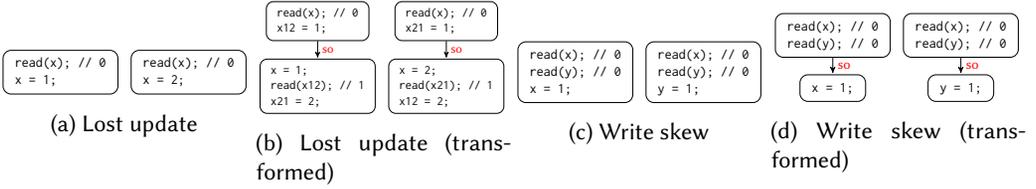
\begin{figure}

    \begin{subfigure}{.24\textwidth}
    \resizebox{\textwidth}{!}{
    \begin{tikzpicture}[->,>=stealth',shorten >=1pt,auto,node distance=3cm,
     semithick, transform shape]
     % \node[draw, rounded corners=2mm] (t1) at (0, 1.2) {\begin{tabular}{l} \texttt{x = 1;} \end{tabular}};
     \node[draw, rounded corners=2mm] (t2) at (-1.6, 0) {\begin{tabular}{l} \texttt{read(x); // 0} \\ \texttt{x = 1;} \end{tabular}};
     \node[draw, rounded corners=2mm] (t3) at (1.6, 0) {\begin{tabular}{l} \texttt{read(x); // 0} \\ \texttt{x = 2;} \end{tabular}};
     % \node[draw, rounded corners=2mm] (t3) at (0, -2.4) {\begin{tabular}{l} \texttt{read(x); // 2} \end{tabular}};
    % \node[draw, rounded corners=2mm] (t3) at (1.7, -1.5) {\begin{tabular}{l} \texttt{read(y); // 1} \\ \texttt{y = 2;} \end{tabular}};
    % \node[draw, rounded corners=2mm] (t3) at (1.5, 0) {\begin{tabular}{l} \texttt{read(x); // 2} \\ \texttt{read(y); // 1} \end{tabular}};
    % \path (t1) edge node {} (t3);
    % \path (t2) edge node {$\co$} (t3);
    % \path (t1) edge node {} (t2); 
    % \path (t3_1) edge node {$\po$} (t3_2);
    \end{tikzpicture}  
    }
     \caption{Lost update}
     \label{si_red_example:1}
    \end{subfigure}
%    \hspace{3mm}
    \begin{subfigure}{.24\textwidth}
    \resizebox{\textwidth}{!}{
    \begin{tikzpicture}[->,>=stealth',shorten >=1pt,auto,node distance=3cm,
     semithick, transform shape]
     % \node[draw, rounded corners=2mm] (t1) at (0, 1.2) {\begin{tabular}{l} \texttt{x = 1;} \end{tabular}};
     \node[draw, rounded corners=2mm] (t2r) at (-1.8, .9) {\begin{tabular}{l} \texttt{read(x); // 0} \\ \texttt{x12 = 1;} \end{tabular}};
     \node[draw, rounded corners=2mm] (t2w) at (-1.8, -.9) {\begin{tabular}{l}  \texttt{x = 1;}  \\ \texttt{read(x12); // 1} \\ \texttt{x21 = 2;} \end{tabular}};
     \node[draw, rounded corners=2mm] (t3r) at (1.8, .9) {\begin{tabular}{l} \texttt{read(x); // 0} \\ \texttt{x21 = 1;} \end{tabular}};
     \node[draw, rounded corners=2mm] (t3w) at (1.8, -.9) {\begin{tabular}{l} \texttt{x = 2;}  \\ \texttt{read(x21); // 1} \\ \texttt{x12 = 2;}\end{tabular}};
     % \node[draw, rounded corners=2mm] (t3) at (0, -2.4) {\begin{tabular}{l} \texttt{read(x); // 2} \end{tabular}};
    % \node[draw, rounded corners=2mm] (t3) at (1.7, -1.5) {\begin{tabular}{l} \texttt{read(y); // 1} \\ \texttt{y = 2;} \end{tabular}};
    % \node[draw, rounded corners=2mm] (t3) at (1.5, 0) {\begin{tabular}{l} \texttt{read(x); // 2} \\ \texttt{read(y); // 1} \end{tabular}};
    % \path (t1) edge node {} (t3);
    % \path (t2) edge node {$\co$} (t3);
    % \path (t1) edge node {} (t2); 
    % \path (t3_1) edge node {$\po$} (t3_2);
    \path (t2r) edge node {$\so$} (t2w);
    \path (t3r) edge node {$\so$} (t3w);
    \end{tikzpicture}  
    }
     \caption{Lost update (transformed)}
     \label{si_red_example:2}
    \end{subfigure}
%    \hspace{3mm}
  \begin{subfigure}{.24\textwidth}
  \resizebox{\textwidth}{!}{
  \begin{tikzpicture}[->,>=stealth',shorten >=1pt,auto,node distance=3cm,
   semithick, transform shape]
   % \node[draw, rounded corners=2mm] (t1) at (0, 0) {\begin{tabular}{l} \texttt{x = 1;} \\ \texttt{y = 1;}\end{tabular}};
   \node[draw, rounded corners=2mm] (t2) at (-1.6, -1.5) {\begin{tabular}{l} \texttt{read(x); // 0} \\ \texttt{read(y); // 0} \\ \texttt{x = 1;} \end{tabular}};
   \node[draw, rounded corners=2mm] (t3) at (1.6, -1.5) {\begin{tabular}{l} \texttt{read(x); // 0} \\ \texttt{read(y); // 0} \\ \texttt{y = 1;} \end{tabular}};
  % \node[draw, rounded corners=2mm] (t3) at (1.7, -1.5) {\begin{tabular}{l} \texttt{read(y); // 1} \\ \texttt{y = 2;} \end{tabular}};
  % \node[draw, rounded corners=2mm] (t3) at (1.5, 0) {\begin{tabular}{l} \texttt{read(x); // 2} \\ \texttt{read(y); // 1} \end{tabular}};
  % \path (t1) edge node {} (t3);
  % \path (t2) edge node {$\so$} (t3);
  % \path (t1) edge node {} (t2);
  % \path (t3_1) edge node {$\po$} (t3_2);
  \end{tikzpicture}  
  }
   \caption{Write skew}
   \label{si_red_example:3}
  \end{subfigure}
%  \hspace{3mm}
  \begin{subfigure}{.24\textwidth}
  \resizebox{\textwidth}{!}{
  \begin{tikzpicture}[->,>=stealth',shorten >=1pt,auto,node distance=3cm,
   semithick, transform shape]
   % \node[draw, rounded corners=2mm] (t1) at (0, 0) {\begin{tabular}{l} \texttt{x = 1;} \\ \texttt{y = 1;}\end{tabular}};
   \node[draw, rounded corners=2mm] (t2r) at (-1.6, 0) {\begin{tabular}{l} \texttt{read(x); // 0} \\ \texttt{read(y); // 0} \end{tabular}};
   \node[draw, rounded corners=2mm] (t2w) at (-1.6, -1.3) {\begin{tabular}{l} \texttt{x = 1;} \end{tabular}};
   \node[draw, rounded corners=2mm] (t3r) at (1.6, 0) {\begin{tabular}{l} \texttt{read(x); // 0} \\ \texttt{read(y); // 0} \end{tabular}};
   \node[draw, rounded corners=2mm] (t3w) at (1.6, -1.3) {\begin{tabular}{l} \texttt{y = 1;} \end{tabular}};
  % \node[draw, rounded corners=2mm] (t3) at (1.7, -1.5) {\begin{tabular}{l} \texttt{read(y); // 1} \\ \texttt{y = 2;} \end{tabular}};
  % \node[draw, rounded corners=2mm] (t3) at (1.5, 0) {\begin{tabular}{l} \texttt{read(x); // 2} \\ \texttt{read(y); // 1} \end{tabular}};
  % \path (t1) edge node {} (t3);
  % \path (t2) edge node {$\so$} (t3);
  % \path (t1) edge node {} (t2);
  % \path (t3_1) edge node {$\po$} (t3_2);
  \path (t2r) edge node {$\so$} (t2w);
  \path (t3r) edge node {$\so$} (t3w);
  \end{tikzpicture}  
  }
   \caption{Write skew (transformed)}
   \label{si_red_example:4}
  \end{subfigure}
   \vspace{-3mm}
  \caption{Reducing SI to SER.}
  \label{si_red_example}
   \vspace{-3mm}
\end{figure}

%we create a new history $\hist' = \tup{\Tr', \wro', \so'}$,
%\begin{itemize}
% \item $\Tr' = \{R_\tr | \tr \in \Tr\} \cup \{W_\tr | \tr \in \Tr\}$
% \item $W_{\tr}$ writes on $\xvar$ if $\tr$ writes on $\xvar$.
% \item For each pair of unique transactions $(\tr_1, \tr_2) \in \Tr \times \Tr$, if $\tr_1, \tr_2$ write on overlapping variables, then $W_{\tr_1}, R_{\tr_2}$ write on a new variable $\xvar_{\tr_2\tr_1}$ $W_{\tr_2}$ reads it from $R_{\tr_2}$. $\wro[\xvar_{\tr_2\tr_1}] = \{\tup{R_{\tr_2}, W_{\tr_1}}\}$.
% \item $\wro[\xvar]' = \{\tup{W_{\tr_1}, R_{\tr_2}} | \tup{\tr_1, \tr_2} \in \wro[\xvar]\}$
% \item $\so' = \{\tup{R_\tr, W_\tr} | \tr \in \Tr\} \cup$
%       
%       $\{\tup{W_{\tr_1}, W_{\tr_2}}, \tup{W_{\tr_1}, R_{\tr_2}}, \tup{R_{\tr_1}, W_{\tr_2}}, \tup{R_{\tr_1}, R_{\tr_2}} | \tup{\tr_1,\tr_2} \in \so \}$ 
%       
%\end{itemize}
%
%Constructing $\Tr', \wro[\xvar]', \so'$ can be done in similar way from prefix consistency reduction. Here, we just need to do a iterations over pair of transaction to check for common write variables and add new $\wro$ relations.
%
%This reduction has exact same $\so'$ from prefix reductions. Therefore, $\so'$ does not have width more than that of $\so$.

Note that $\hist_{R|W}$ and $\hist_{R|W}^c$ have the same width (the session order is defined exactly in the same way), which implies, by Lemma~\ref{lem:pc_width}, that $\hist$ and $\hist_{R|W}^c$ have the same width.

The following result can be proved using similar reasoning as in the case of prefix consistency (see Appendix~\ref{app:si_red}).

\begin{theorem}\label{th:si}
A history $\hist$ satisfies snapshot isolation iff $\hist_{R|W}^c$ is serializable.
\end{theorem}

Note that $\hist_{R|W}^c$ and $\hist$ have the same width, and that $\hist_{R|W}^c$ can be constructed in linear time. Therefore, Theorem~\ref{th:si}, and Corollary~\ref{cor:ser} imply the following result.
 
 \begin{corollary}\label{cor:si}
For an arbitrary but fixed constant $k\in\mathbb{N}$, the problem of checking snapshot isolation for histories of width at most $k$ is polynomial time.
 
 \end{corollary}

\section{Communication graphs}\label{sec:communication}

%$O(\mathsf{size}(h)^{\mathsf{width}(h)}\cdot \mathsf{size}(h)^3)$

In this section, we present an extension of the polynomial time results for PC, SI, and SER, which allows to handle histories where the sharing of variables between different sessions is \emph{sparse}. For the results in this section, we take the simplifying assumption that the session order is a union of transaction sequences (modulo the fictitious transaction writing the initial values), i.e., each transaction sequence corresponding to the standard notion of \emph{session}~\footnote{The results can be extended to arbitrary session orders by considering maximal transaction sequences in session order instead of sessions.}.
We represent the sharing of variables between different sessions using an undirected graph called a \emph{communication graph}. For instance, the communication graph of the history in Figure~\ref{comm_graph_example:1} is given in Figure~\ref{comm_graph_example:2}. For readability, the edges are marked with the variables accessed by the two sessions.

We show that the problem of checking PC, SI, or SER is polynomial time when the size of every \emph{biconnected} component of the communication graph is bounded by a fixed constant. This is stronger than the results in Section~\ref{sec:bounded_width} because the number of biconnected components can be arbitrarily large which means that the total number of sessions is  unbounded. In general, we prove that the time complexity of these consistency criteria is exponential only in the maximum size of such a biconnected component, and not the whole number of sessions.

An undirected graph is biconnected if it is connected and if any one vertex were to be removed, the graph will remain connected, and a biconnected component of a graph $G$ is a maximal biconnected subgraph of $G$. Figure~\ref{comm_graph_example:2} shows the decomposition in biconnected components of a communication graph. This graph contains 5 sessions while every biconnected component is of size at most 3. Intuitively, any potential cycle in the commit order associated to a history will contain a cycle that passes only through sessions in the same biconnected component. Therefore, checking any of these criteria can be done in isolation for each biconnected component (more precisely, on sub-histories that contain only sessions in the same biconnected component). Actually, this decomposition argument works even for RC, RA, and CC. For instance, in the case of the history in Figure~\ref{comm_graph_example:1}, any consistency criterion can be checked looking in isolation at three sub-histories: a sub-history with $S_1$ and $S_2$, a sub-history with $S_2$, $S_3$, and $S_4$, and a sub-history with $S_4$ and $S_5$.

%If a history has a large number of sessions and the sessions sparsely shares variables with each other, we can use communication graph to group sessions to generate smaller sub-histories which we can individually process to verify for $\axser$ of the whole history. TODO SAY THAT IT IS QUITE LIKELY THAN NOT ALL SESSIONS USE THE SAME VARIABLES.
Formally, a \emph{communication graph} of a history $\hist$ is an undirected graph $\mathsf{Comm}(\hist)=(V,E)$ where the set of vertices $V$ is the set of sessions in $\hist$~\footnote{The transaction writing the initial values is considered as a distinguished session.}, and $(v,v')\in E$ iff the sessions $v$ and $v'$ contain two transactions $\tr_1$ and $\tr_2$, respectively, such that $\tr_1$ and $\tr_2$ read or write a common variable $\xvar$. 

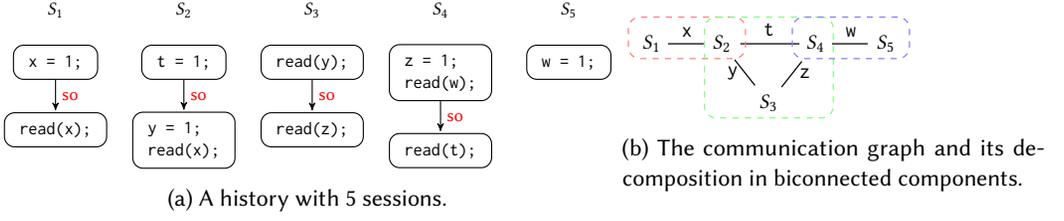
\begin{figure}
  \begin{subfigure}{.59\textwidth}
  \resizebox{\textwidth}{!}{
  \begin{tikzpicture}[->,>=stealth',shorten >=1pt,auto,node distance=3cm,
   semithick, transform shape]
   % \node[draw, rounded corners=2mm] (t1) at (0, 0) {\begin{tabular}{l} \texttt{x = 1;} \\ \texttt{y = 1;}\end{tabular}};
   \node[draw=black!0] (s1) at (0, 0) {$S_1$};
   \node[draw, rounded corners=2mm] (t11) at (0, -1) {\begin{tabular}{l} \texttt{x = 1;} \end{tabular}};
  \node[draw, rounded corners=2mm] (t12) at (0, -2.3) {\begin{tabular}{l} \texttt{read(x);} \end{tabular}};
  \path (t11) edge node {$\so$} (t12);
  
  \node[draw=black!0] (s2) at (2.5, 0) {$S_2$};
  \node[draw, rounded corners=2mm] (t21) at (2.5, -1) {\begin{tabular}{l} \texttt{t = 1;} \end{tabular}};
  \node[draw, rounded corners=2mm] (t22) at (2.5, -2.5) {\begin{tabular}{l} \texttt{y = 1;} \\ \texttt{read(x);} \end{tabular}};
  \path (t21) edge node {$\so$} (t22);
 
 \node[draw=black!0] (s3) at (5, 0) {$S_3$};
 \node[draw, rounded corners=2mm] (t31) at (5, -1) {\begin{tabular}{l} \texttt{read(y);} \end{tabular}};
\node[draw, rounded corners=2mm] (t32) at (5, -2.3) {\begin{tabular}{l} \texttt{read(z);} \end{tabular}};
\path (t31) edge node {$\so$} (t32);

\node[draw=black!0] (s4) at (7.5, 0) {$S_4$};
\node[draw, rounded corners=2mm] (t41) at (7.5, -1.2) {\begin{tabular}{l} \texttt{z = 1;} \\ \texttt{read(w);} \end{tabular}};
\node[draw, rounded corners=2mm] (t42) at (7.5, -2.7) {\begin{tabular}{l} \texttt{read(t);} \end{tabular}};
\path (t41) edge node {$\so$} (t42);

\node[draw=black!0] (s5) at (10, 0) {$S_5$};
\node[draw, rounded corners=2mm] (t51) at (10, -1) {\begin{tabular}{l} \texttt{w = 1;} \end{tabular}};
  % \node[draw, rounded corners=2mm] (t3) at (1.5, 0) {\begin{tabular}{l} \texttt{read(x); // 2} \\ \texttt{read(y); // 1} \end{tabular}};
  % \path (t1) edge node {} (t3);
  % \path (t2) edge node {$\so$} (t3);
  % \path (t1) edge node {} (t2);
  % \path (t3_1) edge node {$\po$} (t3_2);
  \end{tikzpicture}  
  }
   \caption{A history with 5 sessions.}
   \label{comm_graph_example:1}
  \end{subfigure}
%  \hspace{1mm}
    \begin{subfigure}{.4\textwidth}
    \resizebox{.7\textwidth}{!}{
    \begin{tikzpicture}[->,>=stealth',shorten >=1pt,auto,node distance=3cm,
     semithick, transform shape]
     \node[draw=black!0] (s1) at (0, 0) {$S_1$};
  
    \node[draw=black!0] (s2) at (1.2, 0) {$S_2$};

   \node[draw=black!0] (s3) at (2, -1) {$S_3$};
  
  \node[draw=black!0] (s4) at (2.8, 0) {$S_4$};
  
  \node[draw=black!0] (s5) at (4, 0) {$S_5$};
  \path (s1) edge[-] node {\texttt{x}} (s2);
  \path (s2) edge[-, left] node {\texttt{y}} (s3);
  \path (s3) edge[-, right] node {\texttt{z}} (s4);
  \path (s4) edge[-, above] node {\texttt{t}} (s2);
  \path (s4) edge[-] node {\texttt{w}} (s5);

  \node[draw=red!50, dashed, rounded corners=2mm, minimum width=2cm, minimum height=.7cm] () at (.6, .1) {};
  \node[draw=green!50, dashed, rounded corners=2mm, minimum width=2.2cm, minimum height=1.7cm] () at (2, -.4) {};
  \node[draw=blue!50, dashed, rounded corners=2mm, minimum width=2cm, minimum height=.7cm] () at (3.4, .1) {};

  \end{tikzpicture}  
    }
     \caption{The communication graph and its decomposition in biconnected components.}
     \label{comm_graph_example:2}
    \end{subfigure}
    
     \vspace{-3mm}
  \caption{A history and its communication graph.}
  \label{comm_graph_example}
   \vspace{-3mm}
\end{figure}

%TODO GIVE AN EXAMPLE OF A HISTORY AND ITS COMMUNICATION GRAPH, AND THE BICONNECTED COMPONENTS.
%For a biconnected component $C$ of $\mathsf{Comm}(\hist)$, the transactions in $C$ are the transaction

%Given a graph $G=(V,E)$ and a biconnected component $C$ of $G$, we say that 

\begin{lemma}\label{lem:comm_graph}
Let $C_1$,$\ldots$,$C_n$ be the biconnected components of $\mathsf{Comm}(\hist)$ for a history $\hist = \tup{T, \wro, \so}$. Let $P_A$ be a path of the form of type $A$ connecting two transactions of $C_i$~\footnote{That is, transactions that are included in the sessions in $C_i$.} Then, there is a path $P_B$ of the form of type $B$ connecting the same two transactions and $P_B$ never leaves $C_i$.
\end{lemma}
%\begin{proof}
%\renewcommand{\qedsymbol}{}
 \textsc{Proof.} Type $A$ and $B$ are both of the form $\CO^+$. Consider a minimal path $\pi=\tr_0,\ldots,\tr_n$ in $\bigcup_i \co_i$ between two transactions $\tr_0$ and $\tr_n$ of the same biconnected component $C$ of $\mathsf{Comm}(\hist)$ (i.e., from sessions in $C$). We define a path $\pi_s=v_0,\ldots,v_m$ between sessions, i.e., vertices of $\mathsf{Comm}(\hist)$, which contains an edge $(v_j,v_{j+1})$ iff $\pi$ contains an edge $(\tr_i,\tr_{i+1})$ with $\tr_i$ a transaction of session $v_j$ and $\tr_{i+1}$ a transaction of session $v_{j+1}\neq v_j$. Since any graph decomposes to a forest of biconnected components, this path must necessarily leave and enter some biconnected component $C_1$ to and from the same biconnected component $C_2$, i.e., $\pi_s$ must contain two vertices $v_{j_1}$ and $v_{j_2}$ in $C_1$ such that the successor $v_{j_1+1}$ of $v_{j_1}$ and the predecessor $v_{j_2-1}$ of $v_{j_2}$ are from $C_2$. Let $\tr_1$, $\tr_2$, $\tr_3$, $\tr_4$ be the transactions in the path $\pi$ corresponding to $v_{j_1}$, $v_{j_2}$, $v_{j_1+1}$, and $v_{j_2-1}$, respectively. Now, since any two biconnected components share at most one vertex, it follows that $t_3$ and $t_4$ are from the same session and
  \vspace{-1mm}
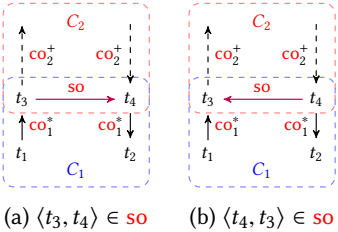
\begin{wrapfigure}{l}{0.4\textwidth} 
  \centering
  \begin{subfigure}{.15\textwidth}
   \resizebox{\textwidth}{!}{
    \begin{tikzpicture}[->,>=stealth',shorten >=1pt,auto,node distance=4cm,
      semithick, transform shape]
     \node[transaction state] at (0,0)       (t_3)           {$\tr_3$};
     \node[transaction state] at (0,1.5)       (t_5)           {};
     \node[transaction state] at (2,0)       (t_4)           {$\tr_4$};
     \node[transaction state] at (2,1.5)       (t_6)           {};
     \node[transaction state] at (0,-1)       (t_1)           {$\tr_1$};
     \node[transaction state] at (2,-1)       (t_2)           {$\tr_2$};
     \node[transaction state,text=red] at (1, 1.4)       ()           {$C_2$};
     \node[transaction state,text=blue] at (1, -1.3)       ()           {$C_1$};
     \node[draw=red!50, dashed, rounded corners=2mm, minimum width=2.7cm, minimum height=2cm] () at (1, .75) {};
     \node[draw=blue!50, dashed, rounded corners=2mm, minimum width=2.7cm, minimum height=2cm] () at (1, -.6) {};
     \path (t_3) edge[dashed,right] node {$\co^+_2$} (t_5); %bend left=90
     \path (t_1) edge node[right] {$\co^*_1$} (t_3);
     \path (t_4) edge node[left] {$\co^*_1$} (t_2);
     \path (t_6) edge[dashed,left] node {$\co^+_2$} (t_4); %bend left=90
     \path (t_3) edge[color=purple] node[above] {$\so$} (t_4);
    \end{tikzpicture}  
   }
   \caption{$\tup{\tr_3,\tr_4} \in \so$}
   \label{comm_graph_proof:1a}
  \end{subfigure}\hspace{3mm}
  \begin{subfigure}{.15\textwidth}
   \resizebox{\textwidth}{!}{
    \begin{tikzpicture}[->,>=stealth',shorten >=1pt,auto,node distance=4cm,
      semithick, transform shape]
     \node[transaction state] at (0,0)       (t_3)           {$\tr_3$};
     \node[transaction state] at (0,1.5)       (t_5)           {};
     \node[transaction state] at (2,0)       (t_4)           {$\tr_4$};
     \node[transaction state] at (2,1.5)       (t_6)           {};
     \node[transaction state] at (0,-1)       (t_1)           {$\tr_1$};
     \node[transaction state] at (2,-1)       (t_2)           {$\tr_2$};
     \node[transaction state,text=red] at (1, 1.4)       ()           {$C_2$};
     \node[transaction state,text=blue] at (1, -1.3)       ()           {$C_1$};
     \node[draw=red!50, dashed, rounded corners=2mm, minimum width=2.7cm, minimum height=2cm] () at (1, .75) {};
     \node[draw=blue!50, dashed, rounded corners=2mm, minimum width=2.7cm, minimum height=2cm] () at (1, -.6) {};
     \path (t_3) edge[dashed,right] node {$\co^+_2$} (t_5); %bend left=90
     \path (t_1) edge node[right] {$\co^*_1$} (t_3);
     \path (t_4) edge node[left] {$\co^*_1$} (t_2);
     \path (t_6) edge[dashed,left] node {$\co^+_2$} (t_4); %bend left=90
     \path (t_4) edge[color=purple] node[above] {$\so$} (t_3);
    \end{tikzpicture}  
   }
   \caption{$\tup{\tr_4,\tr_3} \in \so$}
   \label{comm_graph_proof:1b}
  \end{subfigure}
   \vspace{-2mm}
  \caption{Minimal paths between transactions in the same biconnected component.}
  \label{comm_graph_proof:1}
   \vspace{-2mm}
 \end{wrapfigure}
 \begin{itemize}
  \item if $\tup{\tr_3, \tr_4} \in \so$, then there exists a smaller path between $\tr_0$ and $\tr_1$ that uses the $\so$ relation between $\tup{\tr_3, \tr_4}$ (we recall that $\so\subseteq \bigcup_i \co_i$) instead of the transactions in $C_2$, pictured in Figure~\ref{comm_graph_proof:1a}, which is a contradiction to the minimality of $\pi$,
  \item if $\tup{\tr_4, \tr_3} \in \so$, then, we have a cycle in $\bigcup_i \co_i\cup\so$, pictured in Figure~\ref{comm_graph_proof:1b}, which is also a contradiction.
 \end{itemize} 
 
  Type $A$ and $B$ are of the form $(\wro \cup \so)^+$. ``shortening'' a bigger $(\wro \cup \so)^+$ path (they are also a path in $\CO^+$ since $(\wro \cup \so)^+ \subseteq \CO$) will introduce only $\so$ dependencies. So a minimal $(\wro \cup \so)^+$ never leaves a bicomponent.
  
  Type $A$ and $B$ are of the form $\CO^* \circ (\wro \cup \so)$. Similar to last the case, ``shortening'' a bigger path will introduce only $\so$ dependencies. If the new $\so$ is at the end of the path then the new path is still of the form $\CO^* \circ (\wro \cup \so)$. Else, we can replace $\so$ with $\CO$ to make a path of the form $\CO^* \circ (\wro \cup \so)$. So a minimal $\CO^* \circ (\wro \cup \so)$ never leaves a bicomponent.
  
  Type $A$ is of the form $\CO^* \circ \CO$ where the last $\CO$ dependency is between two transactions writing on same variable and type $B$ is of the same form of type $A$ or $\CO^* \circ \so$. Similar to the previous cases, ``shortening'' a bigger path will introduce only $\so$ dependencies. If the new $\so$ is at the end, then it becomes of the form of $\CO^* \circ \so$. Else, we can replace $\so$ with $\CO$ to make a path of the form $\CO^* \circ \CO$ where the last $\CO$ dependency is the one from the original path. So a minimal path of type $A$ never leaves a bicomponent or it can be ``shortened'' to a minimal path of the form $\CO^* \circ \so$ which never leaves a bicomponent. \hfill $\Box$

For a history $\hist = (T, \so, \wro)$ and biconnected component $C$ of $\mathsf{Comm}(\hist)$, the projection of $\hist$ over transactions in sessions of $C$ is denoted by $h\downarrow C$, i.e., $h\downarrow C=(T', \so', \wro')$ where $T'$ is the set of transactions in sessions of $C$, $\so'$ and $\wro'$ are the projections of $\so$ and $\wro$, respectively, on $T'$.

\begin{theorem}\label{th:comm_graph}
For any criterion $X\in\{\text{RA},\text{RC},\text{CC},\text{PC},\text{SI},\text{SER}\}$, a history $h$ satisfies $X$ iff for every biconnected component $C$ of $\mathsf{Comm}(\hist)$, $h\downarrow C$ satisfies $X$.
% Bi-connected components of a communication graph of a history are consistent to a consistency model if and only if the whole history is consistent to that consistency model.
\end{theorem}
\begin{proof}
The ``only-if'' direction is obvious. For the ``if'' direction, let $C_1$,$\ldots$,$C_n$ be the biconnected components of $\mathsf{Comm}(\hist)$. Also, let $\co_i$ be the commit order that witnesses that $h\downarrow C_i$ satisfies $X$, for each $1\leq i\leq n$. The union $\bigcup_i\co_i$ is acyclic since otherwise, any minimal cycle would be a minimal path between transactions of the same biconnected component $C_j$, and, by Lemma~\ref{lem:comm_graph}, it will include only transactions of $C_j$ which is a contradiction to $\co_j$ being a total order. We show that any linearization $\co$ of $\bigcup_i\co_i$ along with $h$ satisfies the axioms of $X$, for every consistency model $X$. 

The axioms defining RA, RC, CC, PC, and SER involve transactions that write or read a common variable, which implies that they belong to the same biconnected component. For CC, resp., PC, SER using the result from Lemma~\ref{lem:comm_graph}, a minimal $(\wro \cup \so)^+$, resp., $\CO^* \circ (\wro \cup \so)^+$, $\co$ path from $\tr_2$ to $\tr_3$ in Figure~\ref{cc_def}, resp., Figure~\ref{pre_def}, Figure~\ref{ser_def} will include transactions in the same biconnected component as $\tr_2$ and $\tr_3$, since ``shortening'' a bigger path will introduce only $\so$ dependencies. Therefore, they must be satisfied by $\co$. 

% Concerning the axiom defining PC in Figure~\ref{pre_def}, the transactions $\tr_1$, $\tr_2$, and $\tr_3$ belong to the same biconnected component $C$ (since they all read or write a common variable $\xvar$). Then, using again a similar reasoning as in Lemma~\ref{lem:comm_graph}, it can be proved that a minimal $\co^* \circ (\wro \cup \so)$ path from $\tr_2$ to $\tr_3$ will contain only transactions of $C$. Therefore, this axiom must be satisfied by $\co$.
%Therefore, applying Lemma~\ref{lem:comm_graph}, a minimal path from $\tr_2$ to $\tr_4$ will include only transactions from $C$. 

For SI, it's only necessary to discuss the axiom $\mathsf{Conflict}$ (the satisfaction of $\mathsf{Prefix}$ is proved as for PC). Following Figure~\ref{confl_def} and Lemma~\ref{lem:comm_graph}, any $\CO^* \circ \CO$ path connecting $\tr_2$, and $\tr_3$ where the last $\CO$ dependency is between two transactions writing on same variables, has either a minimal path of the same form that never leaves the same bicomponent as $\tr_2, \tr_3$, or a minimal path of the form $\CO^* \circ \so$. Since the bicomponent satisfies SI or particularly $\mathsf{Prefix}$ and $\mathsf{Conflict}$, for all possible $\tr_1$, $\tr_1, \tr_2, \tr_3$, it must satisfy $\mathsf{Conflict}$.
Note that $\tr_3$ cannot be the first transaction in its session because a path from $\tr_2$ to $\tr_3$ passing through $\tr_4$ (which belongs to a different biconnected component) will necessarily have to pass twice through $\tr_3$ which would imply that $\co$ is cyclic. Thus, $\mathsf{Conflict}$ must be satisfied by $\co$.
\end{proof}

Since the decomposition of a graph into biconnected components can be done in linear time, Theorem~\ref{th:comm_graph} implies that any of the criteria \text{PC}, \text{SI}, or \text{SER} can be checked in time $O(\mathsf{size}(h)^{\mathsf{bi\text{-}size}(h)}\cdot \mathsf{size}(h)^3\cdot \mathsf{bi\text{-}nb}(h))$ where $\mathsf{bi\text{-}size}(h)$ and $\mathsf{bi\text{-}nb}(h)$ are the maximum size of a biconnected component in $\mathsf{Comm}(\hist)$ and the number of biconnected components of $\mathsf{Comm}(\hist)$, respectively. The following corollary is a direct consequence of this observation.

\begin{corollary}
For an arbitrary but fixed constant $k\in\mathbb{N}$ and any criterion $X\in\{\text{PC},\text{SI},\text{SER}\}$, the problem of checking if a history $\hist$ satisfies $X$ is polynomial time, provided that the size of every biconnected component of $\mathsf{Comm}(\hist)$ is bounded by $k$.
\end{corollary}

%Also, it is known, the decomposition of a graph into biconnected components can be done in linear time. So even if a history has a large number of sessions, if the communication graph has a bounded number of biconnected components. We can decompose the history to subhistories with fewer sessions and process each subhistory separately.

\section{Experimental Evaluation}\label{sec:exp}

\begin{figure}
\centering
 \begin{subfigure}{.33\textwidth}
  \resizebox{\textwidth}{!}{
   \includegraphics{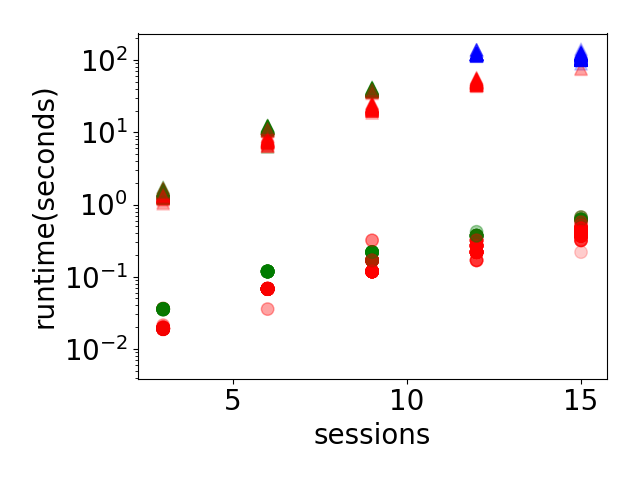}
  }
  \caption{Sessions.}
  \label{ser_node_scale}
 \end{subfigure}
 \begin{subfigure}{.33\textwidth}
  \resizebox{\textwidth}{!}{
   \includegraphics{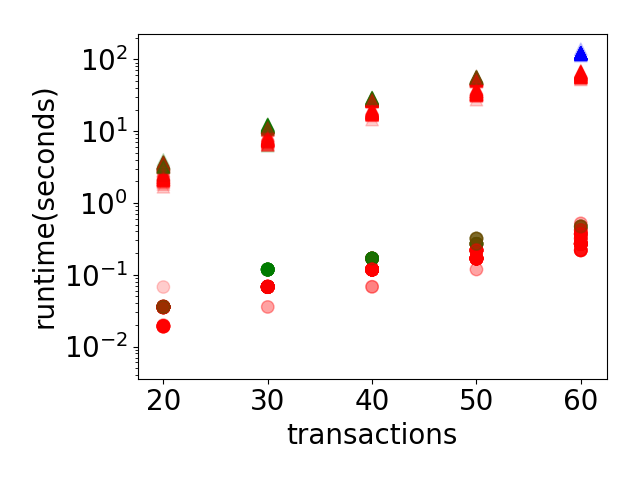}
  }
  \caption{Transactions per session.}
  \label{ser_transaction_scale}
 \end{subfigure}
 \hspace{-2mm}
 \begin{subfigure}{.33\textwidth}
  \resizebox{\textwidth}{!}{
   \includegraphics{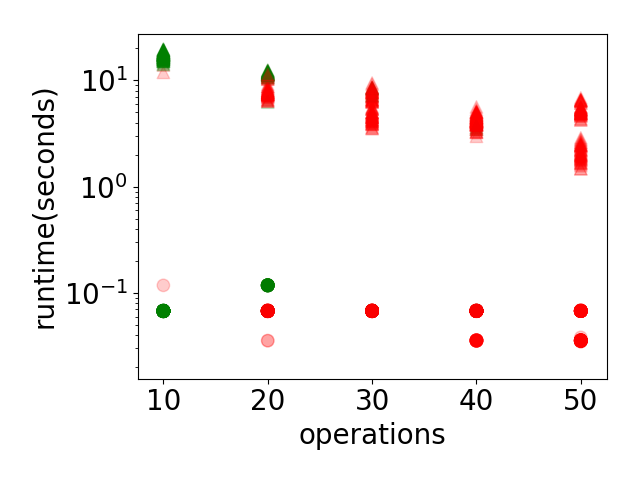}
  }
  \caption{Operations per transaction.}
  \label{ser_operation_scale}
 \end{subfigure}
 \hspace{-3mm}
 \begin{subfigure}{.33\textwidth}
  \resizebox{\textwidth}{!}{
   \includegraphics{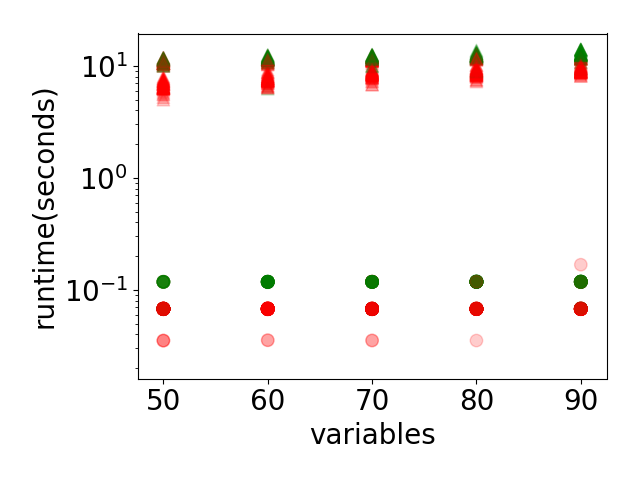}
  }
  \caption{Variables.}
  \label{ser_variable_scale}
 \end{subfigure}
 \vspace{-3mm}
 \caption{Scalability of our Serializability checking algorithm in Section~\ref{ssec:ser_checking} and a comparison to a SAT encoding. The circular, resp., triangle, dots represent wall clock times of our algorithm, resp., the SAT encoding. The x-axis represents the varying parameter while the y-axis represents the wall clock time in logarithmic scale. The red, green, and blue dots represent invalid, valid and resource exhausted instances, respectively.}
 \label{ser_performace_scale}
 \vspace{-3mm}
\end{figure}

To demonstrate the practical value of the theory developed in the previous sections, we argue that our algorithms:
\begin{itemize}
 \item are efficient and scalable, and they outperform a SAT encoding of the axioms in Section~\ref{sec:def},
 \item enable an effective testing framework allowing to expose consistency violations in production databases.
\end{itemize}

The SAT encoding we considered is based on the axioms we presented on this paper. We represent each binary relation with a propositional variable and encode the axioms using clauses to create a 3-SAT formula, which is passed to a SAT solver.

For each ordered pair of transactions $\tr_1$, $\tr_2$ we add two propositional variables representing $\tup{\tr_1,\tr_2} \in (\wro \cup \so)^+$ and $\tup{\tr_1,\tr_2} \in \CO$, respectively. Then we generate clauses corresponding to:
\begin{itemize}
\item Singleton clauses defining the relation $\wro \cup \so$ (extracted from the input history),
\item $\tup{\tr_1, \tr_2} \in \wro \cup \so$ implies $\tup{\tr_1, \tr_2} \in \CO$.
\item $\CO$ being a total order.
\item The axioms corresponding to the considered consistency model.
\end{itemize}

This is an optimization that does not encode $\wro$ and $\so$ separately, which is sound because of the shape of our axioms (and because these relations are fixed apriori).

As for the implementation of our algorithms, we used standard programming optimizations, \eg efficient data structures, to avoid unnecessary runtime and memory usage. Few examples include:
\begin{itemize}
\item using efficient hashsets for searching in a set,
\item grouping transactions which access the same variable (because our algorithms usually iterate over transactions accessing the same variable).
\item reducing PC\/SI to SER on-the-fly while traversing the history.
\end{itemize}

We focus on three of the criteria introduced in Section~\ref{sec:def}: \emph{serializability} which is NP-complete in general and polynomial time when the number of sessions is considered to be a constant, \emph{snapshot isolation} which can be reduced in linear time on-the-fly to serializability, and \emph{causal consistency} which is polynomial-time in general~\footnote{Our implementation is publicly available. URL omitted to maintain anonymity.}. As benchmark, we consider histories extracted from three distributed databases: CockroachDB~\cite{cockroach}, Galera~\cite{galera}, and AntidoteDB~\cite{antidote}. %~\footnote{The databases are deployed using Docker.}.
% snapshot isolation, and causal consistency, and histories extracted from three distributed databases, CockroachDB~\cite{cockroach}, Galera~\cite{galera}, and AntidoteDB~\cite{antidote}.
% which implement serializability~\cite{cockroach-claim}, snapshot isolation~\cite{galera-claim} and causal consistency~\cite{antidote-claim}, respectively (with the default configuration). Therefore, all experiments concerning SER, SI, or CC use histories of CockroachDB, Galera, and AntidoteDB, respectively.
%
% how the histories are generated and executed
%sessions are uniformly distributed among all nodes of the considered distributed database
Following the approach in Jepsen~\cite{jepsen},  histories are generated with random clients. For the experiments described hereafter, the randomization process is parametrized by: (1) the number of sessions ({\bf \#sess}), (2) the number of transactions per session ({\bf \#trs}), (3) the number of operations per transaction ({\bf \#ops}), and (4) an upper bound on the number of used variables ({\bf \#vars})~\footnote{We ensure that every value is written at most once.}. For any valuation of these parameters, half of the histories generated with CockroachDB and Galera are restricted such that the sets of variables written by any two sessions are disjoint (the sets of read variables are not constrained). This restriction is used to increase the frequency of valid histories. 
%We also insert a random pause of at most 200 milliseconds between every two transactions within the same session.
%\begin{itemize}
% \item A session is generated by choosing uniformly at random the type of the operations in each transaction (read or write), the variable accessed by each operation, and a value for each write operation. We ensure that every value is written at most once by a client by maintaining a counter map for each variable. 
% \item Each session writes on specific non-overlapping sets of variables of equal size. Otherwise, everything is same as before. A read and write operation are chosen randomly.
%\end{itemize}
%We consider the second type of histories is to increase the number of valid histories in CockroachDB and Galera, because they often produce inconsistent histories. Each history are generated separately as first and then it is executed on the mentioned databases. We used Docker to deploy these distributed instances. 
%Also, we insert a random pause of at most 200 milliseconds between every two transactions within the same session so that the network does not clog up. 

%Then, each executed history is verified with 10 minutes of time limit, 10GB of memory limit and 10GB of file size limit(to avoid big CNF files).

% ser
In a first experiment, we investigated the efficiency of our serializability checking algorithm (Section \ref{ssec:ser_checking}) and we compared its performance with a direct SAT encoding of the serializability definition in Section~\ref{sec:def} (we used MiniSAT~\cite{DBLP:conf/sat/EenS03} to solve the SAT queries). We used histories extracted from CockroachDB which claims to implement serializability, acknowledging however the possibility of anomalies~\cite{cockroach-claim}. The sessions of a history are uniformly distributed among 3 nodes of a single cluster. To evaluate scalability, we fix a reference set of parameter values: {\bf \#sess}=6, {\bf \#trs}=30, {\bf \#ops}=20, and {\bf \#vars} = 60 $\times$ {\bf \#sess}, and vary only one parameter at a time. For instance, the number of sessions varies from 3 to 15 in increments of 3. 
 We consider 100 histories for each combination of parameter values. The experimental data is reported in Figure~\ref{ser_performace_scale}. Our algorithm scales well even when increasing the number of sessions, which is not guaranteed by its worst-case complexity (in general, this is exponential in the number of sessions). Also, our algorithm is at least two orders of magnitude more efficient than the SAT encoding. We have fixed a 10 minutes timeout, a limit of 10GB of memory, and a limit of 10GB on the files containing the formulas to be passed to the SAT solver. The blue dots represent resource exhausted instances. The SAT encoding reaches the file limit for 148 out of 200 histories with at least 12 sessions (Figure~\ref{ser_node_scale}) and for 50 out of 100 histories with 60 transactions per session (Figure~\ref{ser_transaction_scale}), the other parameters being fixed as explained above. % .(i.e.,  {\bf \#sess}=6, {\bf \#ops}=20, and {\bf \#vars} = 360).

We have found a large number of violations, whose frequency increases with the number of sessions, transactions per session, or operations per transaction, and decreases when allowing more variables. This is expected since increasing any of the former parameters increases the chance of interference between different transactions while increasing the latter has the opposite effect. The second and third column of Table~\ref{violation_stat} give a more precise account of the kind of violations we found by identifying for each criterion X, the number of histories which violate X but no other criterion weaker than X, e.g., there is only one violation to SI which satisfies PC.

\begin{figure}
 \begin{subfigure}{.33\textwidth}
  \resizebox{\textwidth}{!}{
  \includegraphics{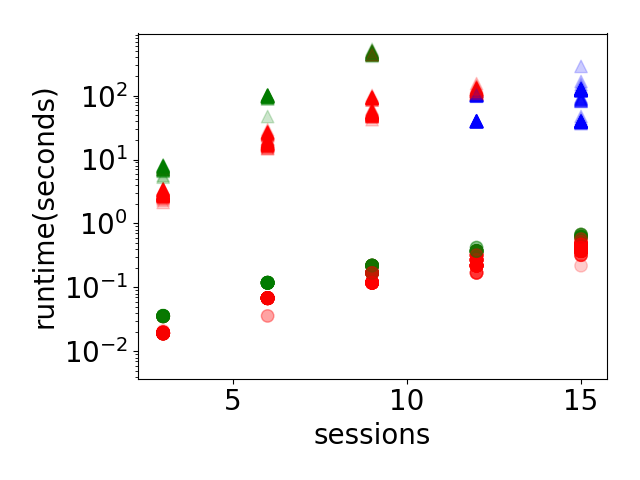}
  }
  \caption{CockroachDB (SI)}
  \label{roach_si_node_scale}
 \end{subfigure}
 \hspace{-2mm}
 \begin{subfigure}{.33\textwidth}
  \resizebox{\textwidth}{!}{
   \includegraphics{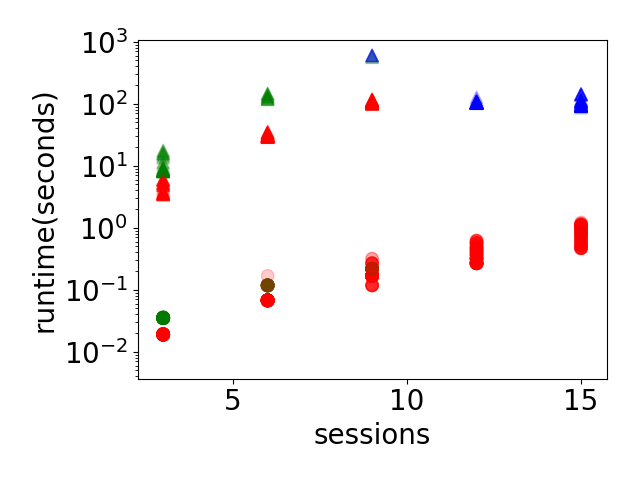}
  }
  \caption{Galera (SI)}
  \label{galera_si_node_scale}
 \end{subfigure}
  \hspace{-2mm}
 \begin{subfigure}{.33\textwidth}
  \resizebox{\textwidth}{!}{
   \includegraphics{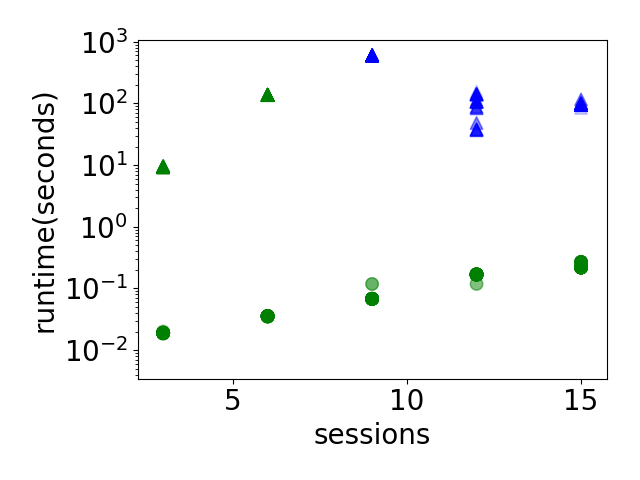}
  }
  \caption{AntidoteDB (CC)}
  \label{cc_session_scale}
 \end{subfigure}
 \vspace{-2mm}
 \caption{Scalability of SI and CC checking, and a comparison to a SAT encoding.}
 \label{si_cc_performace_scale}
 \vspace{-3mm}
\end{figure}

The second experiment measures the scalability of the SI checking algorithm obtained by applying the reduction to SER described in Section~\ref{ssec:si} followed by the SER checking algorithm in Section~\ref{ssec:ser_checking}, and its performance compared to a SAT encoding of SI. We focus on its behavior when increasing the number of sessions (varying the other parameters leads to similar results). As benchmark, we used the same CockroachDB histories as in Figure~\ref{ser_node_scale} and a number of histories extracted from Galera\footnote{In order to increase the frequency of valid histories, all sessions are executed on a single node.} whose documentation contains contradicting claims about whether it implements snapshot isolation~\cite{galera-claim,galera-notclaim}. We use 100 histories per combination of parameter values as in the previous experiment. The results are reported in Figure~\ref{roach_si_node_scale} and Figure~\ref{galera_si_node_scale}. We observe the same behavior as in the case of SER. In particular, the SAT encoding reaches the file limit for 150 out of 200 histories with at least 12 sessions in the case of the CockroachDB histories, and for 162 out of 300 histories with at least 9 sessions in the case of the Galera histories. The last two columns in Table~\ref{violation_stat} classify the set of violations depending on the weakest criterion that they violate.

We also evaluated the performance of the CC checking algorithm in Section~\ref{sec:general} when increasing the number of sessions, on histories extracted from AntidoteDB, which claims to implement causal consistency~\cite{antidote-claim}. The results are reported in Figure~\ref{cc_session_scale}. In this case, the SAT encoding reaches the file limit for 150 out of 300 histories with at least 9 sessions. All the histories considered in this experiment are valid. However, when experimenting with other parameter values, we have found several violations. The smallest parameter values for which we found violations were 3 sessions, 14 transactions per session, 14 operations per transaction, and 5 variables. The violations we found are also violations of Read Atomic. For instance, one of the violations contains two transactions $\tr_1$ and $\tr_2$, each of them writing to two variables $x_1$ and $x_2$, and another transaction $\tr_3$ which reads $x_1$ from $\tr_1$ and $x_2$ from $\tr_2$ ($\tr_1$ and $\tr_2$ are from different sessions while $\tr_3$ is an $\so$ successor of $\tr_1$ in the same session). These violations are novel and they were confirmed by the developers of AntidoteDB.

The refinement of the algorithms above based on communication graphs, described in Section~\ref{sec:communication}, did not have a significant impact on their performance. The histories we generated contained few biconnected components (many histories contained just a single biconnected component) which we believe is due to our proof of concept deployment of these databases on a single machine that did not allow to experiment with very large number of sessions and variables. 

%Lastly, we used our bicomponent technique to scale our algorithm for large histories with relatively less shared variables across the sessions. We present the plot for the runtime of the algorithm with and without bicomponent decomposition in Figure~\ref{bic_plot}.

\begin{table}
\small{
 \begin{tabular}{|l|c|c|c|c|c|}
  \hline
  & \multicolumn{2}{c|}{Serializability checking} & \multicolumn{2}{c|}{Snapshot Isolation checking} \\
  \hline
  Weakest                  & CockroachDB     & CockroachDB     & Galera      & Galera                \\
  criterion violated         & (disjoint writes) & (no constraints) & (disjoint writes) & (no constraints)  \\
  \hline
  Read Committed     &             &             & 19          & 50              \\
  Read Atomic        & 180         & 547         & 91          & 139                \\
  Causal Consistency            & 339         & 382         & 88          & 43                  \\
  Prefix Consistency            & 2           & 7           &             &                    \\
  Snapshot Isolation &             & 1           &             & 1                   \\
  Serializability      & 25          &             &            &                     \\
  \hline
%  Serializable      & 454         & 63          & 48          & 16                \\
  \hline
  Total number of violations           & 546/1000        & 937/1000        & 198/250         & 233/250              \\
  \hline
 \end{tabular}
 }
 \vspace{1mm}
 \caption{Violation statistics. The ``disjoint writes'' columns refer to histories where the set of variables written by any two sessions are disjoint.}
 \label{violation_stat}
\end{table}

\section{Related Work}

\citet{DBLP:conf/concur/Cerone0G15} give the first formalization of the criteria we consider in this paper, using the specification methodology of \citet{DBLP:conf/popl/BurckhardtGYZ14}. This formalization uses two auxiliary relations, a \emph{visibility} relation which represents the fact that a transaction ``observes'' the effects of another transaction and a \emph{commit order}, also called arbitration order, like in our case. Executions are abstracted using a notion of history that includes only a session order and the adherence to some consistency criterion is defined as the existence of a visibility relation and a commit order satisfying certain axioms. Motivated by practical goals, our histories include a write-read relation, which enables more uniform and in our opinion, more intuitive, axioms to characterize consistency criteria. Moreover, \citet{DBLP:conf/concur/Cerone0G15} do not investigate algorithmic issues as in our paper.

\citet{DBLP:journals/jacm/Papadimitriou79b} showed that checking serializability of an execution is NP-complete. Moreover, it identifies a stronger criterion called \emph{conflict serializability} which is polynomial time checkable. Conflict serializability assumes that histories are given as sequences of operations and requires that the commit order be consistent with a \emph{conflict-order} between transactions defined based on this sequence (roughly, a transaction $\tr_1$ is before a transaction $\tr_2$ in the conflict order if it accesses some variable $\xvar$ before $\tr_2$ does). This result is not applicable to distributed databases where deriving such a sequence between operations submitted to different nodes in a network is impossible.

\citet{DBLP:conf/popl/BouajjaniEGH17} showed that checking several variations of causal consistency on executions of a \emph{non-transactional} distributed database is polynomial time (they also assume that every value is written at most once). 
%The also rely on the specification framework in \citet{DBLP:conf/popl/BurckhardtGYZ14}. 
Assuming singleton transactions, our notion of CC corresponds to the causal convergence criterion in~\citet{DBLP:conf/popl/BouajjaniEGH17}. Therefore, our result concerning CC can be seen as an extension of this result concerning causal convergence to transactions.

There are some works that investigated the problem of checking consistency criteria like sequential consistency and linearizability in the case of shared-memory systems. \citet{DBLP:journals/siamcomp/GibbonsK97} showed that checking linearizability of the single-value register type is NP-complete in general, but polynomial time for executions where every value is written at most once. Using a reduction from serializabilty, they showed that checking sequential consistency is NP-complete even when every value is written at most once. \citet{DBLP:journals/pacmpl/EmmiE18} extended the result concerning linearizability to a series of abstract data types called collections, that includes stacks, queues, key-value maps, etc.  sequential consistency is g serializability 

The notion of \emph{communication graph} is inspired by the work of \citet{DBLP:journals/pacmpl/ChalupaCPSV18} which investigates partial-order reduction (POR) techniques for multi-threaded programs. In general, the goal of partial-order reduction~\cite{DBLP:conf/popl/FlanaganG05} is to avoid exploring executions which are equivalent w.r.t. some suitable notion of equivalence, e.g., Mazurkiewicz trace equivalence~\cite{DBLP:conf/ac/Mazurkiewicz86}. They use the acyclicity of communication graphs to define a class of programs for which their POR technique is optimal. The algorithmic issues they explore are different than ours and they don't investigate biconnected components of this graph as in our results.

\section{Conclusions}

Our results provide an effective means of checking the correctness of transactional databases with respect to a wide range of consistency criteria, in an efficient way. We devise a new specification framework for these criteria, which besides enabling efficient verification algorithms, provide a novel understanding of the differences between them in terms of set of transactions that \emph{must} be committed before a transaction which is read during the execution. These algorithms are shown to be scalable and orders of magnitude more efficient than standard SAT encodings of these criteria (as defined in our framework). While the algorithms are quite simple to understand and implement, the proof of their correctness is non-trivial and benefits heavily from the new specification framework. One important venue for future work is identifying root causes for a given violation. The fact that we are able to deal with a wide range of criteria is already helpful in identifying the weakest criterion that is violated in a given execution.  Then, in the case of RC, RA, and CC, where inconsistencies correspond to cycles in the commit order, the root cause could be attributed to a minimal cycle in this relation. We did this in our communication with the Antidote developers to simplify the violation we found which contained 42 transactions. In the case of PC, SI, and SER, it could be possible to implement a search procedure similar to CDCL in SAT solvers, in order to compute the root-cause as a SAT solver would compute an unsatisfiability core. 

%%% Acknowledgments
%\begin{acks}                            %% acks environment is optional
%                                        %% contents suppressed with 'anonymous'
%  %% Commands \grantsponsor{<sponsorID>}{<name>}{<url>} and
%  %% \grantnum[<url>]{<sponsorID>}{<number>} should be used to
%  %% acknowledge financial support and will be used by metadata
%  %% extraction tools.
%  This material is based upon work supported by the
%  \grantsponsor{GS100000001}{National Science
%    Foundation}{http://dx.doi.org/10.13039/100000001} under Grant
%  No.~\grantnum{GS100000001}{nnnnnnn} and Grant
%  No.~\grantnum{GS100000001}{mmmmmmm}.  Any opinions, findings, and
%  conclusions or recommendations expressed in this material are those
%  of the author and do not necessarily reflect the views of the
%  National Science Foundation.
%\end{acks}

%% Bibliography
%\bibliography{bibfile}

%% Appendix

\clearpage
\bibliography{dblp,misc}

\clearpage
\appendix
%!TEX root = draft.tex
\section{Proofs of Section~\ref{sec:def}}\label{app:definitions}

\begin{lemma}
 Let $\hist=\tup{T, \so, \wro}$ be a history. 
 If $\tup{\hist,\co}$ satisfies $\mathsf{Read\ Atomic}$, then %the extension of $\wro[\xvar]$ to transactions 
 for every transaction $\tr$ and two reads $\rd[\id_1]{\xvar}{\val_1},\rd[\id_2]{\xvar}{\val_2}\in \readOp{\tr}$, $\wro^{-1}(\rd[\id_1]{\xvar}{\val_1})=\wro^{-1}(\rd[\id_2]{\xvar}{\val_2})$ and $\val_1 = \val_2$.\end{lemma}
\begin{proof}
 Let $\tup{\tr_1, \rd[\id_1]{\xvar}{\val_1}}, \tup{\tr_2, \rd[\id_2]{\xvar}{\val_2}} \in \wro[\xvar]$. Then $\tr_1, \tr_2$ write to $\xvar$. Let us assume by contradiction, that $\tr_1\neq\tr_2$. By \textsf{Read Atomic}, $\tup{\tr_2, \tr_1} \in \co$ because $\tup{\tr_1, \rd[\id_1]{\xvar}{\val_1}} \in \wro[\xvar]$ and $\tr_2$ writes to $\xvar$. Similarly, we can also show that $\tup{\tr_1, \tr_2} \in \co$. This contradicts the fact that $\co$ is a strict total order. Therefore, $\tr_1 = \tr_2$. We also have that $\val_1 = \val_2$ because each transaction contains a single write to $\xvar$.
 \end{proof}

\begin{lemma}
 The following entailments hold:
 \begin{align*}
   & \mathsf{Causal} \implies \mathsf{Read\ Atomic}\implies \mathsf{Read\ Committed} \\
   & \mathsf{Prefix} \implies \mathsf{Causal}                                        \\
   & \mathsf{Serializability} \implies \mathsf{Prefix}\land \mathsf{Conflict}        
 \end{align*} 
\end{lemma}
\begin{proof}
  We will show the contrapositive of each implication:
 \begin{itemize}
   \item If $\tup{\hist, \co}$ does not satisfy \textsf{Read Committed},  then
\begin{align*}
\exists \xvar,\ \exists \tr_1,\tr_2,\ \exists \alpha,\beta.\ \tup{\tr_1,\alpha}\in \wro[\xvar] \land \writeVar{\tr_2}{\xvar}\ \land \tup{\tr_2,\beta}\in\wro \land \tup{\beta,\alpha} \in \po \land \tup{\tr_1,\tr_2}\in\co. 
\end{align*}
Let $\tr_3$ the transaction containing $\alpha$ and $\beta$. We have that $\tup{\tr_2, \tr_3} \in \wro$. But then we have $\tr_1, \tr_2, \tr_3$ such that $\tup{\tr_1, \tr_3} \in \wro[\xvar]$ and $\tup{\tr_2, \tr_3} \in \wro$ and $\writeVar{\tr_2}{\xvar}$. So by \textsf{Read Atomic}, $\tup{\tr_2, \tr_1} \in \co$. This contradicts the fact that $\co$ is a strict total order. Therefore, $\tup{\hist, \co}$ does not satisfy \textsf{Read\ Atomic}.
   \item If $\tup{\hist, \co}$ does not satisfy \textsf{Read Atomic}, then 
 \begin{align*}  
 \exists \xvar, \exists \tr_1,\tr_2,\tr_3.\ \tup{\tr_1,\tr_3}\in \wro[\xvar] \land \writeVar{\tr_2}{\xvar}\ \land \tup{\tr_2,\tr_3}\in\wro\cup\so \land \tup{\tr_1,\tr_2}\in\co.
 \end{align*}
  Then $\tup{\tr_2,\tr_3}\in(\wro\cup\so)^+$. Then, by \textsf{Causal}, we have $\tup{\tr_2, \tr_1} \in \co$, which contradicts the fact that $\co$ is a strict total order. Therefore, $\tup{\hist, \co}$ does not satisfy \textsf{Causal}.
   \item If $\tup{\hist, \co}$ does not satisfy \textsf{Causal}, then
\begin{align*}
\exists \xvar, \exists \tr_1,\tr_2,\tr_3.\ \tup{\tr_1,\tr_3}\in \wro[\xvar] \land \writeVar{\tr_2}{\xvar}\ \land \tup{\tr_2,\tr_3}\in(\wro\cup\so)^+ \land\tup{\tr_1,\tr_2}\in\co.
\end{align*}
   But, $(\wro\cup\so)^+ = (\wro\cup\so)^* \circ (\wro\cup\so) \subseteq \co^* \circ (\wro\cup\so)$. Therefore, $\tup{\tr_2,\tr_3}\in\co^* \circ (\wro\cup\so)$. Then, by \textsf{Prefix}, we have $\tup{\tr_2, \tr_1} \in \co$, which contradicts the fact that $\co$ is a strict total order. Therefore, $\tup{\hist, \co}$ does not satisfy \textsf{Prefix}.
   \item If $\tup{\hist, \co}$ does not satisfy \textsf{Prefix} or \textsf{Conflict}, then 
\begin{align*}
\exists \xvar, \exists \tr_1,\tr_2,\tr_3, \tr_4.\ \tup{\tr_1,\tr_3}\in \wro[\xvar] \land \writeVar{\tr_2}{\xvar}\ \land \tup{\tr_2,\tr_4}\in\co^* \land \tup{\tr_1,\tr_2}\in\co
\end{align*}
and 
       \begin{itemize}
         \item $\tup{\tr_4,\tr_3}\in \co \land \writeVar{\tr_3}{\yvar}\ \land \writeVar{\tr_3}{\yvar}$ if it violates \textsf{Conflict}.
         \item $\tup{\tr_4,\tr_3}\in (\wro \cup \so)$ if it violates \textsf{Prefix}.
       \end{itemize}
       
In both cases, we have that $\tup{\tr_4, \tr_3} \in \co$. Because $\co$ is transitive, $\tup{\tr_2, \tr_4} \in \co^*$ and $\tup{\tr_4, \tr_3} \in \co$ imply that $\tup{\tr_2, \tr_3} \in \co$. Then by \textsf{Serializability}, we have $\tup{\tr_2, \tr_1} \in \co$, which contradicts the fact that $\co$ is a strict total order. Therefore, $\tup{\hist, \co}$ does not satisfy \textsf{Serializability}.
 \end{itemize}
\end{proof}

%!TEX root = draft.tex
\section{Proofs of Section~\ref{sec:general}}\label{app:sec:general}

\begin{algorithm}
 \SetKwInOut{KwInput}{Input}
 \SetKwInOut{KwOutput}{Output}
 \KwIn{A history $\hist = \tup{T, \so, \wro}$}
 \KwOut{$\mathit{true}$ iff $\hist$ satisfies \textsc{Causal consistency}}
 \BlankLine
 \If{$\so\cup\wro$ is cyclic} {
  \Return{false}\;
 }
 $\co \leftarrow \so\cup\wro$\;
 \ForEach{$\xvar \in \vars{\hist}$}{
  \ForEach{$\tr_1 \neq \tr_2 \in T$ s.t. $\tr_1$ and $\tr_2$ write on $\xvar$}{
   \If{$\exists \alpha, \beta.\ \tup{\tr_1,\alpha}\in \wro[\xvar]\land \tup{\tr_2,\beta}\in (\so\cup\wro) \land \tup{\alpha, \beta} \in \po$} { %\Path{\tr_2}{E_1^+}{\tr_3}, \Path{\tr_1}{\wro[\xvar]}{\tr_3}
    $\co \leftarrow \co \cup \{\tup{\tr_2, \tr_1}\}$\;
   }
  }
 }
 \eIf{$\co$ is cyclic}{
  \Return{false}\;
 }{
  \Return{true}\;
 }
 \caption{Checking \textsc{Read Committed}}
 \label{rcalgo:1}
\end{algorithm}

\begin{algorithm}
 \SetKwInOut{KwInput}{Input}
 \SetKwInOut{KwOutput}{Output}
 \KwIn{A history $\hist = \tup{T, \so, \wro}$}
 \KwOut{$\mathit{true}$ iff $\hist$ satisfies \textsc{Causal consistency}}
 \BlankLine
 \If{$\so\cup\wro$ is cyclic} {
  \Return{false}\;
 }
 $\co \leftarrow \so\cup\wro$\;
 \ForEach{$\xvar \in \vars{\hist}$}{
  \ForEach{$\tr_1 \neq \tr_2 \in T$ s.t. $\tr_1$ and $\tr_2$ write on $\xvar$}{
   \If{$\exists \tr_3.\ \tup{\tr_1,\tr_3}\in \wro[\xvar]\land \tup{\tr_2,\tr_3}\in (\so\cup\wro)$} { %\Path{\tr_2}{E_1^+}{\tr_3}, \Path{\tr_1}{\wro[\xvar]}{\tr_3}
    $\co \leftarrow \co \cup \{\tup{\tr_2, \tr_1}\}$\;
   }
  }
 }
 \eIf{$\co$ is cyclic}{
  \Return{false}\;
 }{
  \Return{true}\;
 }
 \caption{Checking \textsc{Read Atomic}}
 \label{raalgo:1}
\end{algorithm}

\begin{theorem}
The problem of checking whether a history satisfies \emph{\textsc{Read Committed}}, \emph{\textsc{Read Atomic}}, or \emph{\textsc{Causal consistency}} is polynomial time.
\end{theorem}
\begin{proof}
 % ========
 % TODO TAKE THIS INTO ACCOUNT:
 % Also, the writing in the proof of Theorem 3.1. is kind of “naive”. For instance, the first sentence is a description of the algorithm. This has no place in the proof. Such a thing has to be mentioned in the description of the algorithm. Which should be added at some point (together with some example). It should be decomposed in two parts:
 % 
 % 1) if the history h satisfies CC, then the algorithms says yes. In this case, you have to show that the algorithm will not find a cycle in co. This is a consequence of h satisfying the axiom.
 % 
 % 2) if the algorithm says yes, then the history satisfies CC. Here you have to show that the topological order of the co at the end of the algorithm will satisfy the CC axiom (together with the history h). And here, one should say that the new ordering constraints added by doing the topological order will not imply more instantiations of the left part of the entailment in the axiom.
 % =======
 
We first consider the case of \textsc{Read Committed}. Algorithm \ref{rcalgo:1} finds all the $\co$ relations that are implied by the $\mathsf{Read Committed}$ axiom (fig. \ref{lock_rc_def}) \ie for all $\tr_1, \tr_2$ and for all $\alpha, \beta$ if we have, $\tup{\tr_1, \alpha} \in \wro[\xvar]$ and $\tup{\tr_2, \beta} \in \wro$ and $\tup{\beta, \alpha} \in \po$ and $\writeVar{\tr_2}{\xvar}$(from figure \ref{lock_rc_def}), then we add $\tup{\tr_2, \tr_1} \in \co$. Quantification can be done in cubic iteration over the transactions for each variable. Now, we claim $\hist$ is \textsc{Read Committed} if and only if $\co$, the union of these found relations is acyclic. 
 
  First we prove that if $\co$ is acyclic, then $\hist$ satisfies \textsc{Read Committed}. $\co$ is acyclic, hence consider a topological order $\co'$ of $\co$. If $\tup{\hist, \co'}$ does not satisfies \textsc{Read Committed}, there exists $\tr_1, \tr_2$ and $\alpha, \beta$ such that, $\tup{\tr_1, \alpha} \in \wro[\xvar]$ and $\tup{\tr_2, \beta} \in \wro$ and $\tup{\alpha, \beta} \in \po$ and $\writeVar{\tr_2}{\xvar}$ and $\tup{\tr_1, \tr_2} \in \co'$. But with the same $\tr_1, \tr_2, \alpha, \beta$ and variable $\xvar$, we must have added $\tup{\tr_2, \tr_1}$ in $\co$ which implies any topological order of $\co$ can not have $\tup{\tr_1, \tr_2}$ which contradicts that $\co'$ contains $\tup{\tr_1, \tr_2}$.
 
 Now we prove, if $\co$ is not acyclic, then $\hist$ does not satisfie \textsc{Read Committed}. If the history is \textsf{Read Committed}, there must be a commit order $\co'$ for $\hist$, for which $\tup{\hist, \co'}$ satisfies \textsf{Read Committed}. $\co'$ must be acyclic. Now, take any cycle in $\co$, $\tr_1 \xrightarrow{\co} \tr_2 \xrightarrow{\co} \cdots \xrightarrow{\co} \tr_k \xrightarrow{\co} \tr_1$. Now along these, if for all $i, j$ $\tr_i \xrightarrow{\co'} \tr_j$ then, it becomes a cycle in $\co'$. Therefore, there must be atleast one pair, where $\tup{\tr_i, \tr_j} \in \co$, but $\tup{\tr_i, \tr_j} \not\in \co'$ or $\tup{\tr_j, \tr_i} \in \co'$ since $\co'$ is total.
 
 But $\tup{\tr_i, \tr_j} \in \co$ must have been added because there exists $\alpha, \beta$ such that $\tup{\tr_i, \beta} \in \wro$, $\tup{\beta, \alpha} \in \po$ and $\tup{\tr_j, \alpha} \in \wro[\xvar]$ and $\writeVar{\tr_i}{\xvar}$. But $\tup{\tr_j, \tr_i} \in \co'$, so $\co'$ violates \textsc{Read Commmitted} for $\tr_i, \tr_j, \alpha, \beta$. Therefore, if $\co$ has a cyclic, $\hist$ does not satisfy \textsc{Read Committed}.
 
The case of \textsc{Read Atomic} and \textsc{Causal Consistency} is similar. If the $\co$ at the end of the algorithm is acyclic, then if a topological order extending it does not satisfy resp. consistency model, then it contains $\tup{\tr_1, \tr_2}$ with the quantified transactions, variables and operations for resp. consistency model. But, then for same quantified transactions, variables and operations, we must have added $\tup{\tr_2, \tr_1} \in \co$ in the algorithm. It contradicts the topological order contains $\tup{\tr_1, \tr_2}$. For the other direction, if the $\co$ at the end of the algorithm has cycle, we show there exists $\tup{\tr_i, \tr_j} \in \co$, like the case of \textsf{Read Committed}, yet it is not in a commit order satisfying the resp. consistency models. Then the commit order must have $\tup{\tr_j, \tr_i}$ which violates the resp. consistency models for the exact same quantified transactions, variables and operations for which $\tup{\tr_i, \tr_j}$ was added in $\co$ at first.
 % \begin{itemize}
 %  \item If it returns a $\CO$, it does not violate Causal axiom by construction (line 8).
 %  \item If it returns $\none$, then there was a cycle in $\CO$ relations which are implies by $\wro$, $\so$ and Causal axiom. 
 % \end{itemize}
 % 
 % The algorithm run in \textsf{PTIME}. So verifying a history for causally consistent is in PTIME.
\end{proof}

\section{Proofs of Section~\ref{ssec:pc}}\label{app:pc_red}

\begin{lemma}\label{lem:pc_width:app}
The histories $\hist$ and $\hist_{R|W}$ have the same width.
\end{lemma}
\begin{proof}
%Constructing $T', \wro[\xvar]', \so'$ can be done doing constant iterations on the corresponding objects in $\hist$.
We show that if $\hist$ is of width $k$, then the session order $\so'$ of $\hist_{R|W}$ cannot contain an antichain of size $k+1$.
Let $\{X^1_{\tr_1}, X^2_{\tr_2}, \ldots X^k_{\tr_k}, X^{k+1}_{\tr_{k+1}}\}$ with $X^i\in \{R,W\}$, for all $1\leq i\leq k+1$, be a set of $k+1$ transactions in $\hist_{R|W}$. Then,
%Take any $(k+1)$ sized set of transactions in $T'$, $\tr'$, $\{X_{\tr_1}, X_{\tr_2}, \ldots X_{\tr_k}, X_{\tr_{(k+1)}}\}$. 
\begin{itemize}
 \item if $\tr_i = \tr_j=\tr$ for some $i \neq j$, then $X^i_{\tr_i}=R_{\tr}$ and $X^j_{\tr_j}=W_\tr$ or vice-versa. Since $\tup{R_{\tr}, W_{\tr}} \in \so'$, this set cannot be an antichain of $\so'$.
 % that means, they are $W_{\tr}, R_{\tr} \in T'$ for some $\tr \in T$ and $\tup{R_{\tr}, W_{\tr}} \in \so'$. Therefore the set is not an antichain.
 \item otherwise, by hypothesis, the set $\{\tr_1, \tr_2, \ldots, \tr_k, \tr_{k+1}\}$ is not an antichain of $\so$. Thus, there exists $i, j$ such that $\tup{\tr_i, \tr_j} \in \so$. By the definition of $\so'$, $\tup{X^i_{\tr_i}, X^j_{\tr_j}}\in\so'$, which implies that this set is not an antichain of $\so'$.
\end{itemize}
\end{proof}

\section{Proofs of Section~\ref{ssec:si}}\label{app:si_red}

 \begin{figure}[t]
  \centering
  \begin{subfigure}{.27\textwidth}
   \resizebox{\textwidth}{!}{
    \begin{tikzpicture}[->,>=stealth',shorten >=1pt,auto,node distance=4cm,
      semithick, transform shape]
     \node[transaction state] at (0,0)       (t_1)           {$W_{\tr_1}$};
     \node[transaction state, label={below:{$\writeVar{ }{\xvar_{3,4}}$}}] at (2,0)       (t_3)           {$R_{\tr_3}$};
     \node[transaction state, label={below:{$\writeVar{ }{\yvar}$}}] at (4,0)       (t_3_w)           {$W_{\tr_3}$};
     \node[transaction state,label={above:$\writeVar{ }{\xvar}$}] at (-0.5,1.5) (t_2) {$W_{\tr_2}$};
     \node[transaction state, label={above:{$\writeVar{ }{\yvar,\xvar_{3,4}}$}}] at (1.5,1.5) (t_4) {$W_{\tr_4}$};
     \path (t_1) edge node {$\wro[\xvar]$} (t_3);
     % \path (t_2) edge[blue] node {$\CO$} (t_1);
     \path (t_2) edge[red] node {$\co'^*_1$} (t_4);
     \path (t_4) edge[right, red] node[pos=.9] {$\wrosi(\co'_1)$} (t_3);
     \path (t_4) edge node {$\co'_1$} (t_3_w);
     \path (t_3) edge[below] node {$\wro[\xvar_{3,4}]'$} (t_3_w);
     \path (t_1) edge[left] node {$\co'_1$} (t_2);
     \path (t_3) edge[left, red] node[pos=0.1, rotate=-30, yshift=2.5mm] {$\rwo(\co'_1)$} (t_2);
    \end{tikzpicture}
   }
   \caption{Minimal cycle in $\tup{\hist', \co'_2}$}
   \label{si_p_proof:2a}
  \end{subfigure}
  \begin{subfigure}{0.20\textwidth}
   \resizebox{\textwidth}{!}{
    \begin{tikzpicture}[->,>=stealth',shorten >=1pt,auto,node distance=4cm,
      semithick, transform shape]
     \node[transaction state, text=red] at (0,0)       (t_1)           {$\tr_1$};
     \node[transaction state, label={below:{$\writeVar{ }{\yvar}$}}] at (2,0)       (t_3)           {$\tr_3$};
     \node[transaction state, text=red,label={above:\textcolor{red}{$\writeVar{ }{\xvar}$}}] at (-0.5,1.5) (t_2) {$\tr_2$};
     \node[transaction state, label={above:{$\writeVar{ }{\yvar}$}}] at (1.5,1.5) (t_4) {$\tr_4$};
     \path (t_1) edge[red] node {$\wro[\xvar]$} (t_3);
     % \path (t_2) edge[blue] node {$\CO$} (t_1);
     \path (t_2) edge node {$\co^*$} (t_4);
     \path (t_4) edge node {$\co$} (t_3);
     \path (t_1) edge[left] node {$\co$} (t_2);
    \end{tikzpicture}
   }
   \caption{$\mathsf{Conflict}$ violation in $\tup{\hist, \co}$}
   \label{si_p_proof:2b}
  \end{subfigure}
  \caption{Cycles in $\co_2'$ corresponding to $\mathsf{Conflict}$ violations.}
  \label{si_p_proof:2}
 \end{figure}
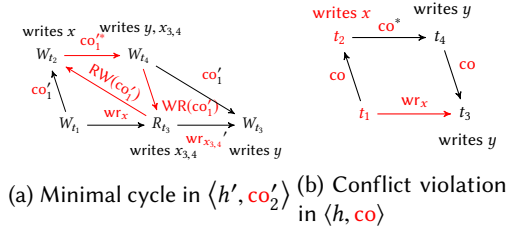

\begin{theorem}\label{th:si:app}
A history $\hist$ satisfies snapshot isolation iff $\hist_{R|W}^c$ is serializable.
\end{theorem}
\begin{proof}
For the ``only-if'' direction, we define partial commit orders $\co_1'$ and $\rwo(\co_1')$ as in the case of prefix consistency. Along with them, we define a partial commit order $\wrosi(\co_1')$ 
\begin{align*}
 \wrosi(\co'_1) = \{\tup{W_{\tr_1}, R_{\tr_2}}| & \exists \xvar_{2,1} \in \vars{\hist_{R|W}^c}.\                             \\
                                & \tup{R_{\tr_2},W_{\tr_2}} \in \wro[\xvar_{2,1}]', \\
                                & \tup{W_{\tr_1}, W_{\tr_2}} \in \co'_1, \writeVar{W_{\tr_1}}{\xvar_{2,1}} \} 
\end{align*}
which intuitively, enforces that the read part $R_{\tr_2}$ of a transaction $\tr_2$ observes the effects of the write part $W_{\tr_1}$ of a transaction $\tr_1$ when $\tr_1$ and $\tr_2$ write on a common variable and the commit order in $\hist$ orders $\tr_1$ before $\tr_2$ (which implies that the corresponding write transactions are ordered in the same way in $\co'_1$). We define $\co_2' = \co_1' \cup \rwo(\co_1') \cup \wrosi(\co'_1)$. 

The characterization of minimal cycles of $\co_1'$ and ultimately, the fact that it is acyclic can be proved as in Lemma~\ref{lem:co1}. The proof that $\co_2'$ is acyclic goes as follows. As for PC, since $\co'_1$ is acyclic, a cycle in $\co'_2$, and in particular a minimal one, must  necessarily contain a dependency from $\rwo(\co'_1)$ or $\wrosi(\co'_1)$. Note that a minimal cycle cannot contain two dependencies in either $\rwo(\co'_1)$ or $\wrosi(\co'_1)$ since this would imply that it contains two non-consecutive write transactions. Differently from the previous case, the cycle in $\co_2'$ here can also contains the dependencies in $\wrosi(\co'_1)$ which are from write transactions to read transactions. % (by the definition of $\hist_{R|W}^c$, a write transaction may read a variable and a read transaction may write the same variable). 
The case of minimal cycles in $\co'_2$ that contain only a dependency from $\rwo(\co'_1)$, and no dependencies from $\wrosi(\co'_1)$, can be dealt with as in the case of PC.

Consider a minimal cycle of $\co_2'$ that contains a dependency $\tup{W_{\tr_4}, R_{\tr_3}}$ in $\wrosi(\co'_1)$, which implies that $W_{\tr_4}, W_{\tr_3}$ must write on some common variable $y$. Because the minimal cycle contains at most two write transactions and one read transaction, it must also contain a dependency from read transactions to write transactions. Note that such a dependency can come only from $\rwo(\co'_1)$.
% So a minimal cycle in $\co'_2$ contains one relation from $\wrosi(\co'_1)$ of the form $\tup{W_{\tr_4}, R_{\tr_3}}$, where $W_{\tr_4}, W_{\tr_3}$ both write on same variable. Because the minimal cycle also contains at most two write transactions and one read transaction, it must contain a dependency from a read transaction to a write transaction as well - which must come from $\rwo(\co_1')$ also.
The red edges in Figure~\ref{si_p_proof:2a} show such a cycle. By the definition of $\hist_{R|W}^c$, we have that $W_{\tr_4}$ and $R_{\tr_3}$ write on a variable $\xvar_{3,4}$ and $\tup{R_{\tr_3}, W_{\tr_3}} \in \wro[\xvar_{3,4}]'$. Since $\tup{R_{\tr_3}, W_{\tr_2}}\in \rwo(\co'_1)$, we have that there exists a write transaction $W_{\tr_1}$ s.t. $\tup{W_{\tr_1}, R_{\tr_3}}\in \wro'[\xvar]$, for some $\xvar$, and $\tup{W_{\tr_1}, W_{\tr_2}}\in \co'_1$.
The relations between these transactions of $\hist_{R|W}^c$ imply that the corresponding transactions of $\hist$ are related as shown in Figure~\ref{si_p_proof:2b}, which implies a violation of $\mathsf{Conflict}$, a contradiction of the hypothesis.

% So if we have a minimal cycle because of $\tup{R_{\tr_3}, W_{\tr_2}} \in \rwo(\co_1'), \tup{W_{\tr_4}, R_{\tr_3}} \in \wrosi(\co_1')$ and $\tup{W_{\tr_2}, W_{\tr_4}} \in \co'^*_1$ (red cycle in figure \ref{si_p_proof:2a}), we have $W_{\tr_4}, R_{\tr_3}$ write on $\xvar_{3,4}$ and $\tup{R_{\tr_3}, W_{\tr_3}} \in \wro[\xvar_{3,4}]'$. 
% The relations between these transactions of $\hist_{R|W}^c$ imply that the corresponding transactions of $\hist$ are related as shown in Figure~\ref{si_p_proof:2b}, which represents a violation of $\axconf$. 
%  Then if we look at the corresponding transactions $\tr_1, \tr_2, \tr_3, \tr_4$ in $\hist$. 
% The relations in figure \ref{si_p_proof:2a} will imply the relations in \ref{si_p_proof:2b}. 
% But to satisfy $\axconf$, $\tup{\tr_1, \tr_2}, \tup{\tr_4, \tr_3} \in \co$ and $\tup{\tr_2, \tr_4} \in \co^*$ and $\tup{\tr_1, tr_3} \in \wro[\xvar]'$ and $\tr_2$ writes on $\xvar$ and $\tr_3, \tr_4$ write on some same variable $\yvar$, implies $\tup{\tr_2, \tr_1} \in \co$ - which clearly contradicts figure \ref{si_p_proof:2b}.  
 
% So $\co'_2$ is acyclic. We take $\co'$ to be any topological order of $\co'_2$. Similarly as prefix consistency case, we can show it must satisfy $\axser$ and it proves our claim in predicate (\ref{si_leftright}).

For the ``if'' direction, let $\co'$ be a commit (total) order on transactions of $\hist_{R|W}^c$ which satisfies the serializability axiom. Let $\co$ be a commit order on transactions of $\hist$ defined by $\co = \{\tup{\tr_1, \tr_2} | \tup{W_{\tr_1}, W_{\tr_2}} \in \co'\}$ ($\co$ is clearly a total order). Showing that $\co$ is an extension of $\wro \cup \so$ and that it doesn't expose a $\axpre$ violation can be done as for prefix consistency. Now, assume by contradiction that there exists a $\axconf$ violation between $\tr_1$, $\tr_2$, $\tr_3$, $\tr_4$ (shown in Figure~\ref{si_p_proof:3a}). Then, the corresponding transactions $W_{\tr_1}, W_{\tr_2}, W_{\tr_4}, R_{\tr_3}, W_{\tr_3}$ in $\hist_{R|W}^c$, shown in Figure~\ref{si_p_proof:3b}, would be related as follows: (1) since $\tup{\tr_1, \tr_3} \in \wro[\xvar]$ and $\tup{\tr_1, \tr_2} \in \co$, we have that $\tup{W_{\tr_1}, R_{\tr_3}} \in \wro[\xvar]'$ and $\tup{W_{\tr_1}, W_{\tr_2}} \in \co'$, (2) since $\co'$ satisfies $\axser$, then $\tup{R_{\tr_3}, W_{\tr_2}} \in \co'$, (3) $\tup{\tr_2, \tr_4} \in \co^*$ implies $\tup{W_{\tr_2}, W_{\tr_4}} \in \co'^*$, (4) $\tup{\tr_4, \tr_3} \in \co$ and $\tr_4, \tr_3$ write on a common variable $y$ implies that $\tup{W_{\tr_4}, W_{\tr_3}} \in \co'$, $\tup{R_{\tr_3}, W_{\tr_3}} \in \wro'[\xvar_{3,4}]$, and $W_{\tr_4}$ writes the variable $\xvar_{3,4}$, which by the serializability axiom, implies $\tup{W_{\tr_4}, R_{\tr_3}} \in \co'$. Therefore, $\co'$ contains a cycle, a contradiction to the hypthesis.
\end{proof}

 \begin{figure}[t]
  \centering
  \begin{subfigure}{.20\textwidth}
   \resizebox{\textwidth}{!}{
    \begin{tikzpicture}[->,>=stealth',shorten >=1pt,auto,node distance=4cm,
      semithick, transform shape]
     \node[transaction state, text=red] at (0,0)       (t_1)           {$\tr_1$};
     \node[transaction state, label={below:{$\writeVar{ }{\yvar}$}}] at (2,0)       (t_3)           {$\tr_3$};
     \node[transaction state, text=red,label={above:\textcolor{red}{$\writeVar{ }{\xvar}$}}] at (-0.5,1.5) (t_2) {$\tr_2$};
     \node[transaction state, label={above:{$\writeVar{ }{\yvar}$}}] at (1.5,1.5) (t_4) {$\tr_4$};
     \path (t_1) edge[red] node {$\wro[\xvar]$} (t_3);
     % \path (t_2) edge[blue] node {$\CO$} (t_1);
     \path (t_2) edge node {$\co^*$} (t_4);
     \path (t_4) edge node {$\co$} (t_3);
     \path (t_1) edge[left] node {$\co$} (t_2);
    \end{tikzpicture}
   }
   \caption{Prefix violation in $\tup{\hist, \co}$}
   \label{si_p_proof:3a}
  \end{subfigure}
  \begin{subfigure}{.20\textwidth}
   \resizebox{\textwidth}{!}{
    \begin{tikzpicture}[->,>=stealth',shorten >=1pt,auto,node distance=4cm,
      semithick, transform shape]
     \node[transaction state] at (0,0)       (t_1)           {$W_{\tr_1}$};
     \node[transaction state] at (2,0)       (t_3)           {$R_{\tr_3}$};
     \node[transaction state, label={below:{$\writeVar{ }{\yvar}$}}] at (3,0)       (t_3_w)           {$W_{\tr_3}$};
     \node[transaction state,label={above:$\writeVar{ }{\xvar}$}] at (-0.5,1.5) (t_2) {$W_{\tr_2}$};
     \node[transaction state, label={above:{$\writeVar{ }{\yvar}$}}] at (1.5,1.5) (t_4) {$W_{\tr_4}$};
     \path (t_1) edge node[below] {$\wro[\xvar]$} (t_3);
     % \path (t_2) edge[blue] node {$\CO$} (t_1);
     \path (t_2) edge[red] node {$\co'^*$} (t_4);
     \path (t_4) edge[red, left] node {$\co'$} (t_3);
     \path (t_4) edge node {$\co'$} (t_3_w);
     \path (t_3) edge[below] node {$\co'$} (t_3_w);
     \path (t_1) edge[left] node {$\co'$} (t_2);
     \path (t_3) edge[red, below left] node[pos=0.4] {$\co'$} (t_2);
    \end{tikzpicture}
   }
   \caption{Cycle in $\tup{\hist', \co'_2}$}
   \label{si_p_proof:3b}
  \end{subfigure}
  
  \caption{$\forall \co'. \exists \co.\tup{\hist', \co'} \models \axser \Rightarrow \tup{\hist, \co} \models \axpre \land \axconf$}
  \label{si_p_proof:3}
 \end{figure}

%!TEX root = draft.tex

\section{Equivalence between our definitions and the formalization in~\cite{DBLP:conf/concur/Cerone0G15}}\label{app:gotsman}

%Cerone et al. worked on a similar consistency models~\cite{DBLP:conf/concur/Cerone0G15}. They use similar definitions for operation, transaction, history like ours, except, unlike our definition, their definition are data dependent. So, they do not have $\wro$ relations in a history \ie $\Hist_{\so} = \tup{\Tr, \so}$. Then, they define the consistency axioms in terms of visibility $\vis$ and $\co$ of a history $\hist$. A history $\hist$ satisfies a consistency axiom if there exists $\vis$ and $\co$ for that history, $\tup{\hist, \vis, \co}$ satisfies the resp. consistency axioms.

\citet{DBLP:conf/concur/Cerone0G15} define the criteria RA, CC, PC, SI, and SER using a notion of history that contains only the session order $\so$. Such a history satisfies one of these criteria in their formalization if there exists a visibility relation $\VIS$ between transactions, and a commit order $\co$ extending the visibility relation that satisfy certain axioms.

\begin{table}
 \centering
 \resizebox{.48\textwidth}{!}{
  \begin{tabular}{ |c c|c c|}
   \hline
   \multicolumn{3}{|c}{
    \shortstack{
     $\forall (O, \textsf{po}) \in \tr, \forall o \in O, o = read(x, n) \land$                    \\
     $\left\{o' \in \textsf{po}^{-1}(o) \mid o' = \_(x, \_) \right\} \neq \emptyset \Rightarrow$          \\
     $max_{\textsf{po}} \left(\left\{o' \in \textsf{po}^{-1}(o) \mid o' = \_(x, \_) \right\}\right) = \_(x, n)$
    }
   }                         & {\textsc{(Int)}}                                                           \\
   \hline
   \multicolumn{3}{|c}{
    \shortstack{
     $\forall \tr = (O, \textsf{po}) \in \tr, \forall x, \tr \models \texttt{read}(x, n) \Rightarrow$ \\ $\left(\left(\VIS^{-1}(\tr) \cap \textsf{Write}_x\right) = \emptyset \land n = 0 \right) \lor$ \\
     $\left(max_{\CO} \left(\VIS^{-1}(\tr) \cap \textsf{Write}_x\right) \models \texttt{write}(x, n)\right)$
    }
   }                         & {\textsc{(Ext)}}                                                           \\
   \hline
   $\SO \subseteq \VIS$      & \textsc{(Session)}      & \VIS is transitive & \textsc{(TransVis)}         \\
   \hline
   $\CO\circ \VIS \subseteq \VIS$ & \textsc{(Prefix)}       & $\VIS = \CO$       & \textsc{(TotalVis)}         \\
   \hline
   \multicolumn{3}{|c}{
    \shortstack{
     $\forall \tr, \tr' \in T, \forall x, (\tr, \tr' \in \textsf{Write}_x \land \tr \neq \tr') $               \\
     $\Rightarrow (\tr \xrightarrow{\VIS} \tr' \lor \tr' \xrightarrow{\VIS} \tr)$
   }  }                      & {\textsc{(NoConflict)}}                                                    \\
   \hline
  \end{tabular}
 }
 \caption{Consistency axioms for a history $\hist_{\so}=\tup{T,\so}$, visibility relation \textcolor{red}{\sf{vis}}, and commit order \textcolor{red}{\sf{co}}.}
 \label{weakconsistency_gotsman:1}
\end{table}

The axioms used by \citet{DBLP:conf/concur/Cerone0G15} are given in Table~\ref{weakconsistency_gotsman:1}. 

\textsc{Int} is an axiom which enforces that if there is a read operation $O$ on variable $\xvar$ in a transaction and there is a read or write operation on $\xvar$ before $O$ \ie $\left\{o' \in \textsf{po}^{-1}(o) \mid o' = \_(x, \_) \right\} \neq \emptyset$, then the latest operation on $\xvar$ before $O$ must read or write the value read by $O$ \ie $max_{\textsf{po}} \left(\left\{o' \in \textsf{po}^{-1}(o) \mid o' = \_(x, \_) \right\}\right) = \_(x, n)$.

\textsc{Ext} is an axiom which enforces that if a transaction $\tr$ has a operation $O$ which reads a variable $\xvar$ and which is not preceded by a write on $\xvar$, denoted by $\tr \models \texttt{read}(x, n)$, then either:
\begin{itemize}
	\item it read the initial value 0, and there is no transaction writing on $\xvar$ visible to $\tr$, i.e., $\VIS^{-1}(\tr) \cap \textsf{Write}_x=\emptyset$, or
	\item it read from a write of another transaction $\tr'$ which writes to variable $\xvar$ and $\tr'$ is the last one in the commit order in the visibility set of $\tr$, i.e., $\VIS^{-1}(\tr)$, that writes on $\xvar$. 
\end{itemize}
	Since, the writes in our history have unique values, this is equivalent to, for all $\tr_1, \tr_3$ if $\tr_1 \models write(\xvar, n)$ and $\tr_3 \models read(\xvar, n)$, then for any $\tr_2 \in \vis^{-1}(\tr_3)$ where $\tr_2 \neq \tr_1$ and $\tr_2 \models write(\xvar, \_)$, $\tr_2$ can not be after $\tr_1$ in $\co$ order (\ie $\tup{\tr_2, \tr_1} \in \co$ since $\co$ is total), because $\tr_1$ must be the maximal among the transactions that wrote $\xvar$. We illustrated the axiom in Figure~\ref{ext_fig}, which is very similar to our definition in Figure~\ref{consistency_defs}.

\begin{figure}
 \begin{subfigure}[t]{.3\textwidth}
  \centering
  \begin{tikzpicture}[->,>=stealth',shorten >=1pt,auto,node distance=4cm,
    semithick, transform shape]
   \node[transaction state, text=red] at (0,0)       (t_1)           {$\tr_1$};
   \node[transaction state] at (2,0)       (t_3)           {$\tr_3$};
   \node[transaction state, text=red,label={above:\textcolor{red}{$\writeVar{ }{\xvar}$}}] at (-.5,1.5) (t_2) {$\tr_2$};
   \path (t_1) edge[red] node {$\wro[\xvar]$} (t_3);
   % \path (t_2) edge[blue] node {$\CO$} (t_1);
   \path (t_2) edge[bend left] node {$\vis$} (t_3);
   \path (t_2) edge[left,double] node {$\co$} (t_1);
  \end{tikzpicture}
  \parbox{\textwidth}{
   $\forall \xvar,\ \forall \tr_1,\tr_2,\ \forall \tr_3.$
   
   \hspace{4mm}$\tup{\tr_1,\tr_3}\in \wro[\xvar] \land \writeVar{\tr_2}{\xvar}\ \land$ 
   
   \hspace{9mm}$\tup{\tr_2,\tr_3}\in\vis$
   
   \hspace{14mm}$\implies \tup{\tr_2,\tr_1}\in\co$
  }
  
  \caption{$\mathsf{Int} \land \mathsf{Ext}$}
  \label{ext_fig}
 \end{subfigure}
 \caption{}
 \label{}
\end{figure}

The definitions for \textsc{Session}, \textsc{TransVis}, \textsc{Prefix}, \textsc{TotalVis} are straightforward. \textsc{NoConflict} enforces $\vis$ to totally order the transactions those write on same variable.

We will show in our axioms definition figures, the path between $\tr_2$ and $\tr_3$ is essentially a $\vis$ relation in $\hist_{\so}$.

\begin{table}
 \centering
 \resizebox{.44\textwidth}{!}{
  \begin{tabular}{|l|l|}
   \hline
   Consistency model  & Axioms                                                                                                          \\
   \hline
   Read atomic        & \textsc{Int} $\land$ \textsc{Ext} $\land$ \textsc{Session}                                                      \\
   \hline
   Causal consistency & \textsc{Int} $\land$ \textsc{Ext} $\land$ \textsc{Session} $\land$ \textsc{TransVis}                            \\
   \hline
   Prefix consistency & \textsc{Int} $\land$ \textsc{Ext} $\land$ \textsc{Session} $\land$  \textsc{Prefix}                             \\
   \hline
   Snapshot isolation & \textsc{Int} $\land$ \textsc{Ext} $\land$ \textsc{Session} $\land$  \textsc{Prefix} $\land$ \textsc{NoConflict} \\
   \hline
   Serializability    & \textsc{Int} $\land$ \textsc{Ext} $\land$ \textsc{TotalVis}                                                     \\
   \hline
  \end{tabular}
 }
 \caption{Consistency model definitions in \citet{DBLP:conf/concur/Cerone0G15}.}
 \label{weakconsistency_gotsman:2}
\end{table}

The definitions of RA, CC, PC, SI, and SER in \citet{DBLP:conf/concur/Cerone0G15} are given in Table~\ref{weakconsistency_gotsman:2}. Next, we show the equivalence between these definitions and our definitions in Figure~\ref{consistency_defs} on histories where every value is written at most once. For a history $\hist = \tup{T, \wro, \so}$ as in our framework, $\hist_{\so} = \tup{T, \so}$.

\begin{itemize}
 \item 
       We show that \textsc{Int} $\land$ \textsc{Ext} $\land$ \textsc{Session} $\equiv$ \textsc{Read Atomic}
       \begin{itemize}

        \item For a history $\hist = \tup{T, \wro, \so}$, if $\hist_{\so}$ satisfies \textsc{Int} $\land$ \textsc{Ext} $\land$ \textsc{Session} for some $\vis$ and $\co$, we show that $\hist$ satisfies \textsc{Read Atomic} for the same $\co$. If it does not, then, there exists $\tr_1, \tr_2, \tr_3$ such that $\tup{\tr_1, \tr_3} \in \wro[\xvar]$, $\tr_2$ writes on $\xvar$, $\tup{\tr_2,\tr_3} \in \wro \cup \so$ and $\tup{\tr_3, \tr_1} \in \co$. But, since $\vis$ and $\co$ satisfies \textsc{Int} $\land$ \textsc{Ext} $\land$ \textsc{Session}, $\vis \supseteq \wro \cup \so$. Hence, $\tup{\tr_2,\tr_3} \in \vis$. Therefore, $\tup{\hist, \vis, \so}$ violates the definition in figure \ref{ext_fig}, which contradicts the fact $\tup{\hist_{\so}, \vis, \so}$ satisfies \textsc{Int} $\land$ \textsc{Ext}.
              
        \item For the other direction, we have a commit order $\co$ for $\hist$ which satisfies \textsc{Read Atomic}. We show that there exists a visibility relation $\vis$ which together with the same $\co$ and $\hist_{\so}$ satisfies \textsc{Int} $\land$ \textsc{Ext} $\land$ \textsc{Session}. Let $\vis = \{ \tup{\tr_1, \tr_2} | \tup{\tr_1, \tr_2} \in \wro \cup \so \}$.
              
              \begin{itemize}
               \item First of all, by definition, the internal reads in our transactions are consistent to the last read or write before them. Only thing is left, to show that the first reads of a variable $\xvar$ before a write to $\xvar$ inside a variable is also reading from a unique transaction.
               \item If $o_1, o_1$ are two reads on $\xvar$ and $\tup{\tr_1, o_1}$, $\tup{\tr_2, o_2}$ $\in$ $\wro$, then by \textsc{Read Atomic} axiom, we have $\tup{\tr_1, \tr_2}$, $\tup{\tr_2, \tr_1}$ $\in \co$. Therefore, the reads to $\xvar$ in a transaction before the first write to $\xvar$ are from same transaction.
               \item We can not have a violation of \textsc{Int} and \textsc{Ext} because we defined $\vis$ as $\{ \tup{\tr_1, \tr_2} | \tup{\tr_1, \tr_2} \in \wro \cup \so \}$. So any violation of \textsc{Ext} will be a violation of \textsc{Read Atomic}
               \item $\vis$ satisfy \textsc{Session}, since $\so \subseteq \vis$
              \end{itemize}
              
%              Therefore, \textsc{Int} $\land$ \textsc{Ext} $\land$ \textsc{Session} $\equiv$ \textsc{Read Atomic}.
              
       \end{itemize}
 \item We show that \textsc{Int} $\land$ \textsc{Ext} $\land$ \textsc{Session} $\land$ \textsc{TransVis} $\equiv$ \textsc{Causal Consistency}.
       \begin{itemize}
        \item If $\hist$ satisfies \textsc{Causal Consistency}, there exists a $\co$ for $\hist$. We define $\vis = (\wro \cup \so)^+$.
              $\vis \supseteq \wro \cup \so$, therefore as previous case, $\vis$ satisfies \textsc{Int} $\land$ \textsc{Ext} $\land$ \textsc{Session} and by construction of $\vis$ it is a transitive closure therefore it also satisfies \textsc{TransVis}. If there is any \textsc{Int} and \textsc{Ext} violation then there exist $\tr_1, \tr_2, \tr_3$ such that $\tup{\tr_1, \tr_3} \in \wro[\xvar]$, $\tr_2$ writes on $\xvar$, $\tup{\tr_1, \tr_2} \in \co$ and $\tup{\tr_2, \tr_3} \in \vis = (\wro \cup \so)^+$ which is a violation of \textsf{Casual Consistency}. This is a contradiction. 
              
        \item If there exists a $\co$ for $\hist_{\so}$ which satisfies \textsc{Int} $\land$ \textsc{Ext} $\land$ \textsc{Session} $\land$ \textsc{TransVis}, we show, the same $\co$ satisfies \textsc{Causal Consistency}. If it does not, we have $\tr_1, \tr_2, \tr_3$ such that $\tup{\tr_1, \tr_3} \in \wro[\xvar]$ $\tr_2$ writes on $\xvar$ and $\tup{\tr_2, \tr_3} \in (\wro \cup \so)^+$ and $\tup{\tr_1, \tr_2} \in \co$. But, by \textsc{Int} $\land$ \textsc{Ext}, $\wro \subseteq \vis$ and by \textsc{Session} $\so \subseteq \vis$ and by \textsc{TransVis} we have $\vis^+ \subseteq \vis$ which implies $(\wro \cup \so)^+ \subseteq \vis$. There fore $\tup{\tr_2,tr_3} \in \vis$. So $\tr_1, \tr_2, \tr_3$ violates \textsc{Int} and \textsc{Ext} axiom for $\co$, which is a contradiction.
       \end{itemize}
%       Therefore, \textsc{Int} $\land$ \textsc{Ext} $\land$ \textsc{Session} $\land$ \textsc{TransVis} $\equiv$ \textsc{Causal Consistency}.
       
 \item We show that \textsc{Int} $\land$ \textsc{Ext} $\land$ \textsc{Session} $\land$ \textsc{Prefix} $\equiv$ \textsc{Prefix consistency}.
       \begin{itemize}
        \item If $\hist$ satisfies \textsc{Prefix consistency}, there exists a $\co$ for $\hist$. We define $\vis = \co^* \circ (\wro \cup \so)$.
              $\vis \supseteq \wro \cup \so$ therefore as previous case, $\vis$ satisfies \textsc{Int} $\land$ \textsc{Ext} $\land$ \textsc{Session}.
              Assume $\tup{\tr_1, \tr_2} \in \co \circ \vis$. 
              Then there exists $\tr_3$ such that $\tup{\tr_1, \tr_3} \in \co$ and $\tup{\tr_3, \tr_2} \in \vis = \co^* \circ (\wro \cup \so)$. Therefore, there exists $\tr_4$ such that either $\tr_3 = \tr_4$ or $\tup{\tr_3, \tr_4} \in \co$ and $\tup{\tr_4, \tr_2} \in (\wro \cup \so)$. Since, $\co$ is a total order, then $\tup{\tr_1, \tr_3} \in \co$ and either $\tr_3 = \tr_4$ or $\tup{\tr_3, \tr_4}$ imply $\tup{\tr_1, \tr_4} \in \co$. We have $\tup{\tr_4, \tr_2} \in (\wro \cup \so)$. Therefore, $\tup{\tr_1, \tr_2} \in \co \circ (\wro \cup \so) \subseteq \vis$. Therefore, $\co \circ \vis \subseteq \vis$ which is \textsc{Prefix} axiom.
              
        \item If $\hist_{\so}$ satisfies \textsc{Int} $\land$ \textsc{Ext} $\land$ \textsc{Session} $\land$ \textsc{Prefix}, then there exists a $\co$ for which the axioms satisfy. The same $\co$ will satisfy \textsc{Read Atomic} axiom for $\hist$. So if we have a violation in \textsc{Prefix consistency}, then there exist $\tr_1, \tr_2, \tr_3$ such that $\tup{\tr_1, \tr_3} \in \wro[\xvar]$, $\tr_2$ writes on $\xvar$ and $\tup{\tr_2, \tr_3} \in \co^* \circ (\wro \cup \so)$ and $\tup{\tr_1, \tr_2} \in \co$. If $\tup{\tr_2, \tr_3} \in (\wro \cup \so)$, then it is violation in \textsc{Read Atomic}, therefore, $\tup{\tr_2, \tr_3} \in \co^+ \circ (\wro \cup \so) = \co \circ (\wro \cup \so)$ because $\co$ is transitive. But by \textsc{Int} $\land$ \textsc{Ext} $\land$ \textsc{Session}, $(\wro \cup \so) \subseteq \vis$, therefore $\co \circ (\wro \cup \so) \subseteq \co \circ \vis \subseteq \vis$. Then $\tr_1, \tr_2, \tr_3$ violates \textsc{Ext}.

       \end{itemize}
       
%       Therefore, \textsc{Int} $\land$ \textsc{Ext} $\land$ \textsc{Session} $\land$ \textsc{Prefix} $\equiv$ \textsc{Causal Consistency}.
       
 \item We have to show, \textsc{Int} $\land$ \textsc{Ext} $\land$ \textsc{Session} $\land$ \textsc{Prefix} $\land$ \textsc{NoConflict} $\equiv$ \textsc{Snapshot isolation}.
       
       \begin{itemize}
        \item If $\hist$ satisfies \textsc{Snapshot isolation}, it also satisfies \textsc{Prefix consistency}. We define $\vis = (\co^* \circ (\wro \cup \so)) \cup (\co^* \circ \{\tup{\tr_1, \tr_2} | \tup{\tr_1, \tr_2} \in \co, \exists \xvar. \tr_1, \tr_2 \text{write on } \xvar \}$. Clearly, $\vis$ contains the the relations for \textsc{Int} $\land$ \textsc{Ext} $\land$ \textsc{Session} $\land$ \textsc{Prefix} proof, therefore, $\vis$ satisfies them. Also, $\vis$ satisfies \textsc{NoConflict} by definition since $\co$ is a total order. Any violation in \textsc{Int} and \textsc{Ext} will imply there is a $\tr_1, \tr_2, \tr_3$ such that $\tup{\tr_1, \tr_3} \in \wro[\xvar]$, $\tr_2$ writes on $\xvar$ and $\tup{\tr_2, \tr_3} \in \vis$ and $\tup{\tr_1, \tr_2} \in \co$. But by definition of $\vis$ we will have violations in either \textsc{Prefix} or \textsf{Conflict} axioms of \textsc{Snapshot isolation} model.
              
        \item If $\hist_{\so}$ satisfies \textsc{Int} $\land$ \textsc{Ext} $\land$ \textsc{Session} $\land$ \textsc{Prefix} $\land$ \textsc{NoConflict}, then there exists a $\co$ for which the axioms satisfy. The same $\co$ will satisfy \textsc{Prefix Consistency} axiom for $\hist$. So if we have a violation in \textsc{Snapshot isolation}, it is a violation of \textsf{Conflict} axiom, \ie there exist $\tr_1, \tr_2, \tr_3$ such that $\tup{\tr_1, \tr_3} \in \wro[\xvar]$ $\tr_2$ writes on $\xvar$ and $\tup{\tr_2, \tr_3} \in \co^* \circ \{\tup{\tr_1, \tr_2} | \exists \xvar, \tr_1, \tr_2 \text{write on } \xvar\}$ and $\tup{\tr_1, \tr_2} \in \co$. But $\{\tup{\tr_1, \tr_2} | \exists \xvar, \tr_1, \tr_2 \text{write on } \xvar\} \subseteq \vis$ by \textsc{NoConflict} and $\co \circ \vis \subseteq \vis$ by \textsc{Prefix}. Hence, $\tup{\tr_2, \tr_3} \in \vis$. Then $\tr_1, \tr_2, \tr_3$ violates \textsc{Ext}.
       \end{itemize}
       
%       Therefore, \textsc{Int} $\land$ \textsc{Ext} $\land$ \textsc{Session} $\land$ \textsc{Prefix} $\land$ \textsc{NoConflict} $\equiv$ \textsc{Snapshot isolation}.
       
 \item We show that \textsc{Int} $\land$ \textsc{Ext} $\land$ \textsc{TotalVis} $\equiv$ \textsc{Serialization}.
       
       \begin{itemize}
        \item If $\hist$ satisfies \textsc{Serialization}, there exists $\co$ for $\hist$ that satisfy \textsc{Serialization}. We define $\vis = \co$. Clearly it satisfies \textsc{TotalVis} because $\co$ is total. We have any violation in \textsc{Int} and \textsc{Ext}, that will imply we have $\tr_1, \tr_2,\tr_3$ such that  $\tr_1, \tr_2, \tr_3$ such that $\tup{\tr_1, \tr_3} \in \vis$ $\tr_2$ writes on $\xvar$ and $\tup{\tr_1, \tr_2} \in \co$. But since $\vis = \co$, $\tr_1, \tr_2, \tr_3$ will violate \textsc{Serialization} axiom, which is a contradiction.
              
        \item If $\hist_{\so}$ satisfies \textsc{Int} $\land$ \textsc{Ext} $\land$ \textsc{TotalVis}, then there exists a $\co$ for which the axioms satisfy. If for same $\co$, we have a violation of \textsc{Serialization}, then there exist $\tr_1, \tr_2, \tr_3$ such that $\tup{\tr_1, \tr_3} \in \wro[\xvar]$ $\tr_2$ writes on $\xvar$ and $\tup{\tr_2, \tr_3} \in \co$ and $\tup{\tr_1, \tr_2} \in \co$. But $\co = \vis$, so then we have a \textsc{Int} and \textsc{Ext} violation in $\hist_{\so}$ for $\tr_1, \tr_2, \tr_3$
       \end{itemize}
       
%       \textsc{Int} $\land$ \textsc{Ext} $\land$ \textsc{TotalVis} $\equiv$ \textsc{Serialization}.
\end{itemize}

\end{document}